\documentclass[sigconf, nonacm]{acmart}

\newcommand\vldbdoi{XX.XX/XXX.XX}
\newcommand\vldbpages{XXX-XXX}
\newcommand\vldbvolume{14}
\newcommand\vldbissue{1}
\newcommand\vldbyear{2020}
\newcommand\vldbauthors{\authors}
\newcommand\vldbtitle{\shorttitle} 
\newcommand\vldbavailabilityurl{URL_TO_YOUR_ARTIFACTS}
\newcommand\vldbpagestyle{plain}

\usepackage{amsmath,etoolbox}
\usepackage[thinc]{esdiff}
\usepackage{titlesec}
\usepackage{booktabs}
\usepackage{tikz}

\usepackage{fancyhdr}

\titleformat{\subsubsection}{\normalfont\normalsize\bfseries}{\thesubsubsection}{1em}{}

\makeatletter
\renewenvironment{cases}[1][\lbrace]{%
  \def\@ldelim{#1}
  \matrix@check\cases\env@cases
}{%
  \endarray\right.%
}
\patchcmd{\env@cases}{\lbrace}{\@ldelim}{}{}
\makeatother

\usepackage{booktabs}

\usepackage[lined,linesnumbered,boxed,vlined,ruled]{algorithm2e}
\usepackage{color}
\usepackage{graphicx,epsfig,subfigure}
\usepackage{xspace}
\usepackage{amsfonts,amssymb,amsmath}
\usepackage{multirow}
\usepackage{epstopdf}
\usepackage{balance}
\usepackage{bbm,enumitem}
\usepackage{tabularx}
\usepackage{multicol}

\newcommand{\keys}{\textnormal{\textsf{L}}\xspace}
\newcommand{\owt}{\textsf{OWT}\xspace}
\newcommand{\pan}{\textsf{PAN}\xspace}

\newcommand{\mono}{\texttt{MonoActive}\xspace}
\newcommand{\allkeys}{\texttt{MonoAll}\xspace}

\newcommand{\ala}{\texttt{AllAlign}\xspace}

\def\rec{\kw{rec}\xspace}
\newcommand{\sky}{\textnormal{\textsf{S}}\xspace}
\newcommand{\key}{\textnormal{\textsf{K}}\xspace}
\newcommand{\map}{\texttt{map}\xspace}

\newcommand{\miao}[1]{{{#1}}}

\newcommand{\miaopf}[1]{{{#1}}}

\newcommand{\fmax}{\ensuremath{f_{\textT}}\xspace}
\newcommand{\aset}{\textnormal{\textsf{X}}\xspace}

\newcommand{\freq}[2]{\ensuremath{f(#1, #2)}\xspace}
\newcommand{\freqq}[2]{\ensuremath{f^2(#1, #2)}\xspace}

\newcounter{subtheorem}
\counterwithin{theoremgrp}{subtheorem}

\newcommand{\wf}{\ensuremath{w}\xspace}

\newcommand{\dset}{\textnormal{\textsf{D}}\xspace}
\newcommand{\textT}{\textnormal{T}\xspace}
\newcommand{\textS}{\textnormal{S}\xspace}

\newcommand{\data}{\textnormal{T}\xspace}
\newcommand{\query}{\textnormal{Q}\xspace}
\newcommand{\datas}{\textnormal{S}\xspace}

\newcommand{\Jaccard}{\texttt{J}\xspace}
\newcommand{\func}{\Jaccard\xspace}
\newcommand{\hf}{\textit{h}\xspace}

\newcommand{\cn}{\textcolor{red}{[\raisebox{-0.2ex}{\tiny\shortstack{citation\\[-1ex]needed}}]}}

\newcommand{\stitle}[1]{\vspace{1ex}\noindent{\bf #1}}

\newcommand{\eetitle}[1]{\vspace{0.8ex}\noindent{\em\underline{#1}}}

\newcommand{\todo}[1]{\textcolor{red}{TODO: #1}\PackageWarning{TODO:}{#1!}}

\newcommand{\ie}{i.e.,\xspace}

\newcommand{\aka}{a.k.a.,\xspace}

\newcommand{\eat}[1]{\xspace}

\newcommand{\kw}[1]{{\ensuremath {\mathsf{#1}}}\xspace}

\newtheorem{theorem}{Theorem}

\newtheorem{lemma}{Lemma}
\newtheorem{example}{Example}

\newtheorem{definition}{Definition}

\begin{document}

\pagestyle{fancy}
\fancyhf{} % Clears all header and footer fields
\renewcommand{\headrulewidth}{0pt} % Removes the horizontal line in the header

\pagenumbering{gobble}
\pagestyle{plain}

% \input{response.tex}

%\clearpage
\setcounter{section}{0}
\pagenumbering{arabic}
\setcounter{page}{1}

%\title{Near-Duplicate Sequence Search at Scale for Large Language Model Memorization Evaluation}

%\title{Near-Duplicate Sequence Alignment with One Permutation Hashing}
%\title{Near-Duplicate Text Alignment under the Bag-of-Words Model}
\title{Near-Duplicate Text Alignment under Weighted Jaccard Similarity}

%% The "author" command and its associated commands are used to define the authors and their affiliations.
\author{Yuheng Zhang$^1$ \quad Miao Qiao$^2$ \quad Zhencan Peng$^1$ \quad Dong Deng$^1$}\vspace{.25em}
\authornote{Dong Deng is the corresponding author.}
\affiliation{%
  \institution{$^1$Rutgers University \quad $^2$University of Auckland}
}\vspace{.25em}
\email{{yuheng.zhang, zhencan.peng, dong.deng}@cs.rutgers.edu,
    miao.qiao@auckland.edu}

\iffalse
%% The "author" command and its associated commands are used to define the authors and their affiliations.
\author{Zhencan Peng}
\affiliation{%
  \institution{Rutgers University}
}
\email{zhencan.peng@cs.rutgers.edu}

\author{Yuheng Zhang}
\affiliation{%
  \institution{Rutgers University}
}
\email{yuheng.zhang@cs.rutgers.edu}

\author{Dong Deng*}
\affiliation{%
  \institution{Rutgers University}
  }
\email{dong.deng@cs.rutgers.edu}

\fi 

%\title{Texts Generated by Language Models and the Training Corpus: How Similar Are They?}
%\title{How Similar are Texts Generated by Language Models and the Training Data?}
%\title{The Roles of Near-Duplicate Training Sequences in Large Language Models}
%\title{DND: Text Corpus De-Near-Duplication for Large Language Model Pre-Training}
%\title{What Role Does Near-Duplicate Sequence Play in Large Language Model?}

 % and Asymptotically Unbiased 
% and Its Application to Efficient Recommendation Model Measurement

% all-range knn-graph
% knn whose search keys are in the range of $[x,y]$, denoted as $N_v^r$ where $r=[x,y]$.
%the knn graph $G_{x}^{y}$ of vectors whose search keys are within the range $[x,y]$.

\begin{abstract}

Near-duplicate text alignment is the task of identifying all subsequences (\ie substrings) in a collection of texts that are similar to a given query. Traditional approaches rely on seeding–extension–filtering heuristics, which lack accuracy guarantees and require many hard-to-tune parameters. More recent methods leverage min-hash techniques. They propose to group all the subsequences in each text by their min-hash and index the groups. When a query arrives, they can use the index to find all the min-hash sketches that are similar to the query's sketch and then return the corresponding subsequences as the results efficiently. Thus these methods guarantee to identify all subsequences whose estimated Jaccard similarity with the query exceed a user-provided threshold. However, these methods only support unweighted Jaccard similarity, which cannot capture token importance or frequency, limiting their effectiveness in real-world scenarios where tokens carry weights, such as TF-IDF.

In this paper, we address this limitation by supporting weighted Jaccard similarity using consistent weighted sampling. We design an algorithm \mono to group all subsequences in a text by their consistent weighted sampling. We analyze the complexity of our algorithm. For raw count term frequency (where a token's weight is proportional to its frequency in the text), we prove \mono generates $O(n + n\log f)$ groups (each group occupies $O(1)$ space) in expectation for a text with $n$ tokens, where $f$ is the maximum token frequency in the text. We further prove that our algorithm is optimal, meaning that any algorithm must generate $\Omega(n + n \log f)$ groups in expectation. Extensive experiments show that \mono outperforms the state-of-the-art by up to 4.7$\times$ in speed and reduces index size by up to 30\%, with superior scalability.

\end{abstract}

%The index size in expectation may even be smaller than the dataset size itself, which is $O(n\textnormal{D})$, where D is the dimensionality. This is because $\textnormal{D}\geq \log n$ often holds in high dimensional space. 

\maketitle

\pagestyle{\vldbpagestyle}
\begingroup\small\noindent\raggedright\textbf{PVLDB Reference Format:}\\
\vldbauthors. \vldbtitle. PVLDB, \vldbvolume(\vldbissue): \vldbpages, \vldbyear.\\
\href{https://doi.org/\vldbdoi}{doi:\vldbdoi}
\endgroup
\begingroup
\renewcommand\thefootnote{}\footnote{\noindent
This work is licensed under the Creative Commons BY-NC-ND 4.0 International License. Visit \url{https://creativecommons.org/licenses/by-nc-nd/4.0/} to view a copy of this license. For any use beyond those covered by this license, obtain permission by emailing \href{mailto:info@vldb.org}{info@vldb.org}. Copyright is held by the owner/author(s). Publication rights licensed to the VLDB Endowment. \\
\raggedright Proceedings of the VLDB Endowment, Vol. \vldbvolume, No. \vldbissue\ %
ISSN 2150-8097. \\
\href{https://doi.org/\vldbdoi}{doi:\vldbdoi} \\
}\addtocounter{footnote}{-1}\endgroup

\ifdefempty{\vldbavailabilityurl}{}{
\vspace{.3cm}
\begingroup\small\noindent\raggedright\textbf{PVLDB Artifact Availability:}\\
The source code, data, and/or other artifacts have been made available at \url{\vldbavailabilityurl}.
\endgroup
}

%\input{src/0-summary}

% \input{src/intro}

% \input{src_old/introduction.tex}
% \input{src_old/preliminary.tex}
% \input{src_old/framework.tex}
% \input{src_old/experiment}
% \input{src_old/related.tex}

\iffalse
\begin{table}[!t]%\vspace{-1em}
\centering
\caption{Comparison between three Min-Hash sketches}
\begin{tabular}{c|c|c|c} \hline
              & $k$-mins w/o OPH & $k$-mins w OPH & bottom-$k$ \\ \hline
sketch time  & $O(nk)$   & $O(n+k)$    & $O(n\log k)$ \\ \hline
sketch size   & $O(k)$ & $O(k)$ & $O(k)$ \\ \hline
index size    & $O(nk)$     & $O(n+k)$     & $O(nk^2)$ \\\hline
% index time   & $()$    & low    & lowest \\ \hline
% bias          & unbiased   & unbiased    & unbiased \\ \hline
\end{tabular}
\end{table}
\fi 

% \newpage
\section{Introduction}\label{sec:intro}

This paper studies the near-duplicate text alignment problem~\cite{DBLP:conf/clef/PotthastGHKMOTBGRS12,DBLP:conf/clef/PotthastHBBTRS14,allign,txtalign,DBLP:journals/pacmmod/PengZD24,llmalign}. Given a collection of data texts, the task takes a short query text and returns all subsequences (i.e., substrings) of the data texts that are similar to the query. Near-duplicate text alignment has become increasingly important in the era of large language models (LLMs), with key applications in test set leakage (also known as data contamination) detection~\cite{DBLP:conf/acl/Magar022}, training data deduplication~\cite{betterlm}, and memorization analysis~\cite{quantifymemo,DBLP:conf/emnlp/VuHHS23}. Beyond LLMs, this problem also plays a crucial role in domains such as bioinformatics~\cite{ALTSCHUL1990403}, log analysis~\cite{DBLP:conf/uss/DingZ0M23}, and plagiarism detection~\cite{DBLP:conf/clef/PotthastHGTKRSS13,DBLP:conf/clef/PotthastBESR10}.

 Due to the high computational cost of near-duplicate text alignment,  previous methods often adopt the seeding–extension–filtering heuristic~\cite{DBLP:conf/clef/PotthastBESR10,DBLP:conf/usenix/Manber94,DBLP:journals/cn/BroderGMZ97,DBLP:conf/sigmod/SchleimerWA03,DBLP:conf/www/HamidBCH09,DBLP:conf/sigmod/BrinDG95,DBLP:conf/sigir/SeoC08,DBLP:conf/www/KimCT09,DBLP:books/aw/Baeza-YatesR99,DBLP:journals/jasis/HoadZ03,DBLP:conf/sigmod/WangXQWZI16}. However, these methods lack accuracy guarantees and often involve many hard-to-tune hyperparameters~\cite{DBLP:journals/csur/FoltynekMG20}. To address these limitations, recent studies have proposed to use min-hash techniques~\cite{minhash,DBLP:conf/nips/0001OZ12} for near-duplicate text alignment. These approaches guarantee to retrieve all subsequences whose estimated Jaccard similarities with the query exceed a user-defined threshold~\cite{allign,txtalign,llmalign,DBLP:journals/pacmmod/PengZD24}. However, a key limitation is that they only support the unweighted Jaccard similarity, which treats all tokens equally—regardless of their frequency or importance (note depending on the tokenizer, a token can be a word~\cite{whitespace_splitted_tokens}, a $q$-gram~\cite{n-gram}, a byte-pair-encoding~\cite{bpe}, etc.).

To highlight the issue of unweighted Jaccard similarity, consider $\query = \text{\textit{AAAAAATTTTTTCCCCCC}}$, $\textT = \text{\textit{AAAAAATTTTTGCCCCCC}},$ and $\textS = \text{\textit{AATTGCC}.}$ Intuitively, $\query$ is much more similar to $\textT$ than to $\textS$. Yet, under unweighted Jaccard similarity (with each token as a 2-gram), both $\query$ and $\textT$, as well as $\query$ and $\textS$ yield the same similarity score: $4/7$. In contrast, weighted Jaccard similarity correctly reflects the difference: when each token in a text is weighted by its frequency in the text, $\query$ and $\textT$ have a similarity of $0.8$, while $\query$ and $\textS$ only score $0.2$. Moreover, unweighted Jaccard similarity fails to distinguish between stop words and content words. For example, consider $\query = \text{``\textit{I read about Einstein in a book}''}$, $\textT = \text{``\textit{I studied Einstein through a book}''}$, $\textS = \text{``\textit{I roamed about in a castle}''}.$
%\end{align*}

% \noindent Clearly, \query is more similar to \textT than to \textS. However, the unweighted Jaccard similarity of \query and \textT is identical to that of \query and \textS, both being 4/7 (here each token is a 2-gram). In contrast, under weighted Jaccard similarity, where each token is weighted by their frequency, the similarity of \query and \textT (which is 0.8) is significantly greater than that of \query and \textS (which is 0.2). In addition, unweighted Jaccard similarity cannot discriminate content tokens and stop tokens. For example, consider the following three texts.

%\vspace{-1.5em}
%\begin{align*}
%    \query &= \text{``\textit{I read about Einstein in a book}''}, \\
%    \textT &= \text{``\textit{I studied Einstein through a book}''}, \ \textnormal{and} \\
%    \textS &= \text{``\textit{I roamed about in a castle}''}.
%\end{align*}
%\vspace{-1.5em}

\noindent Semantically, $\query$ is clearly closer to $\textT$ than to $\textS$. However, the unweighted Jaccard similarity is the same for both: $4/9$ (with each token as a word). If a meaningful token weighting scheme such as TF-IDF~\cite{singhal2001modern} is applied, then stop words, e.g., \textit{I, in, a, about, through}, could have near-zero weight and the other words may have near-one weight, and thereby drawing clear distinctions. Under this scheme, the weighted Jaccard similarity between $\query$ and $\textT$ is around $0.5$, while that between $\query$ and $\textS$ drops to nearly zero.

\stitle{Existing Methods.} Existing methods propose to group, as the indexing process, the subsequences in each data text based on  their min-hash values. This way, when a query arrives, they use the index to identify all the min-hash sketches that are similar to the query's min-hash sketch and then return the corresponding subsequences as the results efficiently~\cite{allign,llmalign,DBLP:journals/pacmmod/PengZD24}. A key observation made by these methods is that the nearby subsequences of a text tend to share the same min-hash (this is because appending a token to a subsequence most likely would not change the min-hash of the subsequence). 
%Thus they group the subsequences by their min-hash. As a result, 
Given a text with $n$ tokens, they group the $O(n^2)$ subsequences in the text into $O(n)$ groups while represent each group with a tuple of $O(1)$ space.

\stitle{Limitations of Existing Methods.} In many real-world scenarios, tokens are associated with weights, such as TF-IDF (term frequency-inverse document frequency) weights~\cite{DBLP:conf/hicss/PasiB23}. Estimating weighted Jaccard similarity requires consistent weighted sampling (CWS)~\cite{DBLP:conf/icdm/Ioffe10}, as standard min-hash\cite{minhash} applies only to the unweighted case. Existing methods do not support CWS or weighted Jaccard similarity. The only exception is \ala~\cite{allign} which supports multi-set Jaccard, a special case of weighted Jaccard where each token's weight is its frequency in the text.  \ala is a greedy algorithm which groups the subsequences by their ``multi-set min-hash''. However, it lacks a formal complexity analysis, including an upper bound on the number of groups generated and its time complexity. Due to its recursive nature, analyzing its complexity is particularly difficult.

%

% The figure corresponds to a text $\textT$ with 10 tokens. There is a one-to-one mapping between the colored cells in the figure and the subsequences of $\textT$. The integer in the cell $(i,j)$ is the min-hash of the subsequence $\textT[i,j]$. A compact window corresponds to a rectangle where all the cells in it share the same min-hash. The figure shows an example of partition with 13 compact windows as indicated by the colors.

% To alleviate this problem, the concept of ``compact windows'' is introduced~\cite{allign,llmalign}. 

% As illustrated in Figure~\ref{fig:partition}, there is a one-to-one mapping between each colored cell $(x,y)$ and subsequence $\textT[x,y]$. The value in the cell is the min-hash  of the corresponding subsequence. Clearly, there are $O(n^2)$ min-hash values in a text with $n$ tokens. A compact window is a rectangle in the figure, which represent all the subsequences in the rectangle which share the same min-hash

% As a result, it has been shown the $O(n^2)$ min-hash values in a text with $n$ tokens can be losslessly compressed in $O(n)$ compact windows using only $O(n)$ space and $O(n\log n)$ time~\cite{allign,llmalign}. Note the generation of compact windows does not necessitate the enumeration of the $O(n^2)$ min-hash values. 

% The standard min-hash techniques~\cite{minhash} cannot accommodate weighted tokens. Instead, the weighted min-hash technique (also known as consistent weighted sampling~\cite{DBLP:conf/icdm/Ioffe10}) is required. Partition generation with weighted min-hash is beyond the capability of existing approaches.

\stitle{Our Approach.} In this paper, we propose an efficient algorithm \mono to group the subsequences of a text by their CWS or multi-set min-hash. We rigorously analyze the complexity for \mono. Specifically, for multi-set Jaccard similarity, we prove that \mono produces $O(n + n\log f)$ groups for a text with $n$ tokens in expectation, where $f$ is the maximum token frequency in the text. Each group occupies only $O(1)$ space. Furthermore, we prove that \mono is \textit{worst-case optimal} by presenting a lower bound analysis on the number of groups. That is, \miao{there exists text instances for which any algorithm under the hash-based framework must generate at least $\Omega(n + n\log f)$ groups in expectation}. We further develop an optimization that improves the time complexity of our algorithm from \miao{$O(nf\log n)$ to $O(n\log n + n \log n\log f)$} and space complexity from $O(nf)$ to $O(n + n\log f)$. Finally, we show our algorithm can be generalized to group the subsequences in a text based on their consistent weighted samplings, as long as the weight of a token in a text is monotonically increasing with its frequency in the text and is independent of other properties of the text. For example, for logarithmic term frequency (where the weight of a token $t$ with frequency $f_t$ is proportional to $\log(f_t+1)$), \mono generates $O(n+n\log \log f)$ groups in expectation.

% \stitle{Contributions.}  

In summary, we make the following contributions in this paper.

% summarize our contributions as below. 

\begin{itemize}[leftmargin=2em]
\item We develop \mono, the first algorithm for near-duplicate text alignment under weighted Jaccard similarity. % No existing method supports weighted Jaccard similarity. 

\item  We rigorously analyze the complexity of \mono and design optimizations to reduce its time and space complexities. %of \mono. 

\item For the special case of multi-set Jaccard similarity, we prove \mono is optimal, while the greedy algorithm \ala lacks theoretical guarantees. 

\item For multi-set Jaccard similarity, experimental results show that \mono outperforms \ala by up to $26\times$ in index construction time, reduces index size by up to $30\%$, and improves query latency by up to $3\times$. The performance gain increases as the text length $n$ grows, exhibiting superior scalability. % of \mono. 
\end{itemize}

% Extensive experiments demonstrate that \mono outperforms the state-of-the-art by up to $30\times$ in index time, reduces index size by up to $30\%$, and improves query latency by up to $3\times$, all while exhibiting superior scalability.

The rest of the paper is organized as follows. We introduce preliminary knowledge in Section~\ref{sec:two} and the framework in Section~\ref{sec:three}. Section~\ref{sec:four} presents our grouping algorithm and Section~\ref{sec:six} extends it to weighted Jaccard similarity. Section~\ref{sec:experiment} show experiment results, Section~\ref{sec:related} reviews related work, and Section~\ref{sec:conclude} concludes the paper.

\section{Preliminaries}\label{sec:two}

% We first consider texts with duplicate tokens where each occurrence of a token in a text is treated as distinct. ]\

We first consider a special case where the tokens are weighted by their frequencies in the text (\ie term frequency). In this case, the weighted Jaccard similarity degrades to multi-set Jaccard similarity.

\subsection{Multi-Set Jaccard Similarity}

We first define notations that will be used in the paper. A text \textT is a sequence of tokens, where $\textT[i]$ is the $i$-th token in the text. We define $\freq{t}{\textT}$ as the frequency of a token $t$ in \textT, \ie the number of occurrences of $t$ in \textT. A text \textT can be uniquely mapped to a set which exclusively contains  all the distinct tokens of \textT as elements. For ease of presentation, we refer to \textT both as a text and a set. We use $|\textT|$ to denote the length of the text $\textT$\eat{ and $||\textT||$ to represent the number of distinct tokens in $\textT$ (\ie the size of the set $\textT$)}. Furthermore, we use $\textT[i,j]$ to represent the subsequence of $\textT$ from the $i$-th token to the $j$-th token (inclusive), where $1\leq i \leq j\leq |\textT|$. Note all the definitions naturally extend to subsequences. Thus $\textT[i,j]$ is both a subsequence and a token set. Given two texts \textT and \textS, their \textit{multi-set Jaccard similarity} is 
$$\Jaccard_{\textT, \textS} = \frac{\sum_{t\in\textT\cup\textS} \min(\freq{t}{\textT}, \freq{t}{\textS})}{\sum_{t\in\textT\cup\textS} \max(\freq{t}{\textT}, \freq{t}{\textS})}.$$

\begin{example}
Consider two texts
$\textT = ABBC$ and $\textS = BCD$ where each letter denotes a token. The union $\textT\cup\textS=\{A, B, C, D\}$. The token frequencies are $\freq{A}{\textT} = 1$, $\freq{B}{\textT} = 2$, $\freq{C}{\textT} = 1$, $\freq{D}{\textT} = 0$, $\freq{A}{\textS} = 0$, $\freq{B}{\textS} = 1$, $\freq{C}{\textS} = 1$, and $\freq{D}{\textS} = 1$. Therefore, their multi-set Jaccard similarity is $\Jaccard_{\textT, \textS} = \frac{2}{5}$.
\end{example}
% Note that when there are no duplicate tokens in \textT and \textS, the above definition is equivalent to the classic Jaccard similarity.

% Next, we formally define the near-duplicate text alignment problem as below.

\subsection{Min-Hash for Multi-Set Jaccard Similarity}\label{sec:multi-set}

The multi-set Jaccard similarity of two texts can be efficiently and accurately estimated by their \textit{min-hash} sketches. Specifically, let $\hf(t,x)$ be a random universal hash function that takes a token $t$ and a positive integer $x$ as input and outputs a non-negative integer hash value. The (multi-set) min-hash of a text \textT is 
\begin{equation}\label{eq:min}
\hf(\data) = \min\{\hf(t, x) \mid t \in \data, 1 \leq x \leq \freq{t}{\data}\}.    
\end{equation}

% \noindent In other words, for each distinct token $t$ in $\textT$, we calculate $\freq{t}{\textT}$ hash values $\hf(t, 1)$, $\hf(t, 2)$, $\cdots$, $\hf(t, \freq{t}{\textT})$. \eat{In total, $|\data|$ hash values are calculated. }The min-hash $\hf(\data)$ of \data is the smallest hash value calculated. 

\begin{example}\label{exp:2}
Consider the text $\textT$ and the hash function $\hf$ from the running example in the caption of Fig.~\ref{fig:run}, we have $\hf(\textT) = \hf(B, 4) = 1$.
% where $\textT=ABABAABBCC$ and $\hf(A,1) = 2, \hf(A,2) = 5, \hf(A,3) = 8, \hf(A,4) = 12$, $\hf(B,1) = 9, \hf(B,2) = 4, \hf(B,3) = 16, \hf(B,4) = 1$, $\hf(C,1) = 3$, and $\hf(C,2) = 6$. We have $\hf(\textT)=\hf(B, 4)=1$. %This running example will be used throughout the paper.
\end{example}

%\begin{example}\label{exp:2}
%Consider a text $\textT=ABABAABBCC$ and a hash function $\hf$ where $\hf(A,1) = 2, \hf(A,2) = 5, \hf(A,3) = 8, \hf(A,4) = 12$, $\hf(B,1) = 9, \hf(B,2) = 4, \hf(B,3) = 16, \hf(B,4) = 1$, $\hf(C,1) = 3$, and $\hf(C,2) = 6$. We have $\hf(\textT)=\hf(B, 4)=1$. This running example will be used throughout the paper.
%\end{example}

Given a random universal hash function \hf and two texts $\data$ and $\datas$, the probability that they share the same min-hash is equivalent to their multi-set Jaccard similarity. Formally, we have 
$$\mathbf{Pr}(\hf(\textT)=\hf(\textS))=\Jaccard_{\textT,\textS}.$$

\noindent This is because there are $\Sigma_{t\in\textT\cup\textS}\max(\freq{t}{\textT}, \freq{t}{\textS})$ unique hash values and $\Sigma_{t\in\textT\cup\textS}\min(\freq{t}{\textT}, \freq{t}{\textS})$ common hash values in the two texts $\data$ and $\datas$. $\data$ and $\datas$ share the same min-hash if and only if the smallest  hash value of all unique hash values is among their common hash values. The probability of the latter is exactly $\Jaccard_{\textT,\textS}$.
%$$\frac{\sum_{t\in\textT\cup\textS} \min(\freq{t}{\textT}, \freq{t}{\textS})}{\sum_{t\in\textT\cup\textS} \max(\freq{t}{\textT}, \freq{t}{\textS})}=\Jaccard_{\textT,\textS}.$$

% the definition of the multi-set Jaccard similarity.

% \stitle{$k$-mins Sketch.} Let $\hf$ be a random universal hash function. The min-hash of a sequence is its smallest token hash value. A commonly used universal hash function family is $\hf(x) = (ax+b)\mod p$ where $p$ is a large prime and $a> 0$ and $b\geq 0$ are two random integers smaller than $p$~\cite{DBLP:conf/stoc/Thorup13}. We assign each token a unique ID and use the token and its ID interchangeably. The hash function $\hf(x)$ takes a token ID $x$ as input and outputs the token's hash value. \eat{It is assumed the universal hash function has no hash collision~\cite{DBLP:conf/sigmod/0004WZZQHDG21}. }We denote the min-hash of a sequence $\textT$ as $\hf(\textT)=\min(\hf(\textT[i]) \mid 1\leq i\leq |\textT|$. It has been shown that the min-hash collision probability of two sequences $\textT$ and $\textS$ is equal to their Jaccard similarity~\cite{DBLP:conf/sequences/Broder97}, \ie $$\mathbf{Pr}(\hf(\textT)=\hf(\textS))=\Jaccard_{\textT,\textS}.$$ 

\stitle{Min-Hash Sketch.} With $k$ independent random universal hash functions $\hf_1, \cdots, \hf_k$, the min-hash sketch of a text $\textT$ consists of $k$ min-hash values $\hf_1(\textT), \cdots, \hf_k(\textT)$. The multi-set Jaccard similarity of two texts $\textT$ and $\textS$ can be unbiasedly estimated by 
 \begin{equation}\label{eq:jaces}
\hat{\Jaccard}_{\textT, \textS}=\frac{1}{k}\sum_{i=1}^{k}\mathbf{1}\{\hf_i(\textT)=\hf_i(\textS)\}
\end{equation}
where $\mathbf{1}$ is an indicator function~\cite{minhash}. % This is an unbiased estimator of the Jaccard similarity with low variance~\cite{DBLP:conf/sequences/Broder97}. % The $k$-mins sketch of the sequence $\textT$ consists of the $k$ min-hash values $\hf_1(\textT), \hf_2(\textT), \cdots, \hf_k(\textT)$.

% \stitle{Extend to Integer Weights.} Extend our hash function $\hf$ to accept both a token ID and an integer, then generate multiple hashes for each item, according to its weight. If item $x$ occurs $f(x)$ times, generate hashes $\hf(x, 1)$, $\hf(x, 2)$, $\cdots$, $\hf(x, f(x))$. In these case, we denote the min-hash of a sequence $\textT$ as $\hf(T) = \min(\hf(T[i], j)) \mid 1 \leq i \leq |T|, 1 \leq j \leq f(T[i])$, which yields the multi-set weighted Jaccard similarity as the collision probability \ie

% $$\mathbf{Pr}(\hf(\textT)=\hf(\textS))=\Jaccard_{\textT,\textS} = \frac{\sum_{x\in\textT\cup\textS} \min(f_{\textT}(x), f_{\textS}(x))}{\sum_{x\in\textT\cup\textS} \max(f_{\textT}(x), f_{\textS}(x))}.$$

\stitle{Implementation Details.} The hash function $\hf(t,x)$ can be drawn from a family $\mathcal{H}$ of universal hash functions $\hf(t, x) = (a_1t + a_2 x + b) \mod p$ where $p$ is a large prime and $a_1, a_2, b$ are randomly chosen integers module $p$ with $a_1 \neq 0, a_2 \neq 0$. We assume there is no hash collision under the universal hash function throughout the paper.

\subsection{Near-Duplicate Text Alignment}

Near-duplicate text alignment is formally defined as below.
    
\begin{definition}\label{def:problem}
Given a collection of texts \dset, a query $\query$, and a similarity threshold $\theta \in [0, 1]$, the near-duplicate text alignment problem returns all the subsequences $\textT[i, j]$, where $\textT \in \dset$ and $1 \leq i \leq j \leq |\textT|$, such that $\hat{\Jaccard}_{\query, \textT[i, j]} \geq \theta$.
\end{definition}

\begin{example}
    Consider a collection of two texts $\dset = \{\textT, \textS\}$, where $\textT = ABBCDE$ and $\textS = BCCDEF$, and a query $\query =ACE$. Let the similarity threshold be $\theta = 0.5$. Suppose the similarity estimation is accurate. The near-duplicate text alignment problem returns three subsequences $\textT[1, 6], \textT[4,6],$ and $\textS[3, 5]$. Note that $\textS[2,5]=CCDE$ is not a result as $\Jaccard_{\query,\textS[2,5]}=\frac{2}{5}<\theta$.
\end{example}

% The goal In this paper, we aim to find all the subsequences whose estimated multi-set Jaccard similarity with the query are no smaller than the given  threshold (\ie replacing the $\Jaccard$ with $\hat{\Jaccard}$ in Definition~\ref{def:problem}).

% The goal of the Jaccard similarity estimator is to identify all subsequences whose estimated Jaccard similarities with the query meet or exceed a specified threshold 
\section{The Framework} \label{sec:three}

\newcommand{\cw}{\ensuremath{\langle \textT, \hf, v, a, b, c, d\rangle}\xspace}
\newcommand{\pt}{\ensuremath{\mathcal{P}}\xspace}
\newcommand{\hset}{\ensuremath{\mathsf{H}}\xspace}

%This framework has been adopted in previous works~\cn.

%To tackle the near-duplicate text alignment problem, we propose to index the min-hash sketches of all subsequences within each text in $\dset$. When a query arrives, we identify all the indexed min-hash sketches that are similar to the query's sketch. The subsequences corresponding to these similar min-hash sketches are then returned as the results. This framework has been adopted in previous works~\cn.

To improve the query performance of near-duplicate text alignment, it is natural to index the min-hash sketches of all subsequences within each text in $\dset$. This way, when a query arrives, one shall use the index to identify all the min-hash sketches that are similar to the query's sketch  and then return the corresponding subsequences as the results efficiently. 

%This framework has been adopted in previous works~\cn.

%To tackle the near-duplicate text alignment problem, we propose to index the min-hash sketches of all subsequences within each text in $\dset$. When a query arrives, we identify all the indexed min-hash sketches that are similar to the query's sketch. The subsequences corresponding to these similar min-hash sketches are then returned as the results. This framework has been adopted in previous works~\cn.

The main challenge of indexing the min-hash sketches is the cost -- the number of min-hash sketches in a text grows quadratically with the length of the text (as there are $O(n^2)$ subsequences in a text with $n$ tokens). To reduce the index size, we have a key observation: the min-hash of adjacent subsequences of a text $\textT$, such as $\textT[i,j]$ and $\textT[i,j+1]$, tend to be the same. Thus, the min-hash of adjacent subsequences can be grouped and compactly represented. For this purpose, we define the  compact window.

\subsection{Compact Window and Partition}

Below, we formally define the compact window~\cite{allign}, a structure used to represent the groups of subsequences where the subsequences in each group share the same min-hash.

\begin{definition}[Compact Window~\cite{allign}]\label{def:cw}
    Given a text $\textT$ and a random universal hash function $\hf$, a compact window in \textT is a tuple $\langle \textT, \hf, v, a, b, c, d\rangle$ satisfying that 

    \begin{itemize}%[leftmargin=*]
        \item $1 \leq a \leq b \leq c \leq d \leq |\textT|$ and

        \item $\hf(\data[i, j]) = v$ for every integers $i\in[a,b]$ and $j\in [c,d]$.
    \end{itemize}

\end{definition}

\begin{figure*}[htbp]\vspace{-5em} 
    \centering
    \subfigure[\small{All subsequences' min-hash.}]{
    \label{fig:hashvalues}
    \includegraphics[width=0.19\linewidth]{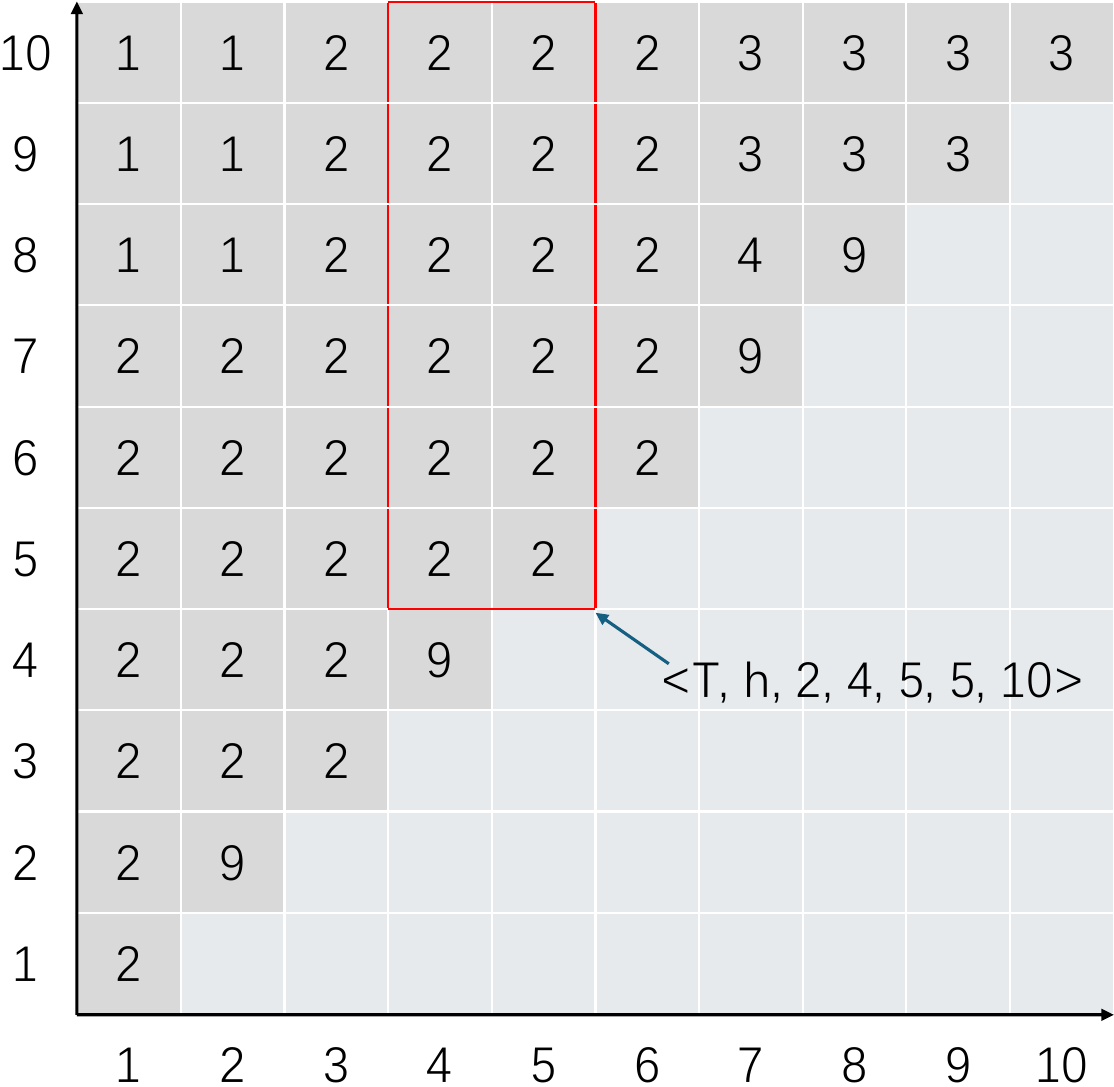}
    }
    \hspace*{-.5em}
    \subfigure[\small{An example partition.}]{
    \label{fig:partition}
    \includegraphics[width=0.19\linewidth]{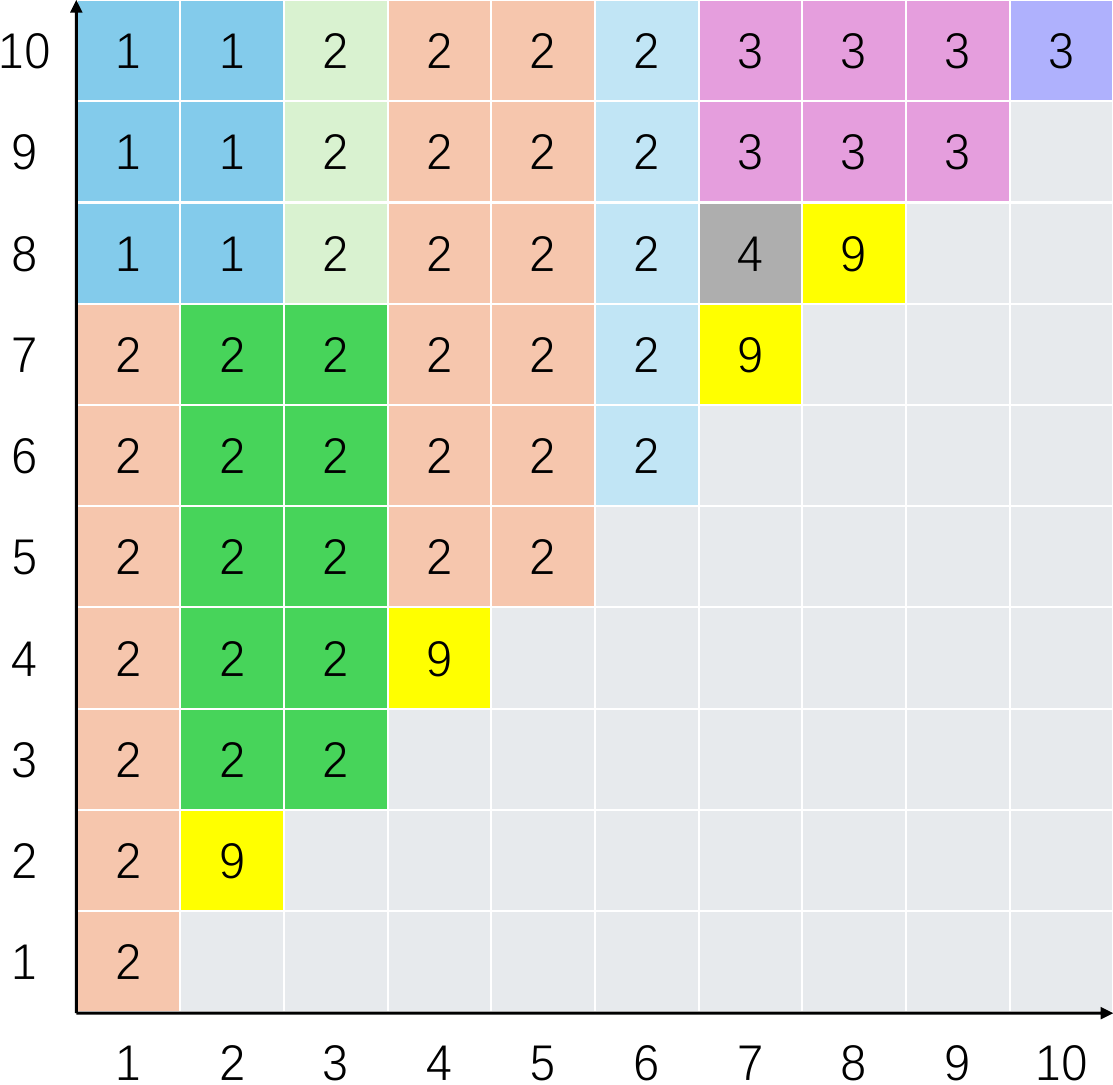}
    }
    \hspace*{-.5em}
    \subfigure[\small{Another example partition.}]{
    \label{fig:grid2}
    \includegraphics[width=0.19\linewidth]{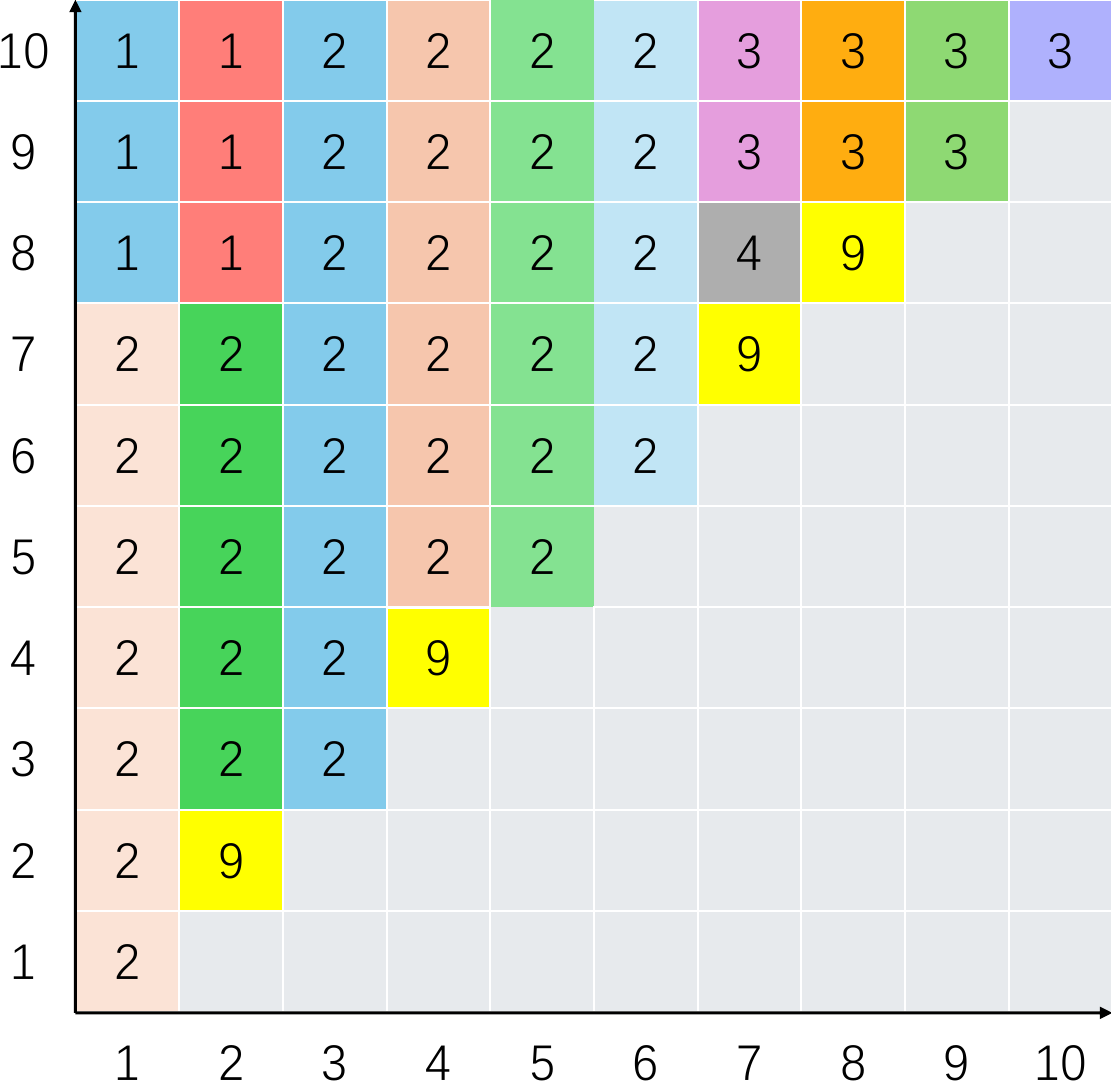}
    }
    \hspace*{-.5em}
    \subfigure[\small{All keys in $\textT$.}]{
    \label{fig:allkeys}
    \includegraphics[width=0.19\linewidth]{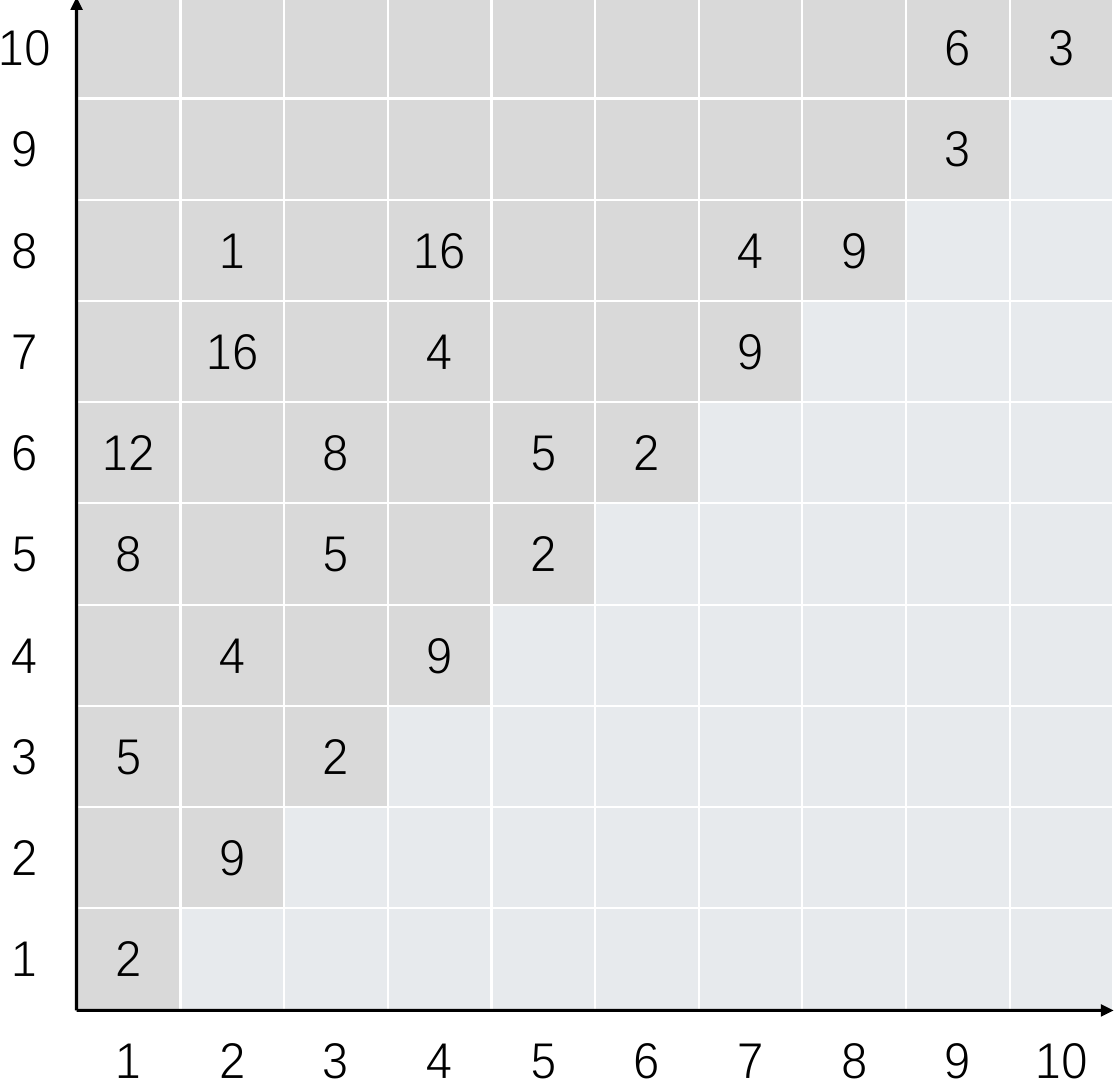}
    }
    \hspace*{-.5em}
    \subfigure[\small{All active keys in $\textT$.}]{
    \label{fig:activekeys}
    \includegraphics[width=0.19\linewidth]{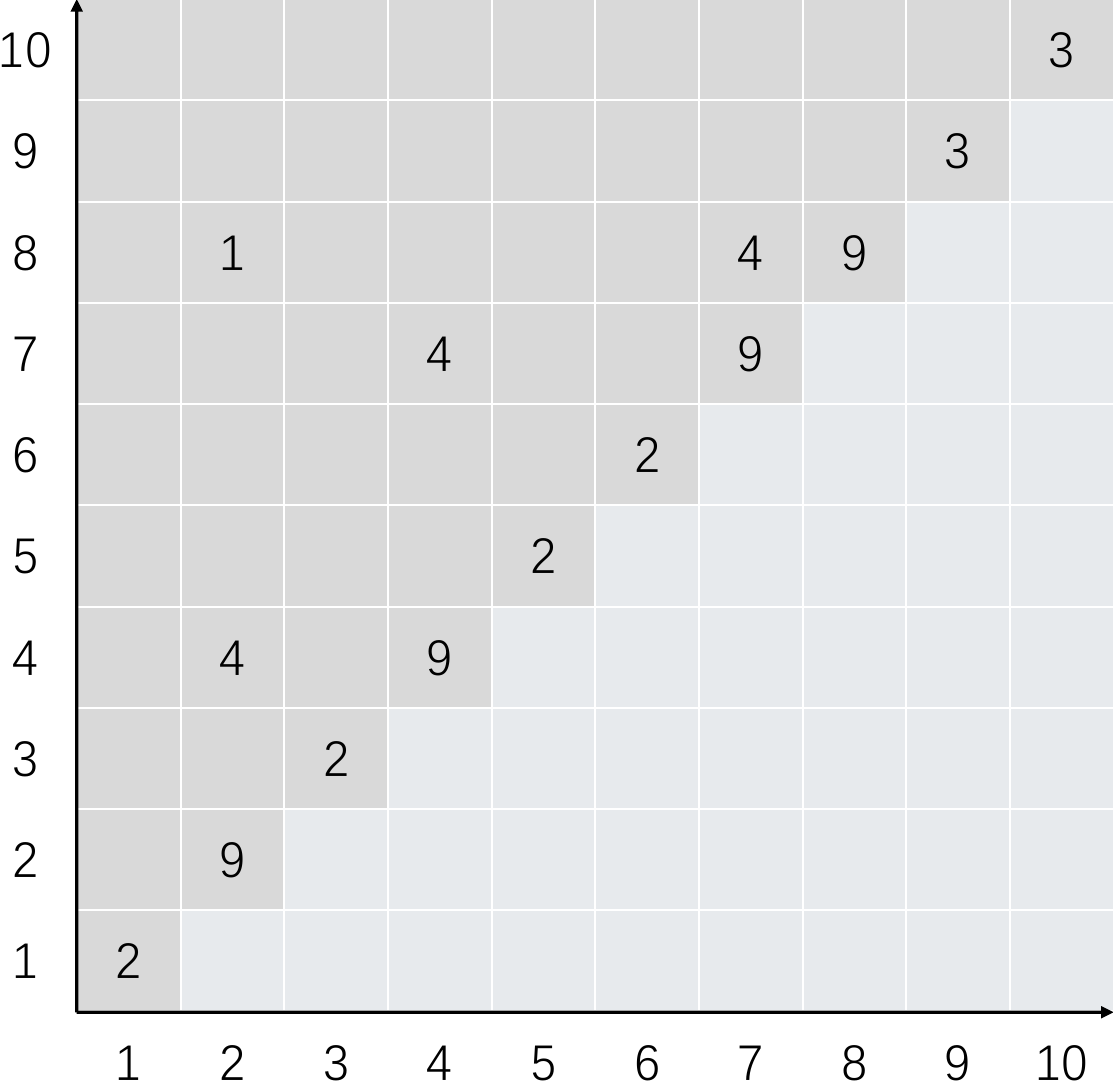}
    }
    \vspace{-1.5em} \\
    \caption{A running example used throughout the paper with a text $\textT=ABABAABBCC$ and a hash function \textnormal{$\hf$} defined as follows: \textnormal{$\hf(A,1) = 2, \hf(A,2) = 5, \hf(A,3) = 8, \hf(A,4) = 12$, $\hf(B,1) = 9, \hf(B,2) = 4, \hf(B,3) = 16, \hf(B,4) = 1$, $\hf(C,1) = 3$,} and \textnormal{$\hf(C,2) = 6$}.}\vspace{-1em} 
    \label{fig:run}
\end{figure*}

\iffalse
\begin{figure}[!ht]
    \centering
    \includegraphics[width=0.55\linewidth]{figures/grid3.pdf}
    \caption{A grid showing the min-hash of all subsequences of text $\textT=ABABAABBCC$ under hash function $\hf$ in Example~\ref{exp:2}.}
    \label{fig:hashvalues}
\end{figure}
\fi 

\begin{example}
Consider the running example in Figure~\ref{fig:run}. Figure~\ref{fig:hashvalues} shows the min-hash of all the subsequences of the text \textT. Specifically, the integer in the cell $(i,j)$ is the min-hash of the subsequence $\textT[i,j]$. For example, the min-hash of $\textT[1,10]$ is shown in the cell $(1,10)$, which is $1$. As one can verify, both  $\langle\textT,\hf,1,1,2,8,10\rangle$ and $\langle\textT,\hf,2,4,5,5,10\rangle$ (the highlighted rectangle) are compact windows. %This is because every subsequence $\textT[i,j]$ where $1\leq i\leq 2$ and $8\leq j\leq 10$ contains four occurrences of $B$, while the smallest hash value produced by $\hf$ and $\textT$ is $\hf(B,4)=1$. 
%Thus their hash value sets all contain $\hf(B, 4)=1$, which is the smallest hash value in  $\hset(\textT[i,j])$. Thus $\hf(\textT[i,j])=1$ for every $1\leq i\leq 2$ and $8\leq j\leq 10$.
\end{example}

Hereinafter, for two integers $p,q$, we abuse $[p,q]$ to denote the integer set $\{p,p+1, \cdots, q\}$. Let $[a,b]\times[c,d]$ be the set of all integer pairs $(i,j)$ where $i\in[a,b]$ and $j\in[c,d]$. A compact window \cw represents all the subsequences $\textT[i,j]$ where $(i,j)\in[a,b]\times[c,d]$. We aim to generate a set of compact windows such that each subsequence in \textT is represented by one and only one compact window in the set (\ie the set of compact windows is a lossless compression of the min-hash of all subsequences in a text). Formally, we define the concept of \textit{partition} as below.

\begin{definition}[Partition] \label{def:partition}
    Given a text $\textT$ and a hash function $\hf$, a partition $\pt(\textT, \hf)$ is
    a set of compact windows $\{W_1, W_2, \cdots, W_{l}\}$, where $l=|\pt(\textT, \hf)|$ and $W_x = \langle \textT, \hf, v_x, a_x, b_x, c_x, d_x\rangle$, that satisfies
    \begin{itemize}[leftmargin=*]
        \item Disjointness: For any two compact windows $W_x$ and $W_y$ in $\pt(\textT, \hf)$,  $$\left([a_x, b_x] \times [c_x, d_x]\right) \cap \left([a_y, b_y] \times [c_y, d_y]\right) = \phi.$$
        \item Coverage: The union of all compact windows in $\pt(\textT, \hf)$ covers all the subsequences in $\textT$, \ie $$\{(i,j)\mid 1\leq i\leq j\leq |\textT|\} \subseteq \cup_{x=1}^{l} \left([a_x, b_x] \times [c_x, d_x]\right).$$ 
        % \bigcup_{i=1}^n \left([i, i] \times [i,n]\right)$$
    \end{itemize}
\end{definition}

\iffalse
\begin{figure*}[!tbhp]
    \centering
    \subfigure[\tiny{$\hset(\textT[3, 6])$}]{
    \includegraphics[width=0.32\linewidth]{figures/grid1.pdf}
    }
    \subfigure[\tiny{$\hset(\textT)$}]{
    \includegraphics[width=0.32\linewidth]{figures/grid2.pdf}
    }
    \subfigure[\tiny{$\hset(\textT)$}]{
    \includegraphics[width=0.32\linewidth]{figures/grid2.pdf}
    }
    \caption{The hash value sets $\hset(\textT)$ and $\hset(\textT[3,6])$.}
    \label{fig:histogram}
\end{figure*}
\fi

%Each compact window represents many nearby subsequences in a text sharing the same min-hash value. 

\iffalse
\begin{figure}[!ht]
    \centering
    \includegraphics[width=0.55\linewidth]{figures/grid4.pdf}
    \caption{An example partition $\pt(\textT,\hf)$ where the text $\textT$ and the hash function $\hf$ are from Example~\ref{exp:2}.}% $|\pt(\textT,\hf)|=13$ as outlined by the colors.}
    % Given the text $\textT$ and hash function in Example~\ref{exp:2}, the figure shows an example partition of 13 compact windows as outlined by the colors.}
    %The grid of a text $\textT$ with 10 tokens, 
    %There is a one-to-one mapping between the colored cells in the grid and the subsequences of $\textT$. The value in cell $(i,j)$ is the min-hash of the subsequence $\textT[i,j]$. A compact window corresponds to a rectangle in the grid where all the cells share the same min-hash. The figure shows 
    % showing a partition with 13 compact windows as outlined by the colors.}
    \label{fig:partition}
\end{figure}
\fi 
\begin{example}
Consider the running example in Figure~\ref{fig:run}. Figures~\ref{fig:partition} and~\ref{fig:grid2} show two example  partitions with 13 and 17 compact windows respectively, as outlined by the colors. Each colored rectangle $[a,b]\times[c,d]$ in the figures corresponds to a compact window $\cw$. All the subsequences in this rectangle share the same min-hash $v$.  % In total, there are 13 and 17 compact windows in the two partitions in Figures~\ref{fig:partition} and~\ref{fig:grid2}, respectively, as outlined by the colors.%, including $\langle\textT,\hf,1,1,2,8,10\rangle$ and $\langle\textT,\hf,2,4,5,5,10\rangle$.

%Consider the text \textT and the hash function \hf from Example~\ref{exp:2}. The cell $(x,y)$ in Figure~\ref{fig:partition} represents the subsequence $\textT[x,y]$. The integer in the cell is the min-hash of the corresponding subsequence. Each colored rectangle $[a,b]\times[c,d]$ in the figure corresponds to a compact window $\cw$, representing all the subsequences in the rectangle which share the same min-hash $v$. An example compact window $\langle\textT,\hf,2,4,5,5,10\rangle$ is highlighted in the figure. In total, there are 13 compact windows in the partition.
% \todo{split this examples to 3 examples and put them earlier} \todo{Here we need a figure to explain a compact windows is a ``rectangle'' in a grid. As another perspective, the compact window is a rectangle in the 2D grid.}     \todo{You need a running example which can be used in many places hereinafter.} The figure represents all subsequences of the given sequence. The x-axis indicates the starting position of each subsequence, while the y-axis represents its ending position. The input sequence is $ABABAABBCC$. $\hf(A,i) = 2, 5, 8, 12, \cdots$, $\hf(B,i) = 9, 4, 16, 1, \cdots$, $\hf(C,i) = 3, 6, \cdots$. Each color in the figure corresponds to a compact window, representing subsequences that share the same min-hash.\todo{running example can be used here again.} 

\end{example}

%For ease of presentation, we omit the text $\textT$, the hash function $\hf$, and the min-hash value $v$ in the compact window when the context is clear. 

\subsection{Indexing and Query Processing}

Algorithm~\ref{algo:index} shows the pseudo-code of indexing in our framework. It takes a collection $\dset$ of data texts and an integer $k$ as input and produces $k$ inverted indexes of compact windows. For this purpose, it first randomly selects $k$ independent hash functions $\hf_1,\hf_2,\cdots,\hf_k$. Then, for each text $\textT$ in $\dset$, and each hash function $\hf_i$, it generates a partition $\pt(\textT,\hf_i)$. The partition generation algorithm will be described in the next section. For each compact window $\langle \textT, h_i, v, a, b, c, d\rangle$ in the partition, it is appended to the inverted list $I_i[v]$. Finally, $k$ inverted indexes $I_1, \cdots, I_k$ are returned.

\begin{figure}[!tbh]%\vspace{-1em} 
%\small
\begin{algorithm}[H]
    \caption{Indexing($\dset, k$)}
    \label{algo:index}
    \KwIn{$\dset$: a collection of data texts; $k$: an integer.}
    \KwOut{$I_1, I_2, \cdots, I_k$: $k$ inverted indexes.}
  %  \Begin{
    
 %   $\pt(\textT,\hf) \gets \emptyset$ \;
 %   $i \gets 1$ \;
    randomly select $k$ hash functions $\hf_1, \hf_2, \cdots, \hf_k$ from $\mathcal{H}$\;
    \ForEach{$\textT\in\dset$}{ \nllabel{algo:index:2}
%         $a, b, c, d \gets i$\;
%        $a \gets i, b \gets i$ \;
%        \While{$d \leq n$}{
%            $\texttt{map}[\textT[d]]\gets \texttt{map}[\textT[d]] + 1$\;
        \ForEach{$i\in[1,k]$}{\nllabel{algo:index:3}
            $\mathcal{P}\gets$\textsc{PartitionGeneration}$(\textT, \hf_i)$\; \nllabel{algo:index:4}
            \ForEach{$\langle \textT, \hf_i, v, a, b, c, d\rangle\in\mathcal{P}$}{ \nllabel{algo:index:5}
                append $\langle \textT, h_i, v, a, b, c, d\rangle$ to $I_i[v]$\; \nllabel{algo:index:6}
            }
        }
    }
    \KwRet{$I_1, I_2, \cdots, I_k$} \;\nllabel{algo:index:7}
 %   }
\end{algorithm}
\begin{algorithm}[H]
%\small
    \caption{QueryProcessing($\query, \theta, k, \hf_1, \cdots, \hf_k, I_1, \cdots, I_k$)}
    \label{algo:query}
    \KwIn{$\query$: a query text; $\theta$: a similarity threshold; $k$: an integer; $\hf_1, \cdots, \hf_k$: hash functions; $I_1, \cdots, I_k$: inverted indexes.} 
    \KwOut{All subsequences $\textT[x,y]$ in $\dset$ where $\hat{\Jaccard}_{\query,\textT[x,y]}\geq\theta$.}
   % \Begin{
    \lForEach{$i\in[1,k]$}{\nllabel{algo:query:1}
        $v_i\gets\hf_i(\query)$\; \nllabel{algo:query:2}
    }
   % $\mathcal{R}\gets$\textsc{PlaneSweep}($k, \theta, I_1[v_1], \cdots, I_k[v_k]$)  \;\nllabel{algo:query:3}
    \KwRet{\textnormal{\textsc{PlaneSweep}}($k, \theta, I_1[v_1], \cdots, I_k[v_k]$);} \nllabel{algo:query:4}
  %  }
\end{algorithm}\vspace{-.25em} 
\end{figure}

Algorithm~\ref{algo:query} presents the pseudo-code of query processing in our framework. The input consists of a query text $\query$, a multi-set Jaccard similarity threshold $\theta$, the integer $k$, the same $k$ random hash functions $\hf_1, \hf_2, \ldots, \hf_k$ used during indexing, and the $k$ inverted indexes $I_1, \cdots, I_k$. It first calculates the $k$ min-hash $v_1, \cdots, v_k$ of $\query$ and then retrieves the $k$ corresponding inverted lists $I_1[v_1], \cdots, I_k[v_k]$ from the $k$ inverted indexes. The near-duplicate subsequences can be identified by a simple plane sweep algorithm over the compact windows in the $k$ inverted lists. Specifically, the plane sweep algorithm produces all the cells $(x,y)$ where there are at least $\lceil k\theta\rceil$ compact windows $\cw$ in the $k$ inverted lists satisfying $(x,y)\in[a,b]\times[c,d]$. This is because the estimated multi-set Jaccard similarity $\hat{\Jaccard}_{\query,\textT[x,y]}\geq\theta$. The simple plane sweep algorithm is detailed in~\cite{llmalign}, while an optimized version is presented in~\cite{DBLP:journals/pacmmod/PengZD24}. Due to space constraints, we omit the details in this paper.

%, along with an optimized algorithm using segment trees, is described in detail in~\cite{DBLP:journals/pacmmod/PengZD24}. % We briefly describe the plane sweep algorithm in our appendix for completeness.

\section{Partition Generation}\label{sec:four}

Given a text \data and a hash function \hf, there exist many possible partitions. In our framework, both the indexing and query costs scale with the partition size. This raises a natural question: how can we generate a small partition, and what is the smallest possible partition? In this section, we study the partition generation problem. 
\begin{definition}[Partition Generation]
Given a text $\textT$ and a random universal hash function $\hf$, the partition generation problem is to generate a partition $\pt(\textT,\hf)$.
\end{definition}

\subsection{Monotonic Partitioning}      

In this section, we present our partition generation algorithm. To this end, we first define the hash value set of a subsequence, which contains all the hash values of the subsequence.

% In a nutshell, given a text $\textT$ and a hash function $\hf$, for each possible hash value $v$ in $\textT$, we propose to first identify all the subsequences of $\textT$ whose min-hash is exactly $v$ and then group these subsequences into compact windows. In this way, we can get a partition $\pt(\textT,\hf)$. Next, we formalize our method.

% Formally, we define the hash value set of a subsequence.
\begin{definition}[hash value set]\label{def:hset}
Given a random universal hash function $\hf$, the hash value set of a subsequence $\textT[i,j]$ is $$\hset(\textT[i,j],\hf)=\{\hf(t,x)\mid t\in\textT[i,j], 1\leq x\leq \freq{t}{\textT[i,j]}\}.$$ 
\end{definition}

% produce the compact windows $\cw$ of \textT in the ascending order of their hash value $v$. 

\noindent Based on Equation~\ref{eq:min}, the \emph{min-hash of a subsequence} is the smallest hash value in the hash value set of the subsequence, \ie
%Given a random universal hash function $\hf$, the multi-set min-hash of a subsequence $\textT[i,j]$ is 
\begin{equation}\label{eq:hfmin}
\hf(\textT[i,j])=\min \hset(\textT[i,j],\hf).    
\end{equation}

\noindent We omit the hash function $\hf$ in the hash value set for brevity. 

% We abbreviate multi-set min-hash to min-hash. For example, we have $\hf(\textT[3, 6]) = \min \hset(\textT[3, 6]) = \min \{2, 5, 8, 9\} = 2$

\begin{figure}[!tbh]\vspace{-1em}
    \centering
    \subfigure[\small{$\hset(\textT[3,6])$.}]{\label{subfig:hset1}
    \includegraphics[width=0.31\linewidth]{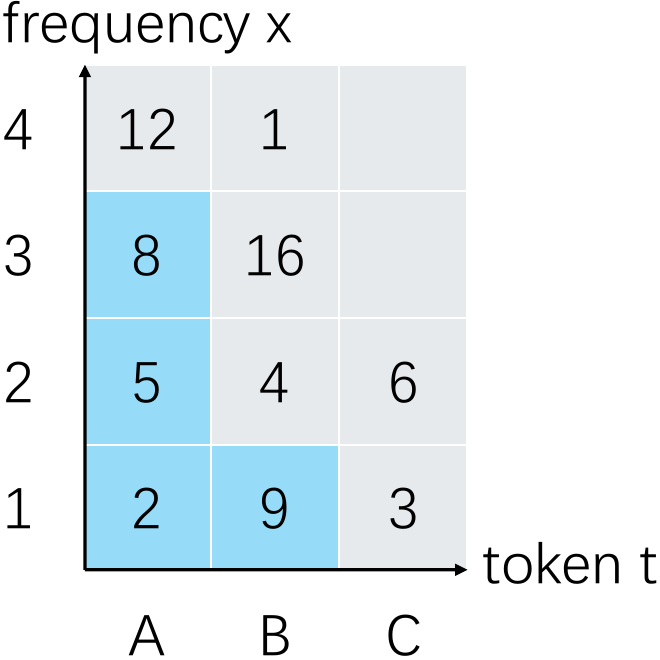}
    }
    \hspace*{-.5em}
    \subfigure[\small{$\hset(\textT)$.}]
    {\label{subfig:hset2}
    \includegraphics[width=0.31\linewidth]{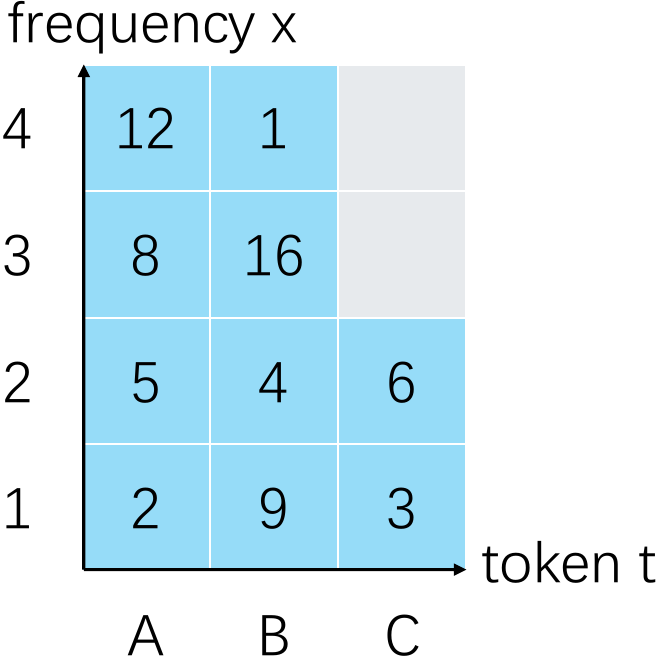}
    }
    \hspace*{-.5em}
    \subfigure[\small{active hash in} $\hset(\textT)$.]
    {\label{subfig:hset3}
    \includegraphics[width=0.31\linewidth]{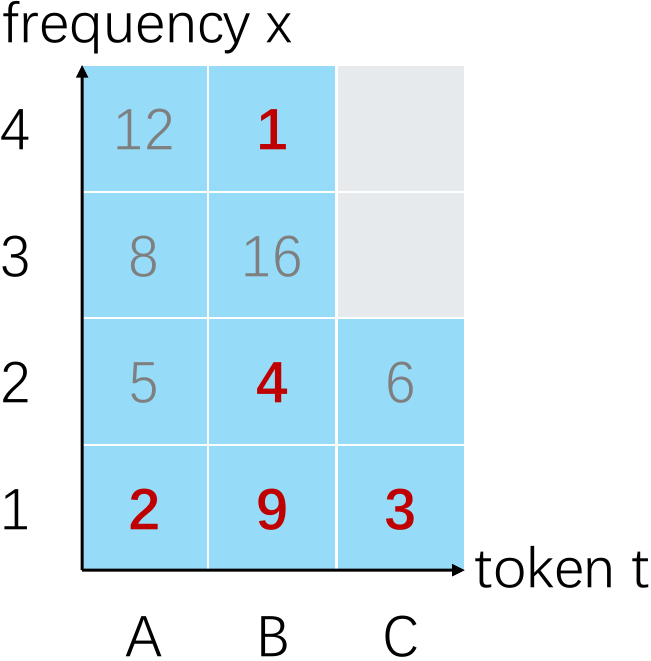}
    }
    \vspace{-1.5em} \\
    \label{fig:histogram}
    \caption{Examples of hash value sets and active hash values.}\vspace{-1em} 
\end{figure}

\begin{example}
Consider the $\textT$ and $\hf$ from the running example. Figures~\ref{subfig:hset1} and~\ref{subfig:hset2} show the hash value sets $\hset(\textT[3,6])$ and $\hset(\textT)$. In the figures, each integer in the highlighted cell $(t,x)$ represents the hash value $\hf(t,x)$ in the corresponding hash value set. Specifically, we have $\hset(\textT[3,6])=\{2, 5, 8, 9\}$. Thus $\hf(\textT[3,6])=\min\hset(\textT[3,6])=2$.

% lists the hash value set of $\textT$ where the hash value $\hf(t,x)$ in the hash value set is shown in $(t,x)$. Figure~\ref{fig:histogram} shows the hash value set of the subsequence $\textT[3,6]$. That is to say, $\hset(\textT[3,6])=\{2, 5, 8, 9\}$. Thus $\hf(\textT[3,6])=\min\hset(\textT[3,6])=2$.
%Consider a hash function $\hf$ where $\hf(A,1) = 2, \hf(A,2) = 5, \hf(A,3) = 8, \hf(A,4) = 12$, $\hf(B,1) = 9, \hf(B,2) = 4, \hf(B,3) = 16, \hf(B,4) = 1$, $\hf(C,1) = 3, \hf(C,2) = 6$. Let $\textT=ABABAABBCC$. We have $\hset(\textT[3,6])=\{2, 5, 8, 9\}$.
\end{example}

The min-hash of any subsequence of a text $\textT$ must belong to the hash value set $\hset(\textT)$ of $\textT$, as formalized below. 

\begin{lemma}\label{lemma:all}
$\hf(\textT[i,j])\in\hset(\textT)$ for any $1\leq i\leq j\leq |\textT|$.
\end{lemma}
\begin{proof}
Based on Definition~\ref{def:hset}, $\hset(\textT[i,j])\subseteq\hset(\textT)$. Based on Equation~\ref{eq:hfmin}, $\hf(\textT[i,j])$$=$$\min\hset(\textT[i,j])$. Thus $\hf(\textT[i,j])\in\hset(\textT)$.
\end{proof}

\iffalse
Based on Definition~\ref{def:hset} and Lemma~\ref{lemma:all}, the following lemma holds.
\begin{lemma}\label{lemma:contain}
For any $i\leq p\leq q\leq j$, $\hset(\textT[p,q])\subseteq\hset(\textT[i,j])$ and $\hf(\textT[p,q])\geq\hf(\textT[i,j])$.
\end{lemma}
\fi

Based on Lemma~\ref{lemma:all}, the min-hash of any subsequence of $\textT$ belongs to $\hset(\textT)$. Thus, for each hash value $v=\hf(t,x)\in\hset(\textT)$, we propose to generate a few disjoint compact windows to represent all the subsequences whose min-hash is $v$. The compact windows generated along the way must form a partition. 

Specifically, we observe that if the min-hash of a subsequence $\textT[i,j]$ comes from a token $t$ and a frequency $x$ such that $\hf(t,x) = \min\hset(\textT[i,j])$, then $x \leq \freq{t}{\textT[i,j]}$, and thus there must exist a subsequence $\textT[p,q]$ of $\textT[i,j]$ (\ie $i\leq p\leq q\leq j$) such that $\textT[p,q]$ starts and ends with token $t$ and has exactly $x$ occurrences of $t$; otherwise $\hf(t,x)\not\in\hset(\textT[i,j])$ and cannot be the min-hash of $\textT[i,j]$ based on Lemma~\ref{lemma:all}. Next, we  formalize our observation.

\begin{definition}[Key]\label{def:key}
Given a text \textT, a key $s$ is a pair of integers $(p,q)$ with $1\leq p\leq q\leq |\textT|$ such that $\textT[p]=\textT[q]$. We refer to $(p, q)$ as the coordinates of $s$, define $s.x = p$ and $s.y = q$. The hash value of $s$ is $\hf(\textT[q],\freq{\textT[q]}{\textT[p,q]})$, abbreviated as $\hf(\textT, p, q)$.
\end{definition}

\begin{definition}[Key Set]\label{def:keyset}
The key set of a subsequence $\textT[i,j]$, denoted as $\key(\textT[i,j])$ is $\{(p,q)\mid [p,q]\subseteq[i,j] \text{ and } (p,q) \text{ is a key }\}$.
\end{definition}

\iffalse
\begin{figure}[!ht]
    \centering
    \includegraphics[width=0.55\linewidth]{figures/allkeys.pdf}
    \caption{All keys in $\textT$ under  hash function $\hf$ in Example~\ref{exp:2}.}
    \label{fig:allkeys}
\end{figure}
\fi 
\begin{example}\label{exp:8}
Consider the running example in Figure~\ref{fig:run}. Figure~\ref{fig:allkeys} shows all the keys in $\key(\textT)$ and their hash values, where the integer in cell $(p, q)$ represents the hash value of key $(p, q)\in\key(\textT)$ (in contrast, Figure~\ref{fig:hashvalues} shows the min-hash of the subsequence $\textT[p, q]$ in cell $(p,q)$). In total, there are 23 keys in $\key(\textT)$. The $x$-value and $y$-value of the key $(1,3)$ are respectively 1 and 3. The hash value of $(1,3)$ is $\hf(\textT,1,3)=\hf(A,\freq{A}{\textT[1,3]})=\hf(A,2)=5$. 
\end{example}

The min-hash of a subsequence is exactly the smallest hash value of all its keys, as formalized below.
\begin{lemma}\label{lemma:keyminhash}
$\hf(\textT[i,j])=\min\{\hf(p,q,\textT)\mid (p,q)\in\textnormal{\key}(\textT[i,j])\}$.
\end{lemma}
\begin{proof}
Based on Definitions~\ref{def:hset},~\ref{def:key}, and~\ref{def:keyset}, we have $\hset(\textT[i,j])=\{\hf(p,q,\textT)\mid (p,q)\in\key(\textT[i,j])\}$. Based on Equation~\ref{eq:min}, we have $\hf(\textT[i,j])=\min\{\hf(p,q,\textT)\mid (p,q)\in\key(\textT[i,j])\}$.
\end{proof}

\begin{example}
In Figure~\ref{fig:allkeys}, the key set $\key(\textT[i,j])$ consists of all keys on the bottom-right  of the cell $(i,j)$, i.e., $(p,q)$ with $p \geq i$ and $q\leq j$. For example, $\key(\textT[1,3])=\{(1,1),(1,3),(2,2),(3,3)\}$. By Lemma~\ref{lemma:keyminhash}, the min-hash of $\textT[1,3]$ is $\hf(\textT[1,3])=\min\{5,2,9,2\}=2$. % Moreover, to derive Figure~\ref{fig:hashvalues}, we can fill each cell $(i,j)$ with the smallest integer on the bottom-right  of $(i,j)$ in Figure~\ref{fig:allkeys}, which corresponds to the min-hash of $\textT[i,j]$.
\end{example}

We say a subsequence $\textT[i,j]$ \textit{contains} a key $(p,q)$ if and only if $(p,q)\in\key(\textT[i,j])$. Lemma~\ref{lemma:iffcontain}  follows directly from Lemma~\ref{lemma:keyminhash}. 

\begin{lemma}\label{lemma:iffcontain}
A subsequence has min-hash $v$ if and only if it contains a key with hash value $v$ and no key with a smaller hash value. 
\end{lemma}

Lemma~\ref{lemma:iffcontain} motivates us to visit all keys in $\key(\textT)$ in ascending order of their hash values (breaking ties arbitrarily). When visiting a key $(p, q)$ with hash value $v$, we identify all subsequences that contain $(p, q)$ but no visited key. The min-hash of all these subsequences must be $v$ based on Lemma~\ref{lemma:iffcontain}. We then group them into disjoint compact windows. Next, we characterize these subsequences.

\begin{figure}[!tbh]\vspace{-1em} 
\vspace{-0.3cm}
    \centering
      \subfigure[Rectangles: $(p,q)$ dominates $(p',q')$]{
    \begin{tikzpicture}[scale = 0.7] \small
    % Set axis range and labels
    \draw[->, thick] (1,2) -- (5.5,2) node[right] {$x$}; % x-axis
    \draw[->, thick] (1,2) -- (1,5.5) node[above] {$y$}; % y-axis

    % Define the range of the axes
    \node at (0.7,5) {$n$};
    \node at (5,1.7) {$n$};
    \node at (1,1.7) {$1$};
    \node at (0.7,2) {$1$};
    
    % Define the points
    % \fill (1,5) circle (2pt) node[above left] {$(1,n)$};
    % \fill (1,1) circle (2pt) node[below left] {$(1,1)$};
    % \fill (5,1) circle (2pt) node[below right] {$(n,1)$};
    % \fill (6,6) circle (2pt) node[above right] {$(n+1,n+1)$};

    \coordinate (P) at (2.5,4); % Point (p',q')
    \coordinate (P2) at (4,3); % Point (p,q)
    \fill (P) circle (2pt) node[below right] {$(p',q')$};
    \fill (P2) circle (2pt) node[below right] {$(p,q)$};

    % Shade the rectangle from (1,n) to (p,q)
    \fill[blue!20,opacity=0.5] (1,5) rectangle (2.5,4);

    % Draw lines for the rectangle
    \draw[dashed, thick] (1,5) -- (4,5);
    \draw[dashed, thick] (1,3) -- (4,3);
    \draw[dashed, thick] (4,5) -- (4,3);
    \draw[dashed, thick] (1,5) -- (1,3);

    % Additional rectangle or points if needed can be added similarly
\end{tikzpicture}
    % \subfigure[Observation]{%
    %     \includegraphics[width=0.38\linewidth]{figures/key2.pdf} % Replace with your file
    %      \label{fig:key}
    % }
    % First subfigure
  %
    
        \label{fig:dominate}
    }
    % Second subfigure
    \hspace{0.1cm}
    \subfigure[Partitioning]{%
\begin{tikzpicture}[scale = 0.47] \small
    % Draw axes
    \draw[->, thick] (1, 1.5) -- (8, 1.5) node[right] {$x$};
    \draw[->, thick] (1, 1.5) -- (1, 7) node[above] {$y$};

    % Labels for x and y ranges
    \node at (7.5, 4.8) {$\sky[j]$};
    \node at (2.5, 1.8) {$\sky[i]$};

    \node at (1,0.89) {$\ $};
    % \node at (0.7,2) {$ 1$};

    % Reverse staircase structure
    \draw[thick] (1, 6) -- (7,6) -- (7,5) -- (5.5, 5) -- (5.5, 4) -- (3, 4) -- (3, 3) -- (2, 3) -- (2, 2) -- (1, 2) -- cycle;
    
    % Labels for staircase steps
    \node[above left] at (1.9, 5.1) {$j$};
    \node[above left] at (2.6, 4.1) {${j-1}$};
    \node[above left] at (2.5, 3.1) {${i+1}$};
    \node[above left] at (1.7, 2.2) {$i$};

    % Point (b,c) and dashed rectangles
    \coordinate (B) at (6.8, 2.2);
    \filldraw[red] (B) circle (2pt) node[ right] {$(b,c)$};
    \draw[dashed, red] (1, 2.2) -- (B);
    \draw[dashed, red] (6.8, 7) -- (B);
    
    % Dashed rectangles padding the staircase
   \draw[dashed] (1, 5) rectangle (5.5, 4);
   \draw[dashed] (1, 3) rectangle (2, 2);
    
    % Labels for dashed rectangles
    \node at (6.3, 4.3) {$r_{j-1}$};
    \node at (5, 3.5) {$r_{i+1}$};
    \node at (5, 2.5) {$r_{i}$};

    \draw[dashed, red] (2.2, 2.9) rectangle (B);
    \draw[dashed, red] (3.2, 3.9) rectangle (6.8,3);
    \draw[dashed, red] (5.7,4.9) rectangle (6.8,4);
    % Arrow annotations for steps
    % \draw[<->] (1, 5.2) -- (6, 5.2) node[midway, above] {$SE[j]$};
    % \draw[<->] (0.8, 2) -- (0.8, 5) node[midway, left] {$SE[j]$};

    % Additional annotations for clarity
    % \node at (6.5, 2.5) {Dashed lines represent padding rectangles};
\end{tikzpicture}
\label{fig:stair}
    }
    \vspace{-1.5em} 
    \caption{Monotonic partitioning.}
    \label{fig:combined}\vspace{-1.25em} 
\end{figure}

\begin{definition} [Rectangle of a key]\label{def:rec}
Given a key $(p,q)$ in a text $\textT$, where $n=|\textT|$, its rectangle $\rec(p,q)=[1,p]\times[q,n]=\{(i,j)\mid i\in[1,p], j\in[q,n]\}$. For a set $\keys$ of keys, $\rec(\keys)$ is the union of $\rec(p,q)$ with $(p,q) \in \keys$. 
\end{definition}

%\begin{definition} [Rectangle of a key]\label{def:rec}
%Given a key $(p,q)$, its rectangle $\rec(p,q)$ is the set of all pair in $[1,p]\times[q,n]$ where $n=|\textT|$. For a set $\keys$ of keys, $\rec(\keys)$ is the union of $\rec(p,q)$ with $(p,q) \in \keys$. 
%\end{definition}

%Consider a key $(p,q)$ with hash value $v$. $[1,p]\times[q,n]$ is the set of all subsequences containing the key. If $v$ happens to be the smallest hash value in $\hset(\textT)$, then we can derive from the key a compact window $\langle \textT, h,v,1,p,q,n\rangle$; when $v$ is not the smallest hash value in $\hset(\textT)$, the compact windows may be in a different shape. As shown in Figure~\ref{fig:dominate}, if there is a key $(p',q')$ with $p'\leq p$ and $q\leq q'$ and a hash value $v' < v$, then any subsequence in $[1,p']\times[q',n]$ will have hash value $\leq v'$ which is smaller than $v$. To find the compact windows covering all subsequences in $[1,p]\times[q,n]$ having the hash value $v$, we only need to consider the subsequences in $[1,p]\times[q,n]\setminus [1,p']\times[q',n]$. %Thus, keys with nesting coordinates partition compact windows. 

Figure~\ref{fig:dominate} shows the rectangles of two keys $(p,q)$ (outlined by dashed lines) and $(p',q')$ (the shaded rectangle). By Definitions~\ref{def:keyset} and~\ref{def:rec}, a subsequence $\textT[i,j]$ contains a key $(p,q)$ if and only if $(i,j)\in\rec(p,q)$, \ie the rectangle $\rec(p,q)$ precisely characterizes all subsequences that contain the key $(p, q)$. To exclude subsequences that contain visited keys, we maintain a ``skyline'' to track the union of the rectangles of visited keys, as formalized below.

% region $[1, p] \times [q, n]$ precisely characterizes all subsequences that contain the key $(p, q)$. To exclude subsequences that contain visited keys, we maintain a ``skyline'' structure that tracks the union of covered regions, as formalized below.

% when $(p,q)$ dominates $(p', q')$. If $(p',q')$ has a hash value $v'$ smaller than the hash value $v$ of $(p,q)$, then the compact window with hash value $v$ will not have intersection with the rectangle of $(p',q')$.  

% \begin{figure}[!tbh]
%     \centering
%     \includegraphics[width=0.75\linewidth]{figures/sky1.pdf}
%     \caption{Domination}
%     \label{fig:domiinate}
% \end{figure}

% \begin{figure}[!tbh]
%     \centering
%     \includegraphics[width=0.75\linewidth]{figures/sky2.pdf}
%     \caption{Skyline}
%     \label{fig:skyline2}
% \end{figure}

%as illustrated in Figure~\ref{fig:skyline2}, all the visited keys form a ``skyline''. 

% Consider Observation~\ref{observation:key1}. As illustrated in Figure~\ref{fig:key}, $\langle\textT,\hf,v,a,p,q,d\rangle$ must be a compact window. The reasons are as follows. Since $v=\hf(t,x)\in\hset(\textT[p,q])$, we have $\hf(\textT[p,q])\leq v$. In addition, $v=\hf(t,x)=\hf(\textT[a,d])$ as $\cw$ is a compact window. On the other hand, for any $i\in[a,p]$ and $j\in[q,d]$, it is obvious that $\hset(\textT[p,q])\subseteq\hset(\textT[i,j])\subseteq\hset(\textT[a,d])$, which leads  $\hf(\textT[p,q])\geq\hf(\textT[i,j])\geq\hf(\textT[a,d])$. Thus $\hf(\textT[i,j])=v$.  Based on Definition~\ref{def:cw}, $\langle\textT,\hf,v,a,p,q,d\rangle$ must be a compact window.

\begin{definition}[Dominance]
Consider two  keys $(p,q)$ and $(p',q')$. $(p,q)$ dominates $(p',q')$ if and only if $[p,q]\subset[p',q']$.
\end{definition}

Informally, a key dominates all keys located on its top-left side, as illustrated in Figure~\ref{fig:dominate}, regardless of their hash values. Note that a key does not dominate itself based on the definition. In Example~\ref{exp:8}, the key $(1,1)$ dominates another key $(1,3)$ as $[1,1]\subset[1,3]$. %Next, we define the skyline on a set of keys. %to enjoy an important ordering. 

\begin{definition}[Skyline]\label{def:skyline}
The skyline of a set $\keys$ of keys, denoted as $\textnormal{\sky}(\keys)$, is the subset of keys in $\keys$ that are not dominated by any key in $\keys$\textnormal{:} $\textnormal{\sky}(\keys) = \{ (p', q') \in \keys \mid \nexists\, (p, q) \in \keys \text{, } (p, q) \text{ dominates } (p',q')\}.$
%$\sky(U)$ is determined only by the coordinates of keys. 
\end{definition}

%  of hash value $v$ being visited by Algorithm~\ref{algo:skyline}, denote by $L$ the visited keys, i.e., the set of keys whose hash value is smaller than $v$.  We need to generate compact windows to cover the area in the rectangle of $\rec(p,q)$ that has not been covered by $\rec(L)$, the rectangles of keys in $L$. The corresponding intervals have min-hash exactly $v$. 

\begin{lemma}\label{lemma:exclude}
For a key set $\keys$ and its skyline $\sky(\keys)$, $\rec(\keys) = \rec(\sky(\keys))$. 
\end{lemma}

\begin{proof}
By definitions, a key $(p,q)$ dominates another key $(p',q')$ if and only if $\rec(p',q') \subset \rec(p,q)$. Furthermore, by the definition of skyline, for any key $s$ of $\keys$ not in the skyline $\sky(\keys)$, there must be another key $s'$ in the skyline that dominates $s$. Thus $\rec(s) \subset \rec(s')$. Therefore $\rec(\keys) = \rec(\sky(\keys))$.
\end{proof}

%all the subsequences in $[1,p]\times[q,n]$ while not in $[1,p']\times[q',n]$ for every key $(p',q')$ in the skyline of $L_v$. However, it is time consuming to compare $(p,q)$ with every key in the skyline of $L_v$. We observe that the following properties of a skyline can speed up this process. 

%We observe that if a visited key $(p'_1,q'_1)$ is on the top left side of another visited key $(p'_2,q'_2)$, we only need to compare the current visiting key $(p,q)$ with $(p'_2,q'_2)$ as illustrated in Figure~\cn. This is because $[1,p'_1]\times[q'_1,n]\subseteq[1,p'_2]\times[q'_2,n]$. Formally, let $[p,q]$ be the set of integers between $p$ and $q$. We define domination as below.

% However, not every subsequence in the set has the min-hash $v$. Consider a previously visited key $(p',q')$. The hash value of $(p',q')$ must be no greater than $v$. As shown in Figure~\ref{fig:domiinate}, suppose $p'\leq p$ and $q\leq q'$. Then we only need to generate compact windows to represent all the subsequences in $[1,p]\times[q,n]\setminus [1,p']\times[q',n]$. 

% When we traverse the keys in $\key(\textT)$, we can maintain a skyline of the set of visited keys. Specifically, 

When visiting a key $(p, q) \in \key(\textT)$ with hash value $v$, let $\keys$ be the set of all visited keys. Based on the discussion above, all subsequences in $\rec(p,q)\setminus\rec(\sky(\keys))$ must have min-hash $v$. As shown in Figure~\ref{fig:combined} (the red dot $(b,c)$ is $(p,q)$ and the indented black line is the skyline), these subsequences collectively form a ``staircase shape''. We generate one compact window for ``each step of the staircase'' (red dashed rectangles).

%\begin{lemma}\label{lemma:minhashv}
%All subsequences in $\rec(p,q)\setminus\rec(\sky(\keys))$ must have min-hash $v$.
%\end{lemma}

%\begin{proof}
%    Combining Lemma~\ref{lemma:iffcontain} and Lemma~\ref{lemma:exclude}.
%\end{proof}

% By Definition~\ref{def:keyset} and Lemma~\ref{lemma:contain}, any subsequence $\textT[i,j]$ with $(i,j) \in [1, p] \times [c, q]$ that does not contain any key in $\sky(\keys)$ must have min-hash $v$. As shown in Figure~\ref{fig:combined} (the red dot is $(p,q)$ and the indented black line is the skyline), these subsequences collectively form a ``staircase shape''. We generate one compact window for ``each step of the staircase'' (red dashed rectangle).

% When visiting a key $(b,c)\in\key(\textT)$ with hash value $v$, denote by $\keys$ the set of all visited keys. By Lemma~\ref{lemma:contain}, every subsequence $\textT[i,j]$, where $(i,j)\in[1,b]\times[c,n]$, that does not contain any key in $\sky(\keys)$ must have min-hash $v$. As illustrated in Figure~\ref{fig:combined}, these subsequences form a ``stairs'' shape. We generate one compact window for each step of the ``stairs'' shape to represent all these subsequences.

In summary, our partitioning algorithm works as follows. Given a text $\textT$, we traverse all keys in $\key(\textT)$ in ascending order of their hash values, maintaining a dynamic skyline of the visited keys. When visiting a key $(p, q)$ with hash value $v$, let $\keys$ be the set of visited keys. For each key in $\sky(\keys)$ that is dominated by $(p, q)$, we generate a compact window. We then update the skyline by adding $(p, q)$ if it is not dominated by any key in the skyline, and removing all keys in the skyline that are dominated by $(p, q)$. As a result, the skyline is updated to $\sky(\keys \cup \{(p, q)\})$. The following sections elaborate on these details.

% Given a text $\textT$, we visit all keys in $\key(\textT)$ in ascending order of their hash value, while maintaining a skyline of the visited keys. Specifically, when visiting a key $(b,c)$ with hash value $v$, denote by $\keys$ the set of all visited keys. For each key in $\sky(\keys)$ dominated by $(b,c)$ (if there is any), we generate a compact window  (which corresponds to a ``step of the staircase shape''). After that, we update the skyline to include $(b,c)$: the visited key $(b,c)$ is added to the skyline if it is not dominated by any key in the skyline and all the keys in the skyline dominated by $(b,c)$ are removed. In this way, the skyline is updated to $\sky(\keys \cup \{(b,c)\})$. 

\begin{figure*}[!tbh]\vspace{-5em} 
    \centering
    \subfigure[\small{Visiting $(2,8)$}]{
    \includegraphics[width=0.19\linewidth]{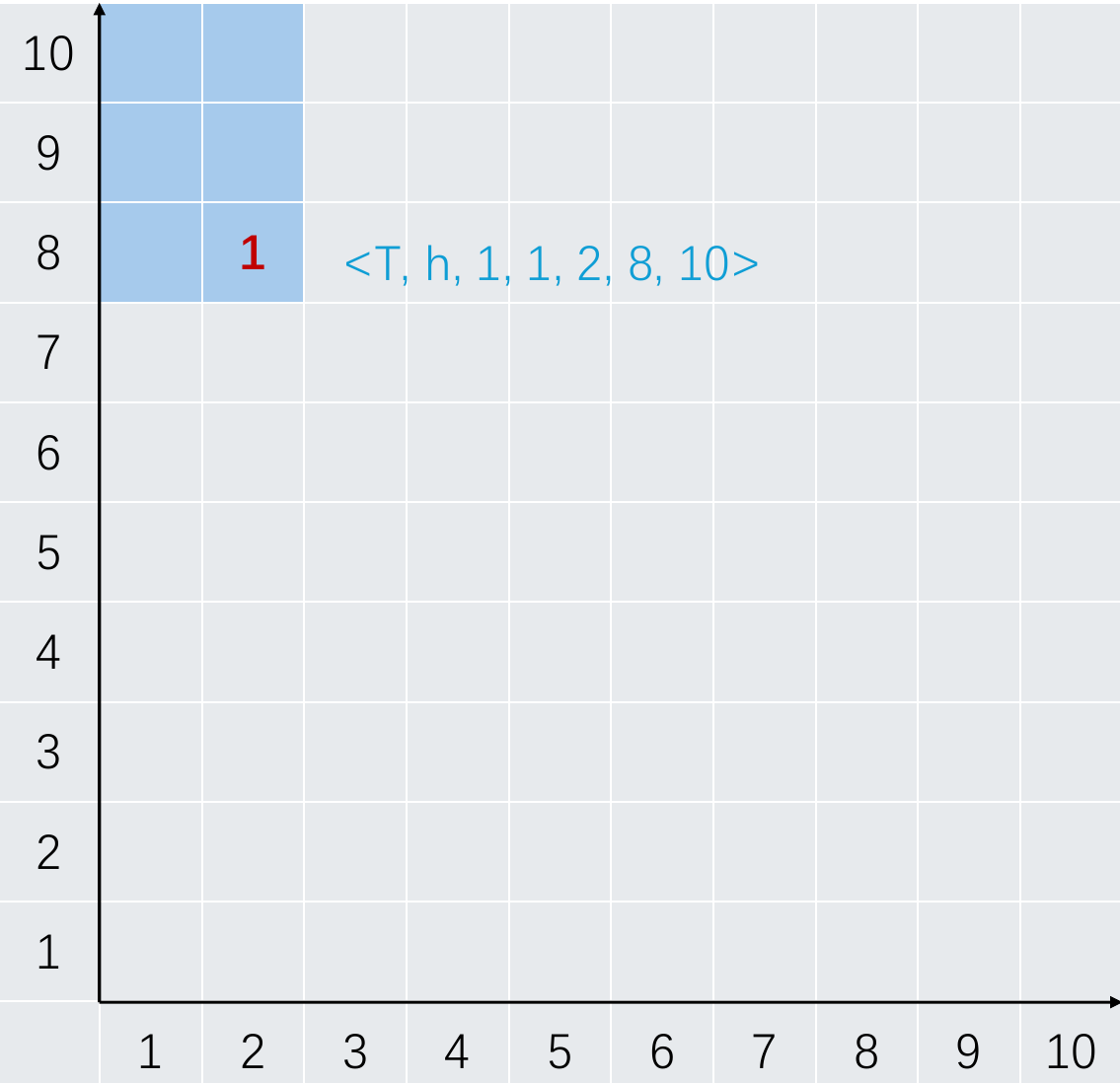}
    }
    \hspace*{-.5em}
    \subfigure[\small{Visiting $(1,1)$}]{
    \includegraphics[width=0.19\linewidth]{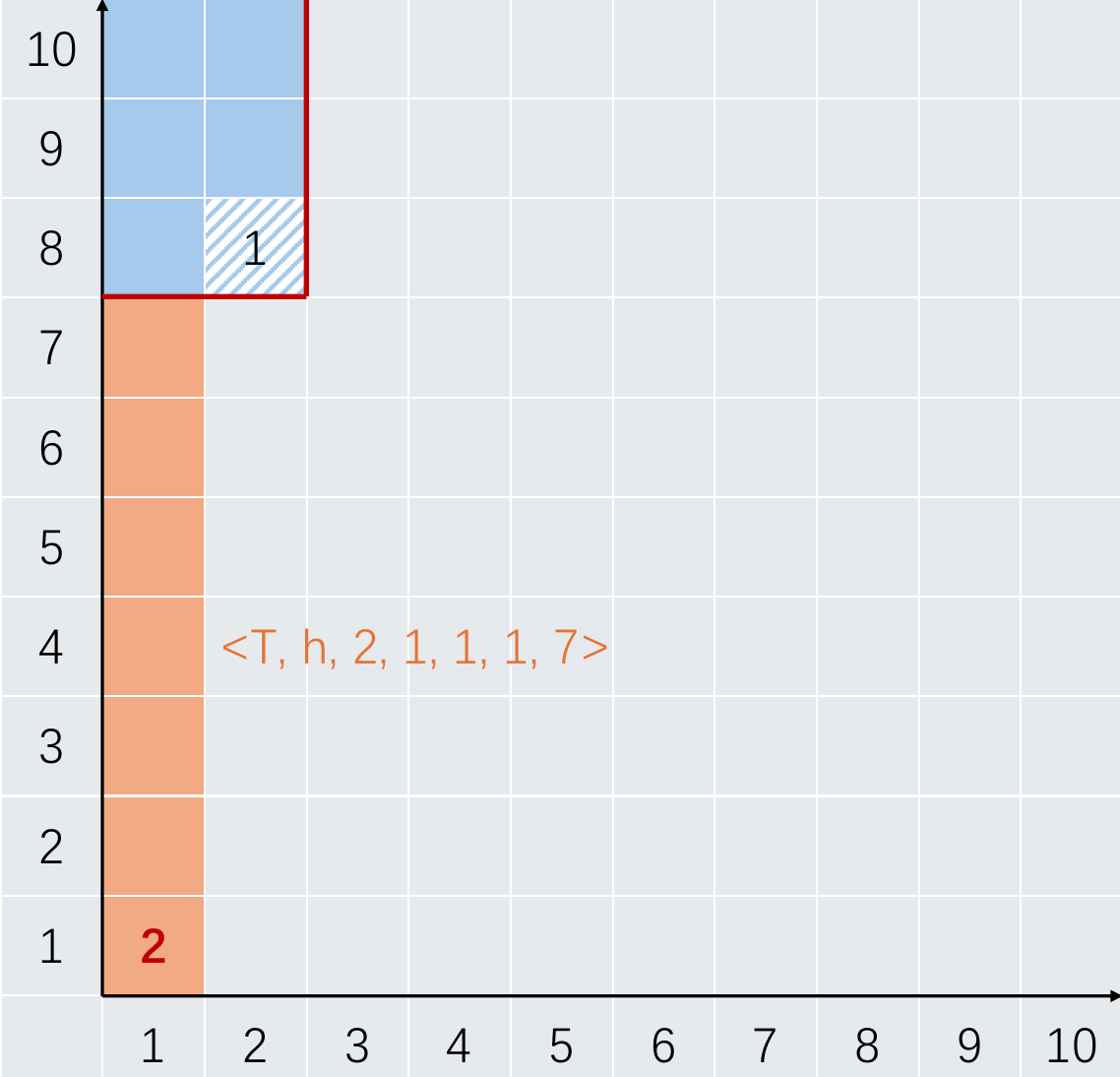}
    }
    \hspace*{-.5em}
    \subfigure[\small{Visiting}$(3,3)$]{
    \label{fig:4:c}
    \includegraphics[width=0.19\linewidth]{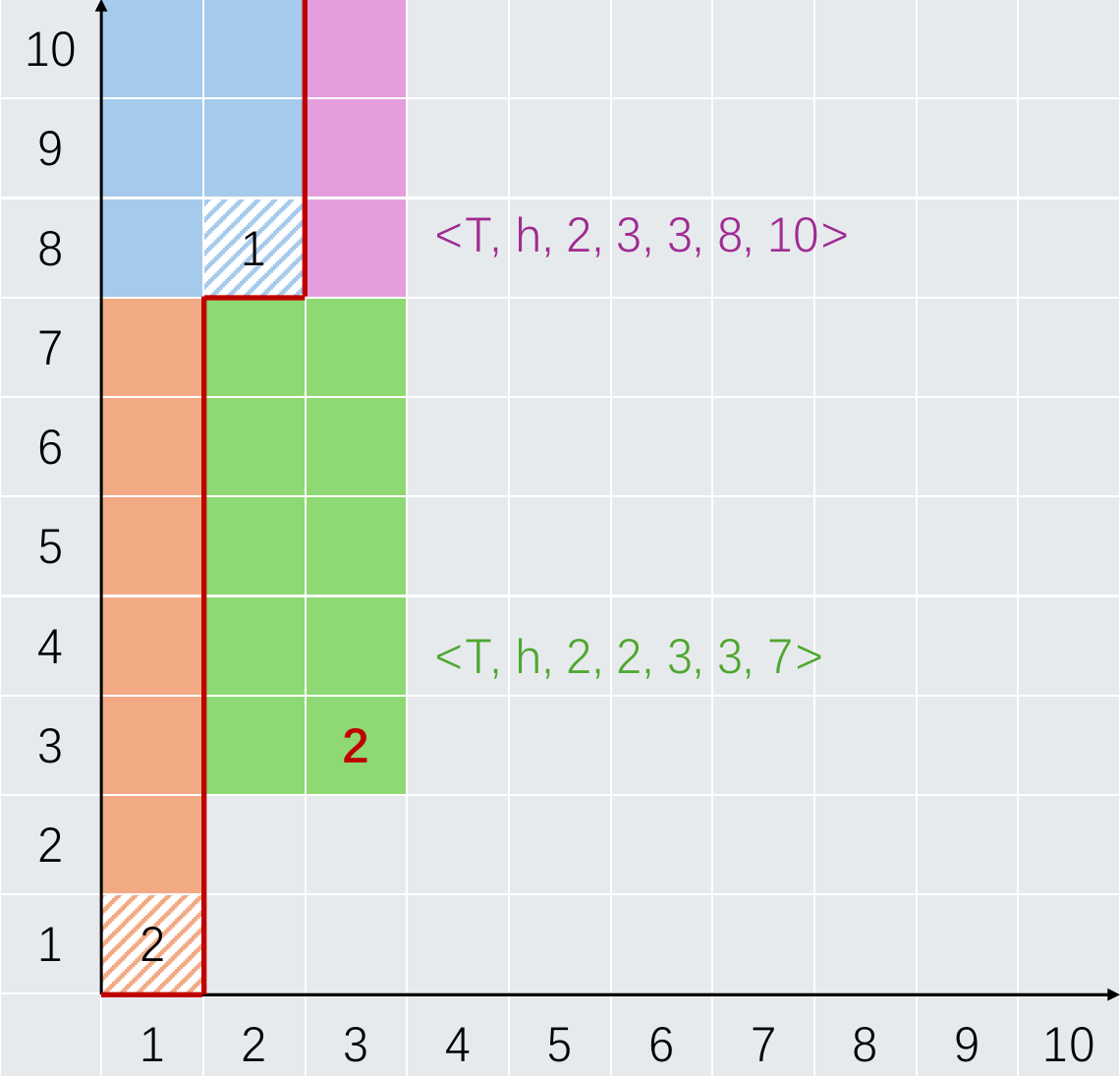}
    }
    \hspace*{-.5em}
    \subfigure[\small{Visiting}$(5,5)$]{
    \includegraphics[width=0.19\linewidth]{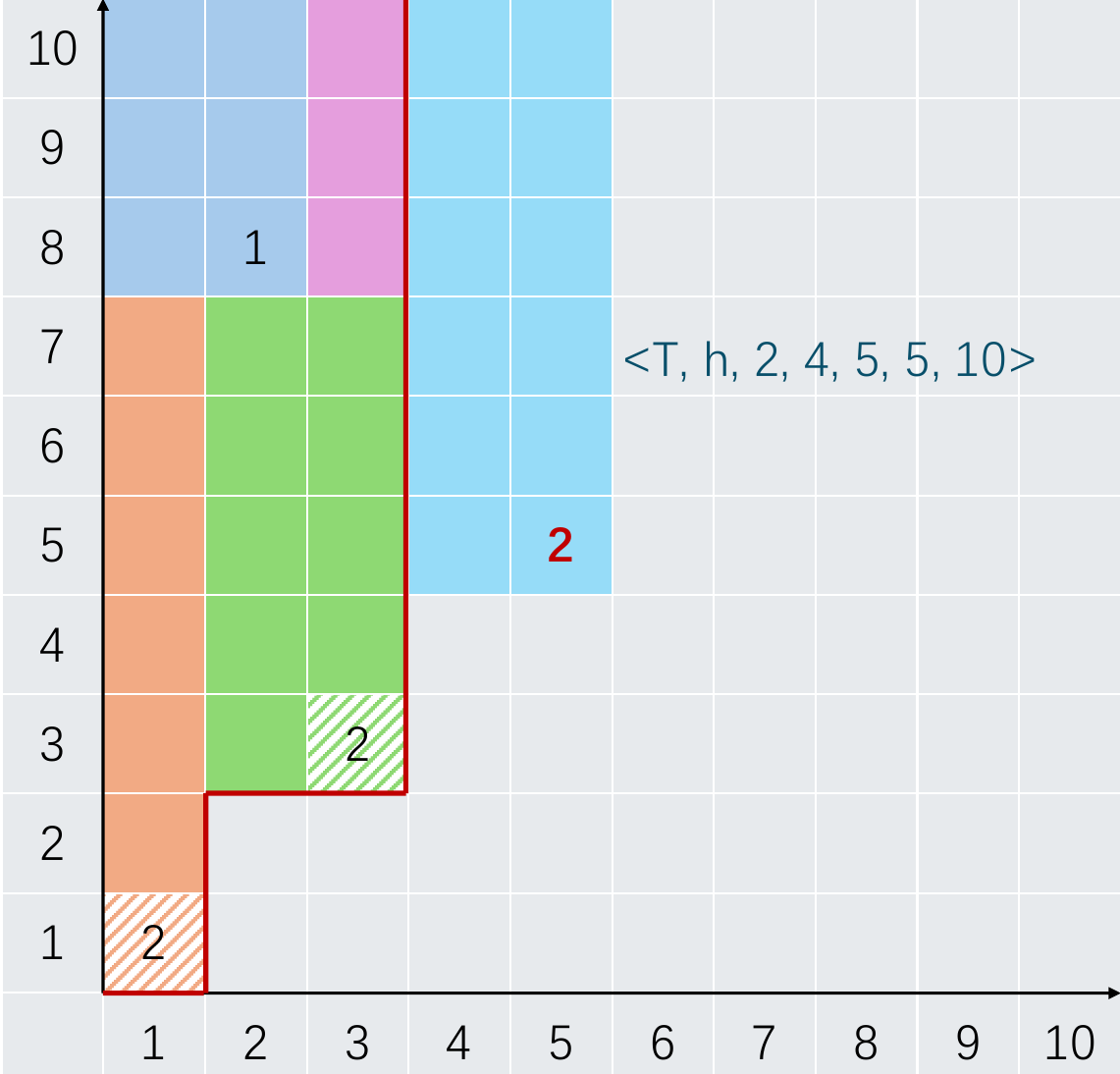}
    }
    \hspace*{-.5em}
    \subfigure[\small{Visiting}$(2,4)$]{
    \includegraphics[width=0.19\linewidth]{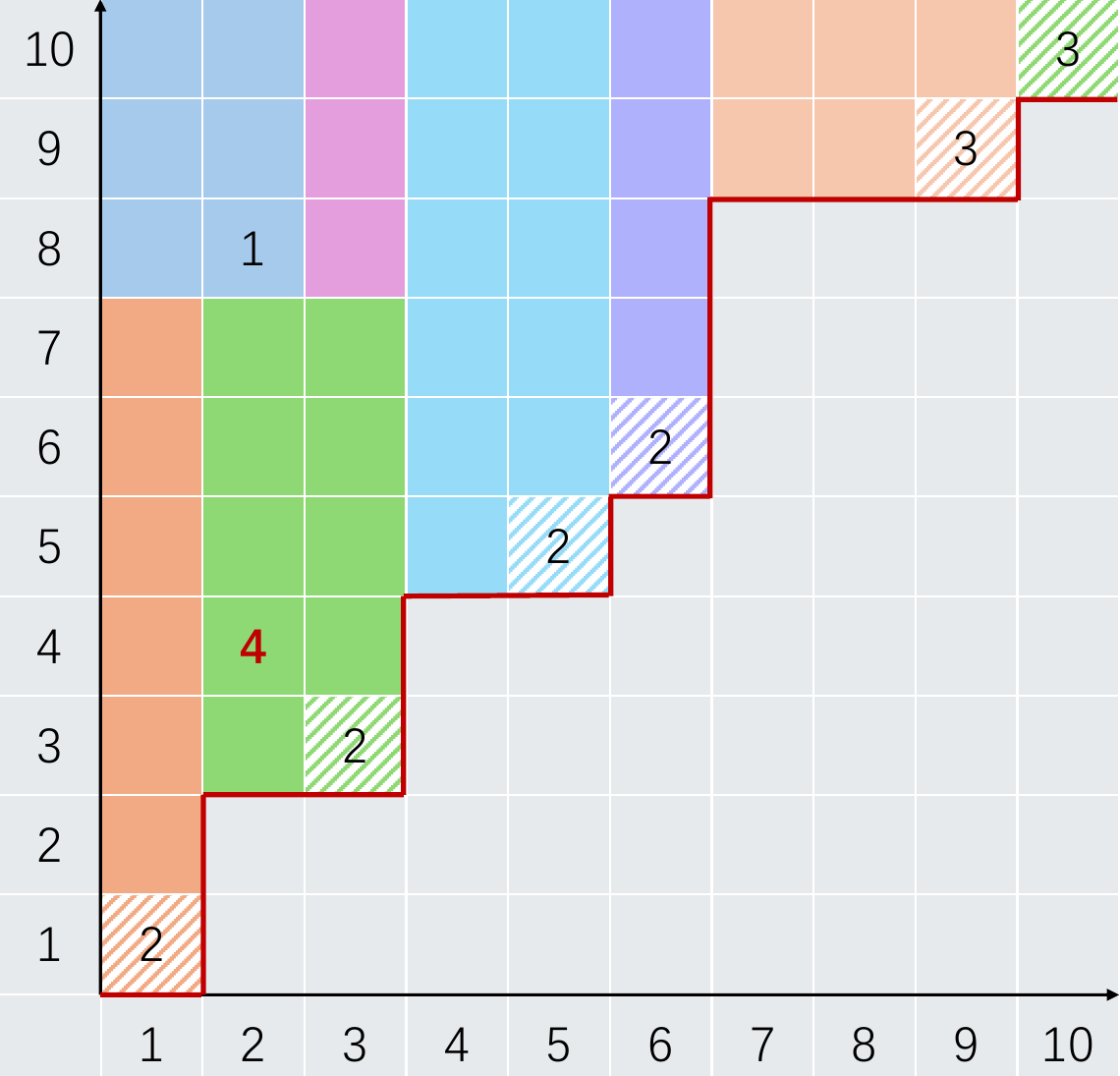}
    }
    \vspace{-1.5em} \\
    \caption{An example of monotonic partitioning.}\vspace{-1.25em} 
    \label{fig:monotone}
\end{figure*}

\begin{example}
Consider the running example. There are 23 keys in $\textT$ in total. The first 8 of them in the order of hash values are $(2,8)$, $(1,1)$, $(3,3)$, $(5,5)$, $(6,6)$, $(9,9)$, $(10,10)$, and $(2,4)$. As illustrated in Figure~\ref{fig:monotone}, we first visit $(2,8)$. The skyline is empty. A compact window $\langle\textT,\hf,1,1,2,8,10\rangle$ is generated and $(2,8)$ is added to the skyline. Then we visit $(1,1)$. A compact window is generated and $(1,1)$ is added to the skyline. Next, we visit $(3,3)$. Notice that two compact windows are generated. Since the key $(2,8)$ in the skyline is dominated by $(3,3)$, it is removed from the skyline, while $(3,3)$ is added. Later we visit $(5,5)$ and one compact window is generated. The process is repeated. When we visit the key $(2,4)$, we find $(2,4)$ is dominated by the key $(3,3)$ on the skyline. Thus no compact window will be generated and the skyline remains the same. In the end, 13 compact windows are generated as shown in Figure~\ref{fig:partition}.
%The skyline is empty. Thus we generate one compact window $\langle\textT,\hf,1,1,2,8,10\rangle$ and add $(2,8)$ to the skyline. Then
%The function \textsc{GenerateKeys} would produce 3 inverted lists $\map[A]=(1,3,5,6)$, $\map[B]=(2,4,7,8)$, $\map[C]=(9,10)$, which respectively results in 10, 10, and 3 tuples in $\key$. In total there are 23 tuples in $\key$. The first 10 tuples in the order of hash value are  The monotonic partitioning algorithm visit the 23 tuples in $\key$ in order. At the beginning, the skyline \sky contains two guard keys. 
\end{example}

\subsection{The Monotonic Partitioning Algorithm}

Before presenting the pseudo-code of our algorithm, we first discuss how to efficiently find all keys in a skyline that are dominated by a key. This can be achieved by binary search, as formalized next.

\begin{definition}[Coordinate Order]
A skyline $\sky$ is in coordinate order if its keys are ordered lexicographically by their coordinates. We denote $\sky[i]$ as the $i$-th key in $\sky$ under this order.
\end{definition}

The following lemmas show a skyline is a totally ordered set.

\begin{lemma} \label{lemma:orderinskyline1}
Let $\sky$ be a skyline of any set of keys. For any two keys $(p, q)$ and $(p', q')$ in $\sky$, if $p \leq p'$, then $q \leq q'$.
\end{lemma}
\begin{proof}
Suppose that $p \leq p'$ and $q > q'$. Then $(p, q)$ dominates $(p', q')$, since $p \leq p'$ and $q > q'$ implies $[p', q'] \subset [p, q]$. This contradicts the definition of the skyline, where no key is dominated by another. Therefore, we must have $q \leq q'$.
\end{proof}

 % the coordinate order is a total order.

\begin{lemma}\label{lemma:orderinskyline2}
Let $\sky$ be the skyline of a set of keys in coordinate order. Then $\sky[1].x \leq \cdots \leq \sky[l].x$ and $\sky[1].y \leq \cdots \leq \sky[l].y$ where $l=|\sky|$.
\end{lemma}
\begin{proof}
Consider any two consecutive keys $\sky[i]$ and $\sky[i+1]$ for $1 \leq i < l$. Since the keys in $\sky$ are sorted lexicographically by their coordinates, we have $\sky[i].x \leq \sky[i+1].x$. By Lemma~\ref{lemma:orderinskyline1}, this implies $\sky[i].y \leq \sky[i+1].y$. This completes the proof.
\end{proof}

\begin{figure}[!tbh]
%\small
\begin{algorithm}[H] %\small
    \caption{GenerateKeys($\textT, \hf$)}
    \label{algo:genkey}
    \KwIn{$\textT:$ a text of length $n$; $\hf$: a universal hash function.}
    \KwOut{$\key$: the set of keys in $\textT$.}
    $\map \gets \emptyset$, a map that maps a token to an array of integers\;  \nllabel{algo:genkey:1}
    \ForEach{$1\leq i\leq n$}{ \nllabel{algo:genkey:2}
        append $i$ to the array of token $\textT[i]$, i.e., $\map[\textT[i]]$\;  \nllabel{algo:genkey:3}
    }
    \ForEach{token $t$ in $\map$ (denote by $A$ the array of $\map[t])$}{ \nllabel{algo:genkey:4}
        \ForEach{$i \in [1,|A|]$}{ \nllabel{algo:genkey:5}
            \ForEach{$j \in [i,|A|]$}{ \nllabel{algo:genkey:6}   
                add $(A[i], A[j], \hf(t,j-i+1))$ to $\key$\; \nllabel{algo:genkey:7}
            }
        }
    % \KwRet{$V$} \;
    }
    \KwRet{\textnormal{\key};}\nllabel{algo:genkey:8}
\end{algorithm}
\begin{algorithm}[H] %\small
    \caption{MonotonicPartitioning($\textT,\hf$)}
    \label{algo:skyline}
    \KwIn{$\textT$: a text of length $n$; $\hf$: a universal hash function.}
    \KwOut{A partition $\pt(\textT,\hf)$.}
    % \Begin{
    $\key\gets$ \textsc{GenerateKeys}($\hf, \textT)$\; \nllabel{algo:skyline:1}  
    % \tcp{$(0,0)$ and $(n+1, n+1)$ are two guard keys;}
    $\sky\gets\{(0,0),(n+1,n+1)\}$, a skyline kept in coordinate order\; \nllabel{algo:skyline:2} 
    \ForEach{$(b,c,v) \in \key$ in ascending order of the hash value $v$}{\nllabel{algo:skyline:3} 
     %   The hash value of key $(b,c)$ is denoted as $v\gets\hf(\textT,b,c)$\; \nllabel{algo:skyline:4} 
        % %\tcp{With the coordinate order, a binary search can find the key in $\sky$ that dominates $(b,c)$;}
   %     \tcp{Check if there is a non-guard key in $\sky$ that dominates $(b,c)$ via Lemma~\ref{lemma:dom};}
        $j'\gets$ binary search the largest index such that $\sky[j'].y\leq c$\; \nllabel{algo:skyline:4}
        \lIf{$\sky[j']$ dominates $(b,c)$}{\textbf{continue}}\;  \nllabel{algo:skyline:5} 
     %   \tcp{generate compact windows with hash value $v$ based on Theorem~\ref{theorem:domed}(3);} 
        $i\gets$ binary search the largest index such that $\sky[i].y < c$\; \nllabel{algo:skyline:6} 
        $j\gets$ binary search the smallest index such that $\sky[j].x > b$\; \nllabel{algo:skyline:7} 
   %         \tcp{generate $j-i$ compact windows to exclusively cover all intervals in $[1,b]\times[c,n]$ that have not been covered before, based on Theorem~\ref{theorem:domed};}
            $c'\gets c$\; \nllabel{algo:skyline:8} 
            \ForEach{$i\leq k < j$}  { \nllabel{algo:skyline:9}
                $a \gets \sky[k].x ;\ \ d \gets \sky[k+1].y - 1$ \; \nllabel{algo:skyline:10} 
                \tcp{this is to avoid the case when $\sky[i+1].y=c$}
                \If{$a\leq b$ \text{and} $c' \leq d$}{ \nllabel{algo:skyline:11} 
                add $\langle\textT, \hf, v, a, b, c', d\rangle$ to $\pt(\textT,\hf)$\;\nllabel{algo:skyline:12} 
                }
                $c' \gets \sky[k+1].y$\;  \nllabel{algo:skyline:13} 
            }     
             remove $\sky[i+1], \sky[i+2], \cdots, \sky[j-1]$ from $\sky$\; \nllabel{algo:skyline:14} 
           % remove from $\sky$ the keys dominated by $(b,c)$, which are $\sky[i+1], \sky[i+2], \cdots, \sky[j-1]$ based on Theorem~\ref{theorem:domed} (1)\; \nllabel{algo:skyline:14} 
            add $(b,c)$ to $\sky$\; \nllabel{algo:skyline:15} 
        }
        % } 
        
    \KwRet{$\pt(\textT,\hf)$;} \nllabel{algo:skyline:16} 
\end{algorithm}
\begin{algorithm}[H]
    \caption{GenerateActiveKeys($\textT, \hf$)}
    \label{algo:genactivekeys}
   %  \KwIn{$\hf$: a hash function, $\textT:$ a text.}
   %  \KwOut{$\key$: the set of active keys in $\textT$.}
   % \ForEach{$1\leq i\leq n$}{ \nllabel{algo:genkey:2}
   %      append $i$ to the array of token $\textT[i]$, i.e., $\map[\textT[i]]$\;  \nllabel{algo:genkey:3}
   %  }
   \tcp{Replace Lines~\ref{algo:genkey:4}-\ref{algo:genkey:7}, Algorithm~\ref{algo:genkey} with:}
    \ForEach{entry $(t,A)\in\map$}{
        $minkey\gets +\infty$\;
        \ForEach{$1\leq x \leq |A|$}{
            \If{$minkey>\hf(t,x)$}{
                $minkey\gets\hf(t,x)$\;
                \ForEach{$1 \leq i \leq |A|-x+1$}{
                    add $(A[i], A[i+x-1], \hf(t,x))$ to $\key$\;
                }
            }
        }
    }
    % \KwRet{$\key$;}
    % \KwRet{$V$} \;
   
\end{algorithm}\vspace{-1em}
\end{figure}

Lemma~\ref{lemma:dom} shows that by a binary search, one can determine if a key is dominated by any key in a skyline under coordinate order. 
\begin{lemma} \label{lemma:dom}
Given a skyline $\sky$ in coordinate order and a key $(b,c)$, let $j$ be the largest index such that $\sky[j].y \leq c$. There exists a key in $\sky$ that dominates $(b,c)$ if and only if $j$ exists and  $\sky[j]$ dominates $(b,c)$.
\end{lemma}

\begin{proof}
Suppose there exists a key $\sky[i]$ in $\sky$ that dominates $(b,c)$, \ie $b \leq \sky[i].x$ and $\sky[i].y \leq c$. Then $j$ exists by the definition of $j$. Since $j$ is the largest index with $\sky[j].y \leq c$, we have $\sky[i].y \leq \sky[j].y$, which by Lemma~\ref{lemma:orderinskyline2}, implies $\sky[i].x \leq \sky[j].x$. Thus, $b\leq \sky[i].x\leq \sky[j].x$ and $\sky[j].y \leq c$. As $\sky[i]$ dominates $(b,c)$, $\sky[j]$ dominates $(b,c)$. On the other hand, if $j$ exists and $\sky[j]$ dominates $(b, c)$, then clearly $\sky$ contains a key (namely $\sky[j]$) that dominates $(b, c)$.
%Note that if $(b, c) \in \sky$, then $\sky[j] = (b, c)$ which does not dominate $(b,c)$ by the definition of skyline, no key in $\sky$ dominates $(b, c)$.
\end{proof}

Algorithm~\ref{algo:skyline} shows the pseudo-code our monotonic partitioning algorithm. It takes a text $\textT$ and a hash function $\hf$ as input and produces a partition $\pt(\textT,\hf)$. It first generates all the keys in the text using the procedure \textsc{GenerateKeys} (Line~\ref{algo:skyline:1}). \textsc{GenerateKeys} (Algorithm~\ref{algo:genkey}) first builds an inverted index to keep track of the positions of each distinct token in $\textT$ (Lines~\ref{algo:genkey:1}-\ref{algo:genkey:3}). For each inverted list $(t,A)$ (Line~\ref{algo:genkey:4}), it enumerates every pair of positions $A[i]$ and $A[j]$ where $i\leq j$ (Lines~\ref{algo:genkey:5}-\ref{algo:genkey:6}). Then $(A[i], A[j])$ must be a key. The key, along with its hash value $h(t,j-i+1)$, is added to the key set $\key$ (Line~\ref{algo:genkey:7}). Finally, the procedure returns the set $\key$ of all keys in $\textT$, along with their hash values  (Line~\ref{algo:genkey:8}).

Once the key set $\key$ is generated, Algorithm~\ref{algo:skyline} initializes the skyline with two guard keys $(0,0)$ and $(n+1,n+1)$ in Line~\ref{algo:skyline:2}. The guard keys help simplify the algorithm by avoiding a few corner cases. The keys in $\sky$ will always be kept in coordinate order using a self-balanced binary search tree. After that, we visit the keys in $\key$ in ascending order of their hash values; let $(b, c)$ be the key currently being visited and $v$ be its hash value (Line~\ref{algo:skyline:3}),
%for the key being $(b,c)$ being visited with hash value denoted as $v$  
denote by $\keys$ the set of keys that have already been visited. $\sky$ maintains the skyline on $\keys$ upon the insertion of $(b,c)$. Specifically, if $(b,c)$ has been dominated by any key in $\sky$, then $\rec(b,c)\setminus\rec(\keys)$ is empty and no subsequence in $\rec(b,c)$ would have min-hash $v$. It means all subsequences in $\rec(b,c)$ have been represented by compact windows generated earlier. In this case, we generate no compact window: skip $(b,c)$ and jump to the next key (Line~\ref{algo:skyline:5}). 

Based on Lemma~\ref{lemma:dom}, we can efficiently determine whether $(b, c)$ is dominated by any key in $\sky$ via binary search. Note that if the result is $j' = 0$  (Line~\ref{algo:skyline:4}), then $\sky[j']$ refers to the guard key $(0, 0)$, which does not dominate $(b, c)$. This case corresponds to the scenario in Lemma~\ref{lemma:dom} where no such index $j$ exists.

When $(b,c)$ is not dominated by any key in $\sky$, we first find all the keys in $\sky$ that are dominated by $(b,c)$. Lemma~\ref{lemma:dom2} (\kw{C1}) shall prove that with two binary searches on $\sky$ (Lines~\ref{algo:skyline:6}-\ref{algo:skyline:7}), $\sky[i+1]$ to $\sky[j-1]$ are exactly the keys in $\sky$ dominated by $(b,c)$. Lines~\ref{algo:skyline:8}-\ref{algo:skyline:13} generate up to $j-i$ compact windows (red dashed rectangles in Figure~\ref{fig:stair}). Lemma~\ref{lemma:dom2} (\kw{C2}) shall prove that these compact windows are disjoint and cover all the subsequences $\textT[i,j]$ where $(i,j)\in\rec(b,c)\setminus\rec(\keys)$, whose min-hash must be $v$ (by Lemma~\ref{lemma:iffcontain} and Lemma~\ref{lemma:exclude}). We update $\sky$ by removing these dominated keys (Line~\ref{algo:skyline:14}) and inserting the new key $(b,c)$ (Line~\ref{algo:skyline:15}). 

In practice, Line~\ref{algo:skyline:4} and Line~\ref{algo:skyline:6} can be combined in one binary search. We use two binary searches for ease of presentation later.

\begin{example}
Consider the running example. As shown in Figure~\ref{fig:4:c}, when visiting $(b=3,c=3)$, we have $v=2$ and $\sky=\{(0,0),(1,1),(2,8),(11,11)\}$. The binary search would find $j'=2$ as $\sky[2].y=1\leq c$ and $\sky[3].y=8>c$. Clearly, $\sky[j']=(1,1)$ does not dominate $(3,3)$. Then the next two binary searches would find $i=2$ and $j=4$. Next, when $k=2$, a compact window $\langle\textT,\hf,v=2,a=2,b=3,c'=3,d=7\rangle$ is generated and $c'$ becomes 8. When $k=3$, another compact window  $\langle\textT,\hf,2,3,3,8,10\rangle$ is produced. Finally, the key $\sky[3]=(2,8)$ is removed from the skyline, while $(3,3)$ is added.
\end{example}

% \stitle{Complexity Analysis.} 

% Let $t_1, t_2, \cdots, t_m$ be the distinct tokens in the text \textT and $f_1, f_2, \cdots, f_m$ the frequencies of these tokens in \textT, respectively. In particular, let $\fmax$ be the \emph{maximum token frequency} in \textT, \ie $\fmax=\max_{i\in [m]}f_i$. This section introduces a partitioning algorithm that produces a partition with $O(\sum_{i\in [m]}f_i \log f_i) = O(n\log\fmax)$ compact windows, better than $O(n\log n)$ especially when $\fmax\ll n$. 

\noindent Let $\fmax$ be the \emph{maximum token frequency} in \textT, \ie $\fmax=\max_{t\in\textT}\freq{t}{\textT}$.  
\begin{theorem} \label{theorem:uppernlogf}
Given a text $\textT$ and a hash function $\hf$, Algorithm~\ref{algo:skyline} generates a partition $\pt(\textT,\hf)$.%; its total number of compact windows is in expectation $O(n + n\log \fmax)$.  % O(n+\sum_{t\in \textT} (f(t,\textT) \ln f(t,\textT))) = 
\end{theorem}

\begin{theorem} \label{theorem:number}
In expectation, the partition $\pt(\textT,\hf)$ generated by Algorithm~\ref{algo:skyline} has $O(n + n\log \fmax)$ compact windows, where $n=|\textT|$.
% O(n+\sum_{t\in \textT} (f(t,\textT) \ln f(t,\textT))) = 
\end{theorem}

\begin{theorem} \label{theorem:timecomplexity}
When given the length and the maximum token frequency of a text, Algorithm~\ref{algo:skyline}, the Monotonic Partitioning algorithm,  is optimal in terms of the expected partition size in the worst case.
\end{theorem}

We prove Theorem~\ref{theorem:uppernlogf} in Appendix~\ref{sec:proofs} (due to space limitation) and Theorem~\ref{theorem:number} in Section~\ref{sec:active}. Note that the ``expectation'' in Theorem~\ref{theorem:number} is due to the nature of the hash function as opposed to any assumption on $\textT$. We also provide a lower bound in Theorem~\ref{theorem:timecomplexity}, with the proof provided in Appendix~\ref{sec:lowerbound} for interested readers.

\subsection{Active Keys and Complexity Analysis} \label{sec:active}

The total number of keys generated by Algorithm~\ref{algo:genkey} is $O(\sum_{t\in \textT}\freqq{t}{\textT})$, which can go up to $O(n^2)$ -- when all the tokens are duplicated (\ie $\fmax=n$). This leads to a running time  of $O(n^2\log n)$ and space $O(n^2)$ just for producing, sorting and storing the keys.

% We observe that not all keys in $\key(\textT)$ need to be enumerated. Specifically, in Algorithm~\ref{algo:skyline} when visiting a key $(b,c)\in\key(\textT)$, it is skipped if a visited key in the skyline dominates the key (Line~\ref{algo:skyline:5}). Although it is non-trivial to identify all such keys in $\key(\textT)$ beforehand, we observe that the \textit{non-active keys} as formalized below must be such keys, which can be skipped beforehand during key generation.

We observe that not all keys in $\key(\textT)$ need to be enumerated. Specifically, in Algorithm~\ref{algo:skyline}, when visiting a key $(b, c) \in \key(\textT)$, the key is skipped if it is dominated by a key in the skyline of visited keys (Line~\ref{algo:skyline:5}). Although it is non-trivial to determine all the skipped keys in advance, we identify a subset of them -- referred to as \textit{non-active keys} -- which, as formalized below, are guaranteed to be dominated and can thus be safely skipped during key generation.

\begin{definition}[Active Hash Value] \label{def:activehash}
Given a text $\textT$ and a hash function $\hf$, a hash value $\hf(t, x)$ for a token $t \in \textT$ is said to be active if and only if $\hf(t, x) < \hf(t, x')$ for all positive integers $x' < x$. 
\end{definition}

That is to say, a hash value $\hf(t, x)$ is active if and only if it is the smallest among $\hf(t,1), \hf(t,2), \cdots, \hf(t,x-1), \hf(t,x)$. For example, Figure~\ref{subfig:hset3} shows all the active hash values in the hash value set $\hset(\textT)$ in Figure~\ref{subfig:hset2}. The bold, red integers are active hash values, while the rest (\ie gray ones) are non-active hash values.

\begin{definition}[Active Key] \label{def:activekeys}
A key $(p,q)$ is active if and only if its hash value $\hf(p,q,\textT)$ is active for token $\textT[p]$ (note that $\textT[p]=\textT[q]$ for a key). The set of active keys in a text $\textT$ is denoted as $\aset(\textT)$.%, where $t=\textT[p]$ and $x=\freq{t}{\textT[p,q]}$.
\end{definition}

For example, Figure~\ref{fig:activekeys} shows the active key set $\aset(\textT)$ of $\textT$ under the hash function $\hf$ from the running example. Although $\textT$ contains 23 keys in total, only 14 of them are active keys.

Note that a key is added to the skyline in Algorithm~\ref{algo:skyline} if and only if it is not skipped, i.e., the condition in Line~\ref{algo:skyline:5} evaluates to false.

\begin{lemma} \label{lemma:activehashingnumber1}
No non-active key is added to the skyline in Algorithm~\ref{algo:skyline}.
\end{lemma}
\begin{proof}
Let $(p, q) \in \key(\textT)$ be a non-active key. Then its hash value $\hf(t,x)$ must be non-active, where $t=\textT[p]=\textT[q]$ and $x=\freq{t}{\textT[p,q]}$. Let $p=q_1<q_2<\cdots<q_x=q$ be the positions such that $\textT[q_i] = t$ for all $1 \leq i \leq x$.  Since $\hf(t, x)$ is non-active, by Definition~\ref{def:activehash}, there exists some $1 \leq i < x$ such that $\hf(t, i) < \hf(t, x)$. Then the corresponding key $(p, q_i)$ has a smaller hash value than $(p, q)$ and satisfies $[p, q_i] \subset [p, q]$, hence it dominates $(p, q)$. Moreover, since $\hf(t, i) < \hf(t, x)$, $(p, q_i)$ is visited before $(p, q)$.

% Since $\hf(t, x)$ is non-active, by Definition~\ref{def:activehash}, there exists some $1 \leq i < x$ such that $\hf(t, i) < \hf(t, x)$. Then the corresponding key $(p, q_i)$ has a smaller hash value than $(p, q)$ and satisfies $[p, q_i] \subset [p, q]$, hence it dominates $(p, q)$. Moreover, since $\hf(t, i) < \hf(t, x)$, $(p, q_i)$ is visited before $(p, q)$.

Now, consider two cases. (1) If $(p, q_i)$ is in the skyline at the time $(p, q)$ is visited, then $(p, q)$ will be skipped by Lemma~\ref{lemma:dom}. (2) If $(p, q_i)$ is not in the skyline, it must be pruned by a key dominates $(p, q_i)$. By Lemma~\ref{lemma:transitive} (dominance is transitive), there exists a key in the skyline that dominates $(p, q_i)$ and thus also dominates $(p, q)$. In this case, $(p, q)$ will also be skipped. In either case, $(p, q)$ will not be added to the skyline. Therefore, only active keys can enter the skyline in Line~\ref{algo:skyline:15}, and all non-active keys are skipped.
\end{proof}

\begin{lemma} \label{lemma:transitive}
Dominance is transitive, i.e., if $(p,q)$ dominates $(p',q')$ and $(p',q')$ dominates $(p'',q'')$, then $(p,q)$ dominates $(p'',q'')$.
\end{lemma}
\begin{proof}
This is because $[p,q]\subset[p',q']\subset[p'',q'']$. 
\end{proof}

% Note that a key is added to the skyline if and only if the key is not skipped in Algorithm~\ref{algo:skyline}.

\begin{lemma} \label{lemma:compactwindowsize_skyline}
The total number of compact windows generated by Algorithm~\ref{algo:skyline} is no more than $2|\aset(\textT)|$. % (the \# of keys that are not skipped). 
\end{lemma}
\begin{proof}

Consider Lines~\ref{algo:skyline:9}-\ref{algo:skyline:13}. The algorithm produces at most $j-i$ compact windows. Note that inserting the key $(b, c)$ results in the removal of $j - i - 1$ other keys from the skyline, and these removed keys will not be reinserted. Thus, among the $j - i$ compact windows generated, one corresponds to the insertion of $(b, c)$ and one to the removal of each key. Thus one compact window is produced upon the insertion of a key and one compact window is produced upon the removal of a key. Moreover, by Lemma~\ref{lemma:activehashingnumber1}, only active keys can be added to the skyline. Hence, the total number of compact windows generated by Algorithm~\ref{algo:skyline} is at most $2 |\aset(\textT)|$.
% Note that the insertion of $(b,c)$ removes $j-i-1$ other keys in the skyline and these keys will never be added back. Therefore, among the $j-i$ compact windows generated, one compact window is produced upon the insertion of a key and one compact window is produced upon the removal of a key. Moreover, by Lemma~\ref{lemma:activehashingnumber1}, only active keys can be added to the skyline. Hence, the total number of compact windows generated by Algorithm~\ref{algo:skyline} is at most $2 \times |\aset(\textT)|$.
% Moreover, a key is added to the skyline if and only if it is not skipped. Thus, the total number of compact windows generated by Algorithm~\ref{algo:skyline} is no more than $2\times$ (the \# of keys that are not skipped). 
\end{proof}
\def\agglf{\kw{agg}(\textT)}
\begin{lemma} \label{lemma:expectednumberofeffkeys}
In expectation, $|\aset(\textT)|=O(n+n\log f_\textT)$, where $n=|\textT|$.
\end{lemma}
\begin{proof}
For any $t\in\textT$ and $i\in[\freq{t}{\textT}]$, the probability that $\hf(t,i)$ is active, i.e., it is smaller than all the hash values $\hf(t,j)$ for all $j < i$ is $\frac{1}{i}$. Therefore, the total number of active hash values among $\hf(t,1),\cdots,\hf(t,\freq{t}{\textT})$ in expectation is $\sum_{i\in[\freq{t}{\textT}]}\frac{1}{i} = O(1 + \ln(f(t,\textT)))$, seeing that even when $f(t,\textT) = 1$, the summation is still $1$. Moreover, there are $O(\freq{t}{\textT})$ keys with hash value $\hf(t,i)$ for any $i\in[\freq{t}{\textT}]$. Therefore, in expectation, the total number of active keys in $\textT$, \ie $|\aset(\textT)|$, is $O(\sum_{t\in \textT} (f(t,\textT) (1 +\ln(f(t,\textT)))) = O(n + \sum_{t\in \textT} (f(t,\textT) \ln(f(t,\textT)))=O(n+n\log f_\textT)$.
% Note that there are $f(t,\textT)$ integers $p$ with $\textT[p] = t$ (starting position), based on Lemma~\ref{lemma:activehashingnumber1}, the expected number of keys that enter the skyline is in expectation $O(\sum_{t\in \textT} (f(t,\textT) (1 +\ln(f(t,\textT)))) = O(n + \sum_{t\in \textT} (f(t,\textT) \ln(f(t,\textT)))=O(n+n\log f_\textT)$.
\end{proof}

Combining Lemmas~\ref{lemma:compactwindowsize_skyline} and~\ref{lemma:expectednumberofeffkeys}, in expectation, the total number of compact windows in the partition $\pt(\textT,\hf)$ generated by Algorithm~\ref{algo:skyline} is $O(|\aset(\textT)|)=O(n+n\log f_\textT)$, which proves Theorem~\ref{theorem:number}.

%Please find the proof in Appendix~\ref{sub:proof_of_lemma_ref_lemma_expectednumberofeffkeys}.

To incorporate this optimization in our algorithm, we replace procedure \textsc{GenerateKeys} with a similar procedure \textsc{GenerateActiveKeys} (Algorithm~\ref{algo:genactivekeys}). It  generates active keys by enumerating, for each token $t$, every possible frequency $x$ of $t$: the algorithm maintains the smallest hash value encountered in $minkey$ and only when the current hash value $\hf(t,x)<minkey$, it produces all the keys whose hash values are $\hf(t,x)$ and updates $minkey$ as $\hf(t,x)$. 

An additional optimization is to sort the active hash values prior to generating the corresponding active keys, rather than generating and subsequently sorting all active keys. This approach is justified by two key observations: (1) active keys associated with the same hash value can be ordered arbitrarily, and (2) the number of active hash values does not exceed the total number of active keys.

% Based on Theorem~\ref{theorem:domed}(1)-(3), Algorithm~\ref{algo:skyline} generates, for each key $(b,c)$ of value $v$ that has not been dominated by any key in $L$ (the set of keys with hash values smaller than $v$), the compact windows that disjointly covers all the subsequences in $[1,b]\times[c,n]$ with hash value $v$, i.e., the subsequences in $[1,b]\times[c,n] \setminus \bigcup_{(p,q) \in L } ([1,p] \times [q,n])$. Since each sequence's hash value must equal to the hash value of a key in $K$ (generated in Line~1, Algorithm~\ref{algo:skyline}), it must have been covered by a compacted generated. Therefore, the compact wondows are a disjoint partition of all the subsequences. This proves Theorem~\ref{theorem:uppernlogf}.

\begin{theorem}\label{the:complexity}
Algorithm~\ref{algo:skyline} with \textnormal{\textsc{GenerateActiveKeys}} in Line~\ref{algo:skyline:1} has time complexity $O(|\aset(\textT)|\log n)$ and space complexity $O(|\aset(\textT)|)$. In expectation, the time complexity is $O(n\log n + n\log n\log \fmax),$
%n\log \fmax \log n + n \log n)$ 
and space complexity is $O(n\log \fmax + n)$.
\end{theorem}

% \begin{theorem}\label{the:complexity}
% Algorithm~\ref{algo:skyline} with \textnormal{\textsc{GenerateActiveKeys}} in Line~\ref{algo:skyline:1} has time complexity $O(|\aset(\textT)|\log n+|\aset(\textT)|\log|\aset(\textT)|+n)$ and space complexity $O(|\aset(\textT)|+n)$. In expectation, the time complexity is $O(n\log \fmax \log n + n \log n)$ and space complexity is $O(n\log \fmax + n)$.
% \end{theorem}

\begin{proof}
Because for any $i \in [n]$, $(i,i)$ is an active key, thus $|\aset(\textT)|\geq n$. The size of the skyline $\sky$ is $O(n)$ as each $b \in [n]$ will have at most one key in $\sky$. The binary search and each update to the skyline take $O(\log n)$. \textnormal{\textsc{GenerateActiveKeys}} only produces active keys $\aset(\textT)$. Each active key costs up to three binary searches of the skyline and two updates and corresponds to up to two compact windows. In addition, \textnormal{\textsc{GenerateActiveKeys}} takes $O(|\aset(\textT)|+n)=O(|\aset(\textT)|)$ time; sorting the active hash values takes  $O(n\log n)$ time as there are at most $n$ active hash values. Thus the time complexity is $O(|\aset(\textT)|\log n)$. The space complexity is $O(|\aset(\textT)|)$. By Lemma~\ref{lemma:compactwindowsize_skyline}, in expectation, $|\aset(\textT)|=O(n + n\log \fmax)$. Therefore, the time complexity is $O(n\log n+n\log n\log \fmax)$; the space complexity is $ O(n\log \fmax + n)$ in expectation. 
\end{proof}

Lemma~\ref{lemma:disjointcomplete} in Appendix~\ref{sec:proofs} shall prove the correctness of this optimization. Note that there are $O(n\fmax)$ keys in $\key(\textT)$. Thus, the time and space complexities of vanilla Algorithm~\ref{algo:skyline} (without the active key optimization) are respectively $O(n\fmax\log n)$ and $O(n\fmax)$. 
% *******************

%\input{src/generation_bak}

\newcommand{\wt}[2]{\ensuremath{w(#1,#2)}\xspace}

% We first consider texts with duplicate tokens where each occurrence of a token in a text is treated as distinct.

% We first consider a special case of weighted tokens where the tokens are weighted by their frequencies in the text (\ie term frequency).

\section{Generalization to Weighted Jaccard Similarity}\label{sec:six}

%We first define a few notations\eat{ that will be used throughout the paper}. A text \textT is a sequence of tokens, where $\textT[i]$ is the $i$-th token in the text. We define $\freq{t}{\textT}$ as the frequency of the token $t$ in \textT, \ie the number of occurrences of $t$ in \textT. The length of the text $\textT$ is $|\textT|$. A subsequence $\textT[i,j]$ of a text $\textT$ consists of the $i$-th token to the $j$-th token (inclusive) of $\textT$. 

Consider a text $\textT$ where each distinct token $t$ is associated with a weight. The weight is determined by a weight function $\wf(t,\textT)$ where $\wt{t}{\textT} > 0$ for $t \in \textT$ and $\wt{t}{\textT} = 0$ for $t \notin \textT$. The weighted Jaccard similarity of two texts $\textT$ and $\textS$ is defined as 
$$\Jaccard^{\wf}_{\textT, \textS} = \frac{\sum_{t\in \textT\cup\textS} \min(\wt{t}{\textT}, \wt{t}{\textS})}{\sum_{t\in \textT\cup\textS} \max(\wt{t}{\textT}, \wt{t}{\textS})}.$$
%$$\Jaccard^{\wf}_{\textT, \textS} = \frac{\sum_{t\in\textT\cap\textS} \min\big(\wt{t}{\textT}, \wt{t}{\textS}\big)}{\sum_{t\in\textT\cup\textS} \max\big(\wt{t}{\textT}, \wt{t}{\textS}\big)}.$$

To estimate weighted jaccard similarity, one can use \textbf{improved consistent weighted sampling}~\cite{ioffe2010improved}. In a nutshell, it designs a hash family $\mathcal{H}$. Given a text $\textT$, for each distinct token $t$ in $\textT$, a hash function $h\in\mathcal{H}$ takes the token $t$ and its weight $\wt{t}{\textT}$ as input and outputs a hash value, denoted as $\hf(t, \wt{t}{\textT})$. With hash function $h$, we define the \textit{weighted min-hash} of a text $\textT$ as the smallest hash value among all distinct tokens in $\textT$, denoted as $\hf(\textT, \wf)$. Formally, 
\begin{equation}\label{eq:wmin}
\hf(\data, \wf) = \min\{\hf(t, \wt{t}{\data}) \mid t \in \data\}.   
\end{equation}
%Note that when the weight function \wt{t}{\textT} is the frequency of $t$ in $\textT$, the weighted min-hash degenerates to the problem we defined in Section~\ref{sec:multi-set}. 
\noindent The hash family $\mathcal{H}$ designed in improved consistent weighted sampling scheme guarantees that for any two texts $\textT$ and $\textS$ and weight function $\wf$, $\mathbf{Pr}[\hf(\textT, \wf) = \hf(\textS, \wf)] = \func^w_{\textT,\textS}.$

% , where $h$ is a random hash function in a hash family $\mathcal{H}$ designed in the improved consistent weighted sampling scheme. The smallest hash value is the 

Thus the weighted Jaccard similarity of two texts $\textT$ and $\textS$ can be accurately and unbiasedly estimated by $k$ independent  hash functions $\hf_1, \cdots, \hf_k$ randomly drawn from the hash family $\mathcal{H}$ as
 \begin{equation}\label{eq:wjaces}
\hat{\Jaccard}^{\wf}_{\textT, \textS}=\frac{1}{k}\sum_{i=1}^{k}\mathbf{1}\{\hf_i(\textT,\wf)=\hf_i(\textS,\wf)\}.
\end{equation}
% with small variance, where $\mathbf{1}$ is an indicator function. In this paper, we aims to identify all the subsequences whose estimated weighted Jaccard similarities to the query are no smaller than the threshold. That is to say, replace the $\Jaccard^{\wf}_{\query, \textT[i,j]}$ in Definition~\ref{def:problem} with $\hat{\Jaccard}^{\wf}_{\query, \textT[i,j]}$.

% The weighted min-hash sketch of a text $\textT$ consists of $k$ weighted min-hash  $\hf_1(\textT,\wf), \cdots, \hf_k(\textT,\wf)$ generated by $k$ independent random hash functions $\hf_1, \cdots, \hf_k$ from the hash family $\mathcal{H}$. The weighted Jaccard similarity of two texts $\textT$ and $\textS$ can be unbiasedly estimated by 
% \begin{equation}\label{eq:jaces}
%\hat{\Jaccard}^{\wf}_{\textT, \textS}=\frac{1}{k}\sum_{i=1}^{k}\mathbf{1}\{\hf_i(\textT,\wf)=\hf_i(\textS,\wf)\}
%\end{equation}
%where $\mathbf{1}$ is an indicator function.

\begin{figure}[!t]\vspace{-1em}
\begin{algorithm}[H] %\small
    \caption{Improved Consistent Weighted Sampling~\cite{ioffe2010improved}}
    \label{algo:icws}
    \kw{Class~ICWS}: \\
    \hspace{1em}\kw{def} \textbf{init}(self, $\Sigma$: the vocabulary of tokens): \\
    \hspace{2em}\textbf{foreach}~~token$~t\in\Sigma$~~\textbf{do}\\  
    \hspace{3em} self.$r_t\sim \textnormal{Gamma}(2,1)$\\
    \hspace{3em} self.$c_t\sim \textnormal{Gamma}(2,1)$\\
    \hspace{3em} self.$\beta_t\sim \textnormal{Uniform}(0,1)$\\
    \vspace{0.1cm}
    \hspace{1em}\kw{def} \textit{h}($t$: a token, $weight$: a positive real value):\\
    \hspace{2em}$y=\textnormal{exp}(r_t(\lfloor\frac{\ln{weight}}{r_t}+\beta_t\rfloor-\beta_t))$\;
    \hspace{2em}$a=\frac{1}{y}\cdot\frac{c_t}{\textnormal{exp}(r_t)}$\;
    \hspace{2em}\KwRet{$\kw{HashValue}(t,y,a)$\;}
%    \vspace{0.1cm}
%    \hspace{1em}\kw{def} \textit{h}($\textT$: a text or a subsequence, $\wf$: a weight function):\\
%    \hspace{2em}\KwRet{$\min\{h(t,\wt{t}{\textT})\mid t\in\textT\}$\;}
%    \vspace{0.2cm}
%    \hspace{1em}\kw{def} \textit{h}($t$: a token, $x$: a positive integer, $\wf$: a weight function):\\
%    \hspace{2em}\KwRet{$h(t, \wt{t}{x})$\;}
    \vspace{0.1cm}
    \hrule
    \vspace{0.1cm}
    \kw{Class~HashValue}:\\
    \hspace{1em}\kw{def} \textbf{init}(self, $t, y, a$): 
    %\\\hspace{2em} 
    self.$t=t$,~~self.$y=y$,~~self.$a=a$\;
    % \vspace{0.1cm}
    
    \hspace{1em}\kw{def} operator$=$($v_1, v_2$): 
    %\\\hspace{2em} 
    \KwRet{$v_1.t = v_2.t$ \textbf{\textnormal{and}} $v_1.y = v_2.y$\;}
    \hspace{1em}\kw{def} operator$<$($v_1, v_2$): 
    %\\\hspace{2em} 
    \KwRet{$v_1.a < v_2.a$\;}
    % \vspace{0.1cm}
     % \tcp{: let $\Sigma$ be the token universe;}

%     Initialization: Sample, for each token $t$ in the token universe $\Sigma$, $r_t\sim \textnormal{Gamma}(2,1)$, $c_t\sim \textnormal{Gamma}(2,1)$ and $\beta_t\sim \textnormal{Uniform}(0,1)$\;
    % }
    %\hrule
   % \vspace{0.1cm}
 %   \KwIn{$\textT$: a text; $w$: a weight function.}
 %   \KwOut{$(t,y)$: the weighted min-hash of \textT.}
    % \Begin{
  %  $t=$\;

    %\ForEach{distinct token $t\in\textT$}{
    %    $y_t=\textnormal{exp}(r_t(\lfloor\frac{\ln{\wt{t}{\textT}}}{r_t}+\beta_t\rfloor-\beta_t))$; \ $a_t=c_t/(y_t~\textnormal{exp}(r_t))$\;
   % }
   % $t\gets \argmin_{t\in\textT} a_t$; $y\gets y_t$\;
   % \KwRet{$(t,y)$\;}
    % }
\end{algorithm}\vspace{-1.25em}
\end{figure}
% Algorithms~\ref{algo:init} and~\ref{algo:wmh} 

\stitle{Implementation Details.} Algorithm~\ref{algo:icws} implements the improved consistent weighted sampling scheme~\cite{ioffe2010improved}. Drawing a random hash function $\hf\in\mathcal{H}$ is the same as instantiating a new object $o$ of \kw{ICWS} class (Lines~1-10). The constructor (Lines~2-6) samples three independent random variables for every token $t$ in the vocabulary $\Sigma$ from the \kw{ICWS} specified distributions. The hash computation (Lines 7-10) $o.\hf(t,\wt{t}{\textT})$ of $o$ returns the hash value based on token $t$ and the weight of $t$ in $\textT$. Note that an object is initialized with a random seed; the set of possible objects form the hash family. 

%corresponds to the method call $o.\hf(t,\wt{t}{\textT})$. 
%The hash computation $\hf(t,\wt{t}{\textT})$ and $\hf(\textT,\wf)$ respectively correspond to the method calls $o.\hf(t,\wt{t}{\textT})$ and $o.\hf(\textT,\wf)$. 

A hash value $(t,y,a)$ is an instance of the \kw{HashValue}, where the comparators $=$ and $<$ are defined. For simplicity, we assume there are no hash collisions: for any two hash values $(t_1,y_1,a_1)$ and $(t_2,y_2,a_2)$, if $t_1\neq t_2$, then $a_1\neq a_2$. 
Comparator $<$ provides hash values a total order, \ie for any two hash values $v_1=(t_1,y_1,a_1)$ and $v_2=(t_2,y_2,a_2)$, we say $v_1<v_2$ if and only if $a_1 < a_2$. 
Besides, for a token $t$, $r_t, c_t$ and $\beta_t$ depend solely on $t$. Thus the weight determines both $y$ and $a$, \ie there is a one-to-one mapping between $y$ and $a$ in a hash value $(t,y,a)$. Thus, $v_1=v_2$ iff. $t_1=t_2$ and $y_1=y_2$.

% For ease of presentation, we assume there is no hash collision in this paper: for any two hash values $(t_1,y_1,a_1)$ and $(t_2,y_2,a_2)$, if $t_1\neq t_2$, then $a_1\neq a_2$. Based on the implementation of the hash function $\hf(t,weight)$, there is a one-to-one mapping between $y$ and $a$ for the same $t$. This is because the derivative of $a$ with respect to $y$ is always negative. Thus, the hash values have a total order (which is the order of $a$ value in $(t,y,a)$). Thus the $\min$ operator is well defined.

% Both the hash value and the weighted min-hash are instances of the \kw{HashValue} class, which are tuples $(t,y,a)$. Based on the implementation of the hash function $\hf(t,weight)$, for the same token $t$, there is a one-to-one mapping between $a$ and $y$ as the derivative of $a$ with respect to $y$ is always negative. 

% For the same token $t$, there is a one-to-one mapping between the values of $a$ and $y$. The value of $a$ is uniquely determined by $t$ and $y$. This is because $c_t$ and $r_t$ are constants, while the derivative of $a$ with respect to $y$ is always negative.

\stitle{Consistency and Uniformity.} The weighted min-hash under improved consistent weight sampling has the nice properties below: % $(t,y,a)$ of a text $\textT$ satisfies
\begin{itemize}[leftmargin=*]
    \item \textit{Uniformity}: Denote by $(t,y,a)$ the weighted min-hash of a text $\textT$, then $(t, y)$ is distributed uniformly over $\cup_{t'\in\Sigma} \{t'\} \times [0, \wt{t'}{\textT}]$.
    \item \textit{Consistency}: Let $v=(t,y,a)$ be the weighted min-hash of a text $\textT$. Given a text $\textS$ with $\wt{t'}{\textS}\leq\wt{t'}{\textT}$ for all $t'\in\Sigma$. If $y\leq\wt{t}{\textS}$, then $v$ must also be the weighted min-hash of $\textS$.   
%    Let $v=(t,y,a)$ and $v'=(t',y',a')$ be the weighted min-hash of two texts $\textT$ and $\textS$. If  $\wt{t}{\textT}\leq\wt{t}{\textS}$ for all $t\in\Sigma$, then whenever $y'\leq\wt{t}{\textT}$, $v=v'$.%$(t, y)$ is also sampled from $\textT$.
\end{itemize}
% and arbitrary weights, $\hf(t, \mathit{weight}_1) \ne \hf(s, \mathit{weight}_2)$.

\stitle{Assumption on the Weight Function (\kw{AoW}):} \emph{Given a text $\textT$, we assume that for any token $t$ in $\textT$, its weight is a monotonically increasing function of its frequency $\freq{t}{\textT}$ and is independent of other properties of $\textT$.} With \kw{AoW}, we can simplify the notation of the weight of a token $t$ in $\textT$ as $\wt{t}{x}$ where $x=\freq{t}{\textT}$. We can also simplify the hash function as $\hf_{\wf}(t,x)\doteq\hf(t,\wt{t}{x})$  where $x$ is a positive integer. The \kw{AoW}, \ie for any $1\leq x\leq x'$, $\wt{t}{x}\leq\wt{t}{x'}$, is reasonable especailly when it comes to TF-IDF weights. % applications below. 

%\stitle{Assumption on the Weight Function.} Given a text $\textT$, we assume that for a token $t$, its weight is dependent only on its frequency $\freq{t}{\textT}$ in $\textT$, and is monotonically increasing with the frequency $\freq{t}{\textT}$. With this assumption, we can overload the weight of a token $t$ in $\textT$ as $\wt{t}{x}$ where $x=\freq{t}{\textT}$. Define \[\hf(t,x,\wf)\doteq\hf(t,\wt{t}{x})\text{ where $x$ is a positive integer}.\] Second, the weight function $\wt{t}{x}$ is monotonically increasing with frequency $x = \freq{t}{\textT}$, \ie for any $ x\leq x'$, $\wt{t}{x}\leq\wt{t}{x'}$. This assumption is reasonable as it aligns with the TF-IDF applications below. 

\stitle{Term Frequency–Inverse Document Frequency (TF-IDF)}~\cite{singhal2001modern}. TF-IDF is a widely used weighting scheme that captures both the importance of a token within a text and its rarity across a corpus $\dset$. The TF-IDF weight function is defined as $\wf(t, \textT) = \kw{tf}\cdot\kw{idf}$ where $\kw{tf}$ is the TF weight and $\kw{idf}$ is the IDF weight. A few example TF and IDF weights we support (but not limited to) are listed in Table~\ref{tab:tf_and_idf_weights_we_support}. 

\begin{table}[!t]%\vspace{-.5em}
\centering  \small
\begin{tabular}{|rl|rl|}\hline
\multicolumn{2}{|c|}{\kw{tf} weight functions} & \multicolumn{2}{c|}{\kw{idf} weight functions} \\\hline
binary: & $\textbf{1}\{\freq{t}{\textT}>0\}$ & unary: & 1 \\
raw count: & $\freq{t}{\textT}$ & standard: & $\log\frac{N}{N_t}$\\
logarithmic: & $\log(\freq{t}{\textT} + 1) $ & smooth: & $\log(\frac{N+N_t}{N_t})+1$\\
squared: &$\freq{t}{\textT}^2$ & probabilistic: &$\log\frac{N-N_t}{N_t}$ \\\hline
% exponential: &$\exp{\freq{t}{\textT}}$&  &  \\\hline
\end{tabular}
\caption{A snippet of TF and IDF weights we support. \textnormal{Here $N=|\dset|$ is the number of texts in the corpus $\dset$ and $N_t=|\{\textS\in\dset\mid t\in\textS\}|$ is the number of texts  in the corpus containing $t$. $N$ is a global constant while $N_t$ is a constant local to the token $t$.}}
\label{tab:tf_and_idf_weights_we_support}
\vspace{-1.5em}
\end{table}
% \[\small
%   \kw{tf} = \begin{cases}
%     \textbf{1}\{\freq{t}{\textT}>0\} & \text{binary TF,}\\
%     \freq{t}{\textT} & \text{raw count TF,} \\
%     \log(\freq{t}{\textT} + 1) & \text{log normalization TF.} 
%       \end{cases}
% \]
% \noindent The IDF weight we support includes (but is not limited to):
% \[\small
%   \kw{idf} = \begin{cases}
%     1 & \text{unary IDF,} \\
%     \log\frac{N}{N_t} & \text{standard IDF,} \\
%     \log(\frac{N}{N_t}+1)+1 & \text{smooth IDF,} \\
%     \log\frac{N-N_t}{N_t} & \text{probabilistic IDF,} 
%   \end{cases}
% \]

\stitle{Hash Value Set.} We omit the weight function $\wf$ in the notation of a hash value $\hf_{\wf}(t,x)$ when the context is clear and abbreviate the hash value as $\hf(t,x)$. We can define the hash value set of a subsequence in the same way as Definition~\ref{def:hset}.

% Formally, we define the hash value set of a subsequence.
\begin{definition}\label{def:whset}
Given a random hash function $\hf$ from $\mathcal{H}$ and a weight function $\wf$, the hash value set of a subsequence $\textT[i,j]$ is $$\hset(\textT[i,j],\hf)=\{\hf(t,x)\mid t\in\textT[i,j], 1\leq x\leq \freq{t}{\textT[i,j]}\}.$$ 
\end{definition}

% produce the compact windows $\cw$ of \textT in the ascending order of their hash value $v$. 
Next we show that, under the assumption of $\kw{AoW}$, the hash values under consistent weighted sampling have the property below.

\begin{lemma}\label{lemma:whset}
Given a random hash function $\hf$ from $\mathcal{H}$ and a weight function $\wf$ under the assumption of $\kw{AoW}$, we have $\hf(t,1)\geq \hf(t,2)\geq\cdots\geq\hf(t,x)$ for any token $t$ in any text $\textT$ where $x=\freq{t}{\textT}$.
\end{lemma}
\begin{proof}
Under the assumption of $\kw{AoW}$, we have $\wf(t,1)\leq \wf(t,2)\leq\cdots\leq\wf(t,x)$. Now consider texts $\textT_1, \textT_2, \cdots, \textT_x$ where $\textT_i$ consists of exactly $i$ copies of token $t$. Then $\textT_i$ has only one hash value $\hf(t,i)$, which must be its weighted min-hash by Equation~\ref{eq:wmin}. Consider any $j>i$. Let $\hf(t,i)=(t,y,a)$ and $\hf(t,j)=(t,y',a')$, which are the weighted min-hash of $\textT_i$ and $\textT_j$. By uniformity, $y$ distributes uniformly over $[0,\wt{t}{i}]$ and so does $y'$ over $[0,\wt{t}{j}]$. By consistency, because $\wt{t}{i}\leq\wt{t}{j}$, if $y'\leq\wt{t}{i}$, we have $\hf(t,i)=\hf(t,j)$. Otherwise, \ie $y'>\wt{t}{i}$, because $y\leq \wt{t}{i}$, we have $y<y'$. Under the same $t$, i.e., $r_t, c_t, \beta_t$ are fixed, we have $a \propto \frac{1}{y}$ (Line~9, Algorithm~\ref{algo:icws}). Thus we have $a>a'$, which indicates $\hf(t,i)>\hf(t,j)$. In both cases, we have $\hf(t,i)\geq \hf(t,i)$. Thus $\hf(t,1)\geq \hf(t,2)\geq\cdots\geq\hf(t,x)$ for any $t$ and $\textT$.
\end{proof}

Based on Lemma~\ref{lemma:whset} and Equation~\ref{eq:wmin}, we  have 
\begin{equation}\label{eq:whfmin}
\hf(\textT[i,j])=\min \hset(\textT[i,j],\hf).    
\end{equation}

% The hash value $\hf(t,\wt{t}{\textT})$ of a distinct token $t$ in a text $\textT$ depends solely on the token $t$ and the weight $\wt{t}{\textT}$. Under the above assumption, the weight $\wt{t}{\textT}$ depends solely on the token $t$ and its frequency $\freq{t}{\textT}$. Thus the hash value depends solely on the token $t$ and the frequency $\freq{t}{\textT}$.  Thus we denote the hash value as $\hf(t,\freq{t}{\textT},\wf)$. We omit the weight function whenever it is clear from the context and denote the hash value as $\hf(t,x)$.

% \begin{algorithm}[ht]
%     \caption{Initialization}
%     \label{algo:init}
%     \KwIn{$\Sigma$: the set of all possible tokens.}
% %    \KwOut{independent random variables for each token.}
%     % \Begin{
%     \ForEach{token $t\in\Sigma$}{    
%         sample $r_t\sim \textnormal{Gamma}(2,1)$\;
%         sample $c_t\sim \textnormal{Gamma}(2,1)$\;
%         sample $\beta_t\sim \textnormal{Uniform}(0,1)$\;
%     }
% %    \KwRet{$(a,t,x,y)$\;}    
%     % }
% \end{algorithm}

With the above definitions, all the definitions, lemmas, theorems, algorithms, and conclusions naturally extend to the weighted min-hash. For example, a hash value $\hf(t,x)$ is active if and only if it is strictly greater than all of $\hf(t,1), \cdots, \hf(t,x-1)$.

% \todo{Analyze the complexity for different weight function. The complexity is actually very different from $O(n\log f)$}

\begin{lemma} \label{lemma:weightfunctions}
In expectation, $|\aset(\textT)|$ is $O(n)$ for binary TF, $O(n + n\log\fmax)$ for raw count TF, $O(n + \log\log\fmax)$ for logarithmic TF, and $O(n + n\log\fmax)$ for squared TF, each combined with any IDF in Table~\ref{tab:tf_and_idf_weights_we_support}.
% $|\aset(\textT)|=O(n+n\log f)$ for raw count TF, $|\aset(\textT)|=O(1+\log f)$ for log normalization TF.
\end{lemma}
\begin{proof}
Given a text with $n$ tokens and a weight function $\wf$. Consider a token $t\in\textT$. Let $x=\freq{t}{\textT}$. Clearly, $\hf(t,1)$ must be an active hash value. For any $i\in[2,x]$, based on the proof of Lemma~\ref{lemma:whset}, $\hf(t,i)$ is active if and only if $\hf(t,i)<\hf(t,i-1)$. Let $\hf(t,i)=(t,y_i,a_i)$ and $\hf(t,i-1)=(t,y_{i-1},a_{i-1})$. $\hf(t,i)<\hf(t,i-1)$ if and only if $y_i\in (\wt{t}{i-1},\wt{t}{i}]$. By uniformity, $y_i$ distributed uniformly over $[0,\wt{t}{i}]$. Thus the probability that $\hf(t,i)$ is active is $\frac{\wt{t}{i}-\wt{t}{i-1}}{\wt{t}{i}}$. There are $x-i+1=O(x)$ keys corresponding to the hash value $\hf(t,i)$ (imagine a sliding window of size $i$ over a string of length $x$). Thus in expectation, the number of active keys in $\textT$ is 
$
\mathbb{E}[\aset(\textT)]=O(\sum_{t\in\textT}\freq{t}{\textT}\sum_{i=1}^{\freq{t}{\textT}}\frac{\wt{t}{i}-\wt{t}{i-1}}{\wt{t}{i}})
.$ Note that $\wt{t}{0}=0$.
Clearly, for $x\geq 1$, when $\wt{t}{x}=x\cdot\textsf{IDF}$, we have $\mathbb{E}[|\aset(\textT)|]=O(n+n\log\fmax)$. When $\wt{t}{x}=1\cdot\textsf{IDF}$,  $\mathbb{E}[|\aset(\textT)|]=O(n)$. For $\wt{t}{x}=\log(x+1)\cdot\textsf{IDF}$,  $\mathbb{E}[|\aset(\textT)|]=O(n+n\log\log\fmax)$ based on taylor expansion. For $\wt{t}{x}=x^2\cdot\textsf{IDF}$,  $\mathbb{E}[|\aset(\textT)|]=O(n+n\log\fmax)$.  %For $\wt{t}{x}=e^x\cdot\textsf{IDF}$,  $\mathbb{E}[|\aset(\textT)|]=O(n\fmax)$.
\end{proof}
% the equal (\ie $==$), less than $(<)$, and greater than ($>$) operators over two hash values $(t,y,a)$ and $(t',y',a')$.

% We observe that the value $a_t$ is monotonically increasing with $\freq{t}{\textT}$, \ie the frequency of $t$ in $\textT$, in $\mathrm{WeightedMinHash}$. Moreoever, the values $a_s$ and $a_t$ are independent for any two distinct tokens $t$ and $s$ in the alphabet. Thus we can calculate the hash value 

% Thus, $\tf\idf(t, \textT, \dset) = \tf(t, \textT) \cdot \idf(t, \dset)$. 

A caveat is that, instead of breaking ties arbitrarily for keys with the same hash value when visiting the key set of a text, we need to order the keys firstly by their hash values and secondly by their frequencies. That is to say for two keys $(p,q)$ and $(p',q')$ with the same hash value, if $\freq{\textT[p]}{\textT[p,q]}>\freq{\textT[p']}{\textT[p',q']}$, we should visit $(p,q)$ before $(p',q')$. If $\freq{\textT[p]}{\textT[p,q]}=\freq{\textT[p']}{\textT[p',q']}$, we can break tie arbitrarily.

\begin{figure*}[htbp] \vspace{-5em}
    \centering
    \setlength{\subfigcapskip}{-5pt}
    \includegraphics[width=0.3\linewidth]{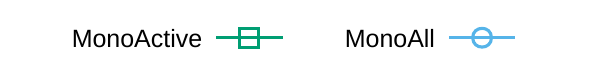}
    \vspace{-1em} \\
    \subfigure[\owt ($k=64$)]{
    \label{subfig:5a}
    \includegraphics[width=0.16\linewidth]{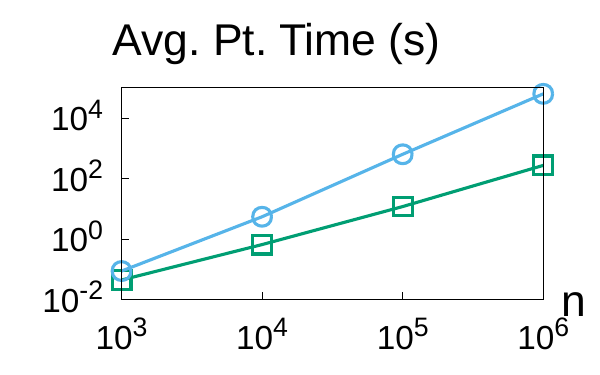}
    }
    \hspace*{-1em}
    \subfigure[\pan ($k=64$)]{
    \label{subfig:5b}
    \includegraphics[width=0.16\linewidth]{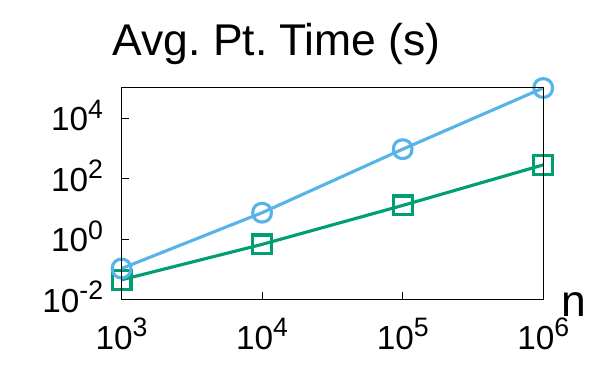}
    }
    \hspace*{-1em}
    \subfigure[\owt ($k=64,n=5000$)]{
    \label{subfig:5c}
    \includegraphics[width=0.16\linewidth]{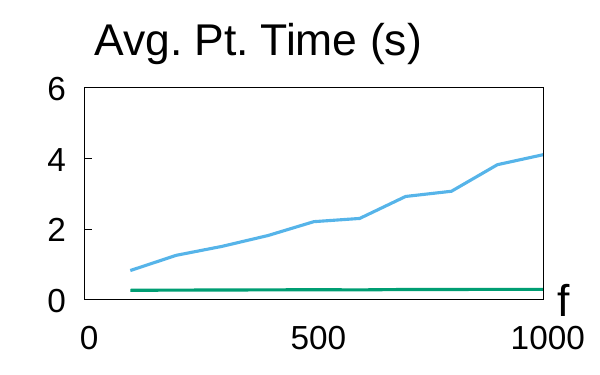}
    }
    \hspace*{-1em}
    \subfigure[\pan ($k=64,n=5000$)]{
    \label{subfig:5d}
    \includegraphics[width=0.16\linewidth]{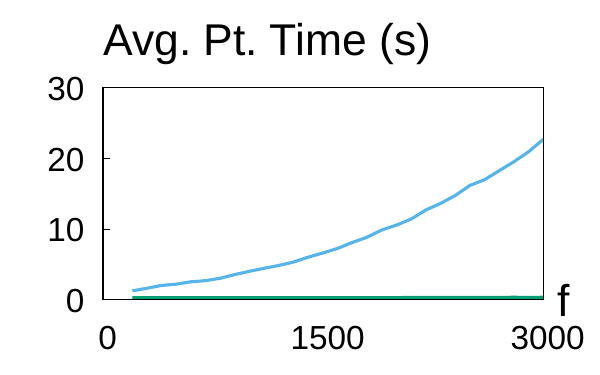}
    }
    \hspace*{-1em}
    \subfigure[\owt ($n=100000$)]{\label{subfig:5e}
    \includegraphics[width=0.16\linewidth]{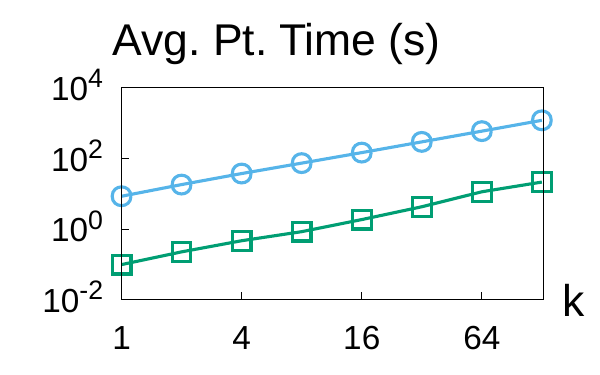}
    }
    \hspace*{-1em}
    \subfigure[\pan ($n=100000$)]{\label{subfig:5f}
    \includegraphics[width=0.16\linewidth]{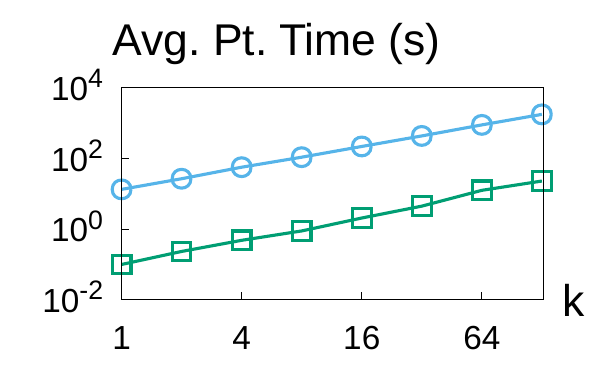}
    }
    \vspace{-1.25em} \\
    \caption{Evaluating the active hash optimization. $k$:  sketch size, $n$:  text length, $f$:  maximum token frequency.} \vspace{-1.25em}
    \label{exp:active}
\end{figure*}

\begin{figure*}[htbp]
    \centering
   % \vspace{-1.3cm}
    \setlength{\subfigcapskip}{-5pt}
    % \hspace*{-10em}
    \includegraphics[width=0.45\linewidth]{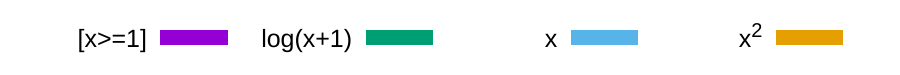}
    % \hspace*{-15em}
    \includegraphics[width=0.45\linewidth]{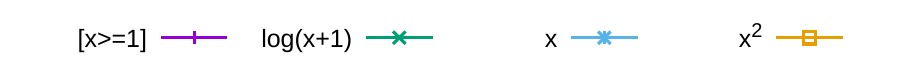}
    \vspace{-1em} \\
    \subfigure[\owt\eat{ ($k=64$)}]{
    \label{subfig:9a}
    \includegraphics[width=0.16\linewidth]{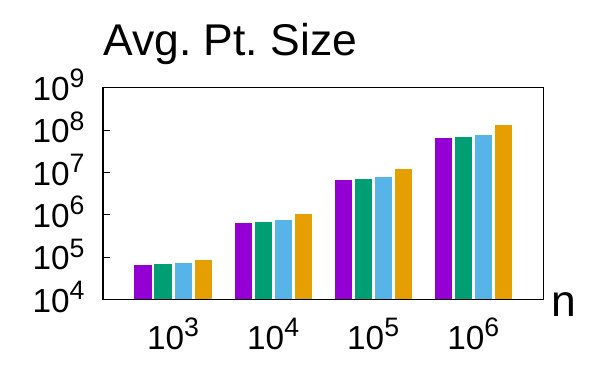}
    }
    \hspace*{-1em}
    \subfigure[\pan\eat{ ($k=64$)}]{
    \label{subfig:9b}
    \includegraphics[width=0.16\linewidth]{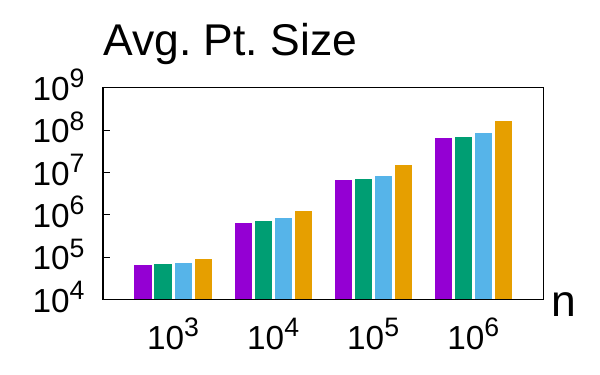}
    }
    \hspace*{-1em}
    \subfigure[\owt\eat{ ($k=64$)}]{
    \label{subfig:9c}
    \includegraphics[width=0.16\linewidth]{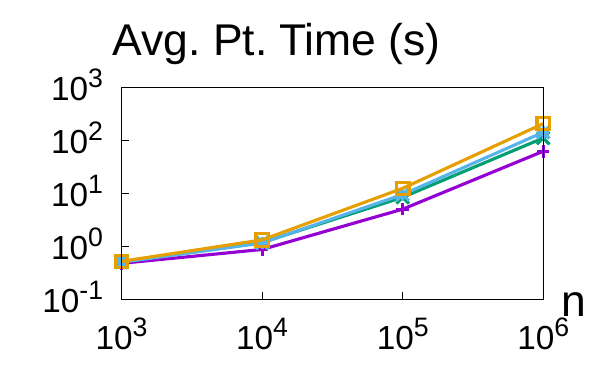}
    }
    \hspace*{-1em}
    \subfigure[\pan\eat{ ($k=64$)}]{
    \label{subfig:9d}
    \includegraphics[width=0.16\linewidth]{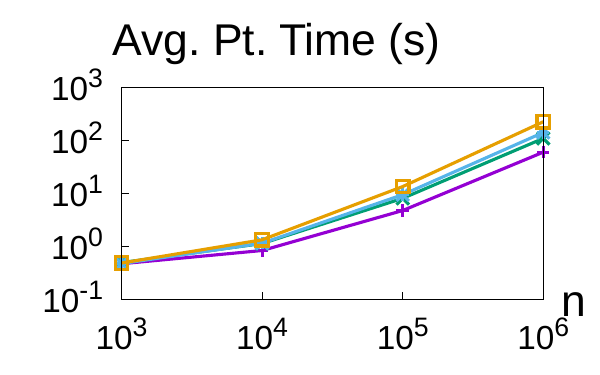}
    }
    \hspace*{-1em}
    \subfigure[\owt\eat{ ($k=64$)}]{
    \label{subfig:9e}
    \includegraphics[width=0.16\linewidth]{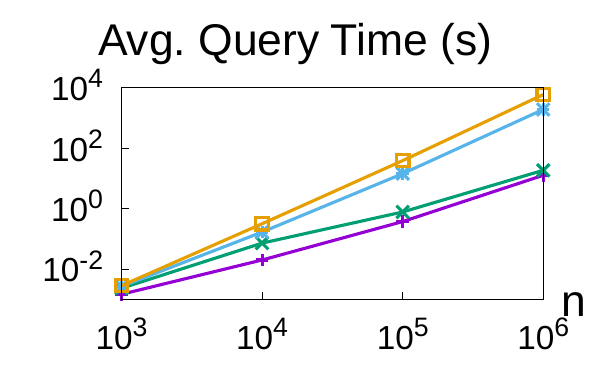}
    }
    \hspace*{-1em}
    \subfigure[\pan\eat{ ($k=64$)}]{
    \label{subfig:9f}
    \includegraphics[width=0.16\linewidth]{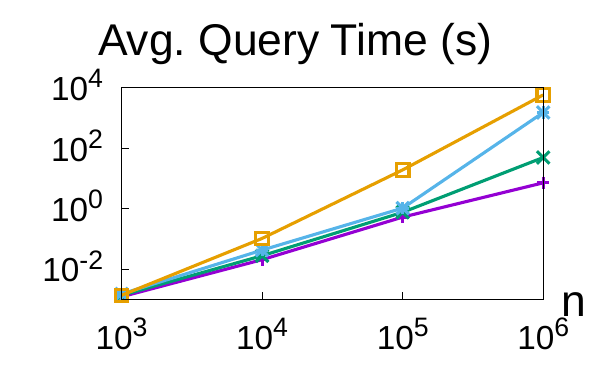}
    }
    \vspace{-1em} 
    \\
    \includegraphics[width=0.45\linewidth]{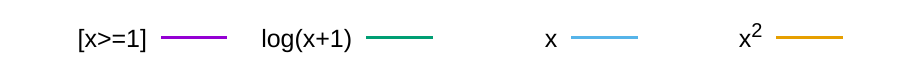}
    \vspace{-1em} \\
    \subfigure[\owt ($\eat{k=64,}n=5000$)]{
    \label{subfig:9g}
    \includegraphics[width=0.16\linewidth]{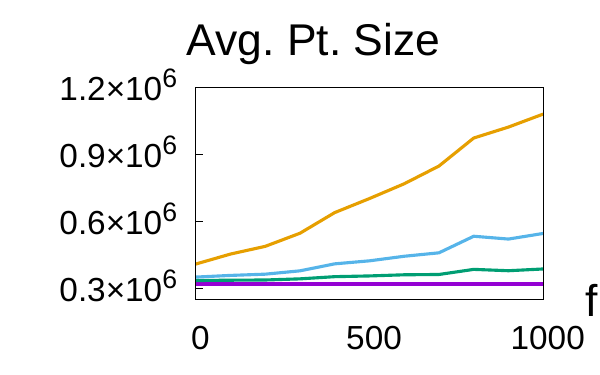}
    }
    \hspace*{-1em}
    \subfigure[\pan ($\eat{k=64,}n=5000$)]{
    \label{subfig:9h}
    \includegraphics[width=0.16\linewidth]{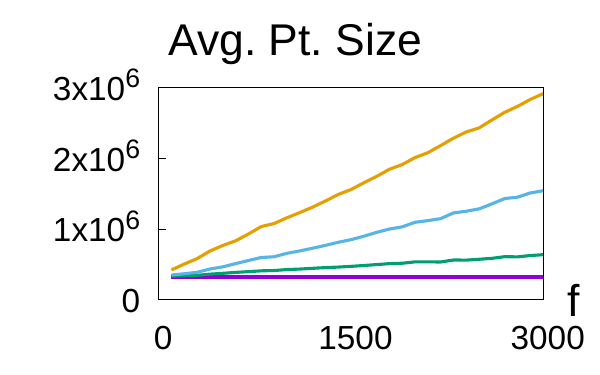}
    }
    \hspace*{-1em}
    \subfigure[\owt ($\eat{k=64,}n=5000$)]{
    \label{subfig:9i}
    \includegraphics[width=0.16\linewidth]{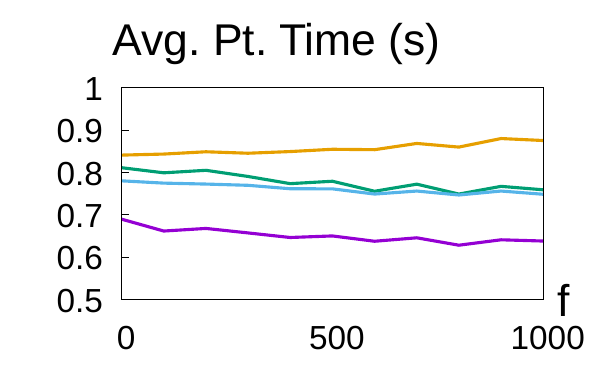}
    }
    \hspace*{-1em}
    \subfigure[\pan ($\eat{k=64,}n=5000$)]{
    \label{subfig:9j}
    \includegraphics[width=0.16\linewidth]{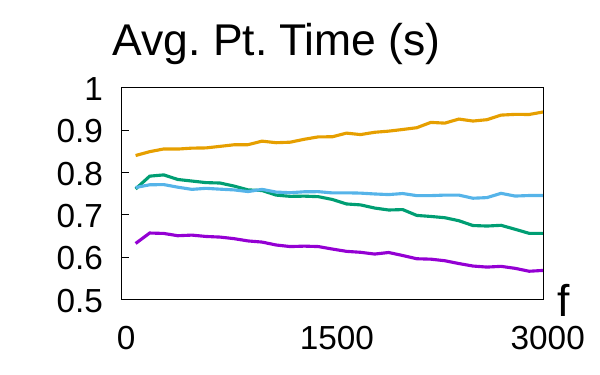}
    }
    \hspace*{-1em}
    \subfigure[\owt ($\eat{k=64,}n=5000$)]{
    \label{subfig:9k}
    \includegraphics[width=0.16\linewidth]{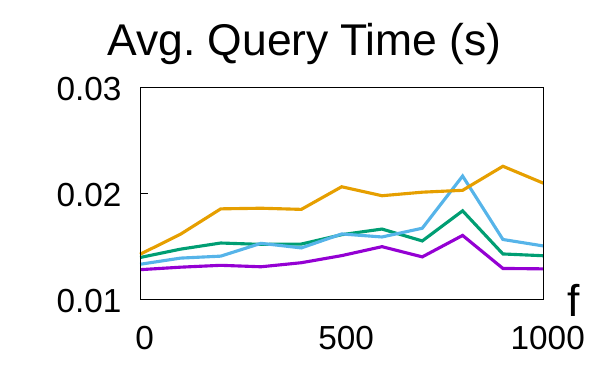}
    }
    \hspace*{-1em}
    \subfigure[\pan ($\eat{k=64,}n=5000$)]{
    \label{subfig:9l}
    \includegraphics[width=0.16\linewidth]{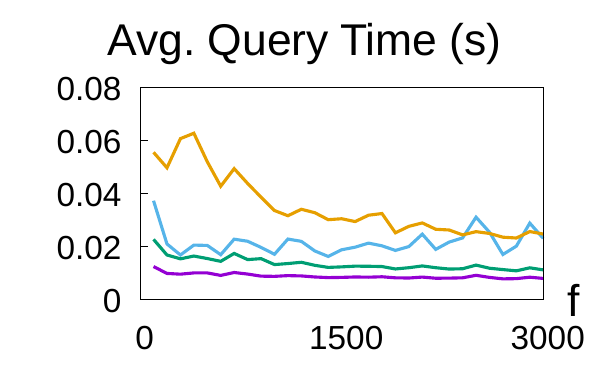}
    }
    \vspace{-1.25em} \\
    \caption{Evaluating weighted Jaccard similarity under various weight functions. $k=64$ in all experiments. \eat{ \miao{k:  sketch size, n:  text length, f:  maximum frequency.}}} \vspace{-1.25em}
    \label{exp:weighted}
\end{figure*}

\section{Experiment} \label{sec:experiment}

\stitle{Environment.} We implemented all the methods, including the baselines, using C++. The programs were compiled with GCC 11.4.0 and optimized using the -O3 flag. All experiments were conducted on a machine equipped with an Intel Xeon Gold 6212 CPU @ 2.40GHz and 1 TB of memory, running Ubuntu 18.04.5.

\stitle{Datasets.} We conducted experiments on two datasets, \pan~\cite{DBLP:conf/clef/PotthastEBSR11} and \owt~\cite{Gokaslan2019OpenWeb}. \pan (\url{https://pan.webis.de/data.html}) is a benchmark for external plagiarism detection~\cite{DBLP:conf/clef/PotthastEBSR11}. Each text in \pan is a book. Tokens are generated by splitting words based on whitespace. In total, there are 11,093 texts and 642,380,109 tokens in the dataset. \owt consists of 8 million web texts highly ranked on Reddit. Each token is a byte-pair-encoding generated by the GPT-2 tokenizer\cite{radford2019language}. The vocabulary size is 50,257. \eat{The token frequency distributions of the two datasets are shown in Figure~\todo{XXX add if you have time}.}

\stitle{Parameters and Settings.} There are three parameters, the text length $n$, the maximum token frequency $f$, and the min-hash sketch size $k$ (see Section~\ref{sec:multi-set}). We report the compact window generation time (\ie partition time) and the number of compact windows generated (\ie partition size) per text. To avoid text-dependent bias, for each configuration of parameters, we randomly choose $10$ text under the configuration and report the average partition time/size. 

To fix the text length $n$, texts with more than $n$ tokens were truncated to their first $n$ tokens. Texts with fewer than $n$ tokens were concatenated until they reached or exceeded $n$ tokens, after which they were truncated to the first $n$ tokens as a single text.

To evaluate the impact of the maximum token frequency $f$, we fixed the text length $n$ and used the first 100,000 texts in \owt and the first 10,000 texts in \pan. These texts were grouped based on their maximum token frequencies using intervals of 100. For example, texts with maximum token frequencies in the range $[1,100]$ formed one group, while those in $[101,200]$ formed another. 

% and reported the average results. 

\stitle{Compared Methods.} We consider 3 methods in our experiment.  

\begin{itemize}[leftmargin=0.5cm]
%\item \scp: Our single column partition, Algorithm~\ref{algo:naive}.
\item \allkeys: Our vanilla monotonic partitioning, Algorithm~\ref{algo:skyline}.
\item \mono: Our monotonic partitioning with active hash optimization (Section~\ref{sec:active}).  It calls Algorithm~\ref{algo:genactivekeys} to generate keys. 
\item  \ala: A greedy partitioning algorithm for multi-set Jaccard similarity~\cite{allign}. \ala generates compact windows in recursion. In each iteration, it takes a rectangle of the shape $[a,b]\times[a,c]$ as input and partitions all the subsequences $\textT[i,j]$ in this rectangle into a few compact windows and one or more smaller rectangles of that shape. The smaller rectangles are recursively partitioned until no rectangles left. At the beginning, the input rectangle is $[1,n]\times[1,n]$. However, \ala is greedy and its time and space complexities are unknown.
\end{itemize}

\subsection{Evaluating Active Hash Optimization}
This section evaluates our two methods \mono and \allkeys on two datasets \owt and \pan. Since the optimization in \mono does not change the generated compact windows, the partition size of \mono and \allkeys are the same,  we only consider the partition time. Recall that in Theorem~\ref{the:complexity}, the time complexities of \mono and \allkeys are respectively $O(n\log n \log f )$ and $O(nf\log n)$ when the maximum frequency $f>1$.

Figures~\ref{subfig:5a}-\ref{subfig:5b} show that the  partition time of two methods increases with the text length. \mono consistently outperforms \allkeys because \mono avoids unnecessary computations for non-active keys. The ratio between the partition time of \allkeys and that of \mono echoes our complexity analysis.

% Figures~\ref{subfig:5a}-\ref{subfig:5b} show that when $k$ is fixed to 64,  % in Section~\ref{sec:active}. 

%The gap be bigger when $f$ grows large 

% <<<<<<< HEAD
Figures~\ref{subfig:5c}-\ref{subfig:5d}  show that with $k=64$ and $n=5000$, as the maximum frequency $f$ increases, the runtime of \mono remained nearly constant whereas that of \allkeys grew almost linearly with $f$. For example, on \pan, when $f\in [901,1000]$, $f\in[1901,2000]$, and $f\in[2901,3000]$, \mono took 0.27, 0.29, 0.33 seconds for partition generation on average, while \allkeys took 3.98, 10.56, 22.59 seconds. When $f$ increased around threefold, the runtime of \mono and \allkeys respectively increased 1.2$\times$ and 5.7$\times$. This observation is consistent with our complexity analysis.% (Theorem~\ref{the:complexity}), where the time complexities of \mono and \allkeys are respectively $O(n\log n + n\log f \log n)$ and $O(nf\log n\log f)$.
% =======
% Figures~\ref{subfig:6c}-\ref{subfig:6d}  show that with $k=64$ and $n=5000$, as  the maximum frequency $f$ increases, the runtime of \mono remained nearly constant whereas the runtime of \allkeys grew almost linearly with $f$. For example, on \pan, when $f\in [901,1000]$, $f\in[1901,2000]$, and $f\in[2901,3000]$, \mono took 0.0042, 0.0045, 0.0048 seconds for partition generation on average, while \allkeys took 0.1, 0.23, 0.49 seconds. When $f$ increased around threefold, the runtime of \mono and \allkeys respectively increased 1.1$\times$ and 5$\times$. This observation is consistent with our complexity analysis (Theorem~\ref{the:complexity}), where the time complexities of \mono and \allkeys are respectively $O(n\log n + n\log f \log n)$ and $O(nf\log n)$.
% >>>>>>> refs/remotes/origin/master

Figures~\ref{subfig:5e}-\ref{subfig:5f} show that both methods exhibit linear runtime growth with $k$. This is because the partition generation is independent across the $k$ random universal hash functions. Thus, hereinafter, we skip evaluating the impact of the parameter $k$.

%as $k$ increases, both methods grew linearly, though \mono performed better. With $n$ increasing, the runtime of the original Monotonic algorithm grows faster because its  complexity depends on $f$, which is greater than the $\log f$ term used in the active version. And $f$ itself grows as $n$ increases. The experiment that varies max frequency $f$ makes this distinction clearer: the original version scales almost linearly with $f$, whereas the active approach remains largely unaffected.

%we evaluate the impact of the “active” optimization by comparing two algorithms: Monotonic and MonotonicActive. The Monotonic algorithm considers all keys, whereas MonotonicActive only considers active keys. Since both produce the same set of compact windows, we focus on their runtime. Based on earlier complexity analysis, the time complexity of Monotonic algorithm is $O(nf \log n)$, while MonotonicActive has $O(n\log f \log n)$. As shown in the figure~\ref{exp:active}, when $k$ increases, both algorithms grow linearly, though the active version performs better. With $n$ increasing, the runtime of the original Monotonic algorithm grows faster because its  complexity depends on $f$, which is greater than the $\log f$ term used in the active version. And $f$ itself grows as $n$ increases. The experiment that varies max frequency $f$ makes this distinction clearer: the original version scales almost linearly with $f$, whereas the active approach remains largely unaffected.

\begin{figure*}[htbp]
    \centering
     \vspace{-5em}
    \setlength{\subfigcapskip}{-5pt}
    \includegraphics[width=0.3\linewidth]{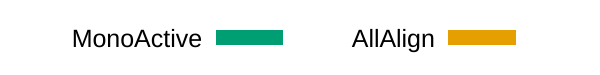}
    \includegraphics[width=0.3\linewidth]{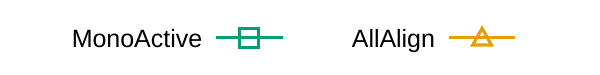}
    \includegraphics[width=0.3\linewidth]{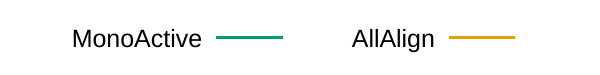}
    \vspace{-1.25em} \\
    
    \subfigure[\owt\eat{ ($k=64$)}]{
    \label{subfig:6a}
    \includegraphics[width=0.16\linewidth]{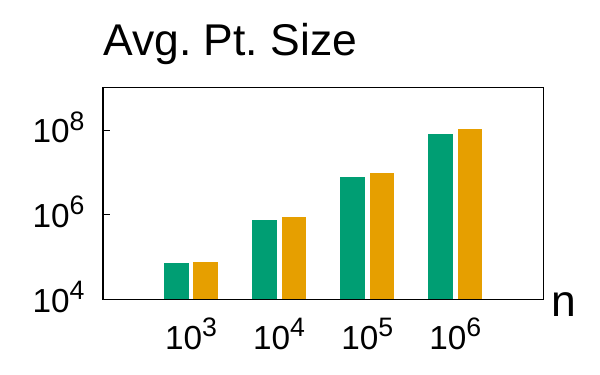}
    }
    \hspace*{-1em}
    \subfigure[\pan\eat{ ($k=64$)}]{
    \label{subfig:6b}
    \includegraphics[width=0.16\linewidth]{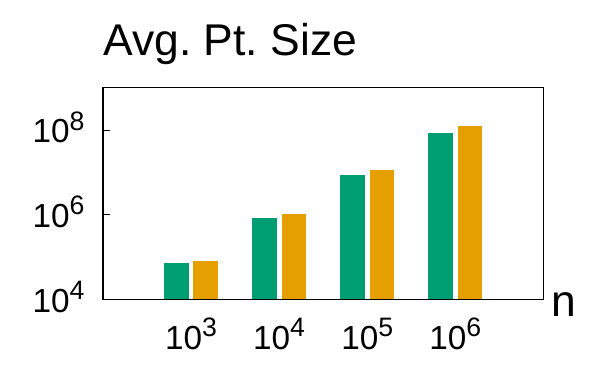}
    }
    \hspace*{-1em}
    \subfigure[\owt\eat{ ($k=64$)}]{
    \label{subfig:6c}
    \includegraphics[width=0.16\linewidth]{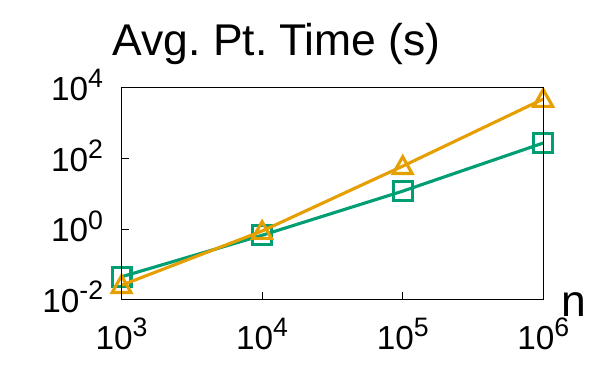}
    }
    \hspace*{-1em}
    \subfigure[\pan\eat{ ($k=64$)}]{
    \label{subfig:6d}
    \includegraphics[width=0.16\linewidth]{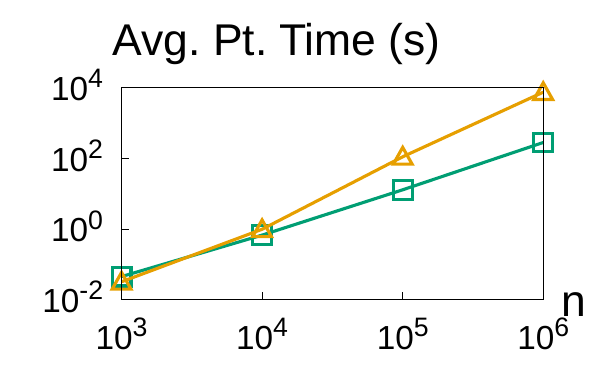}
    }
    \hspace*{-1em}
    \subfigure[\owt\eat{ ($k=64$)}]{
    \label{subfig:6e}
    \includegraphics[width=0.16\linewidth]{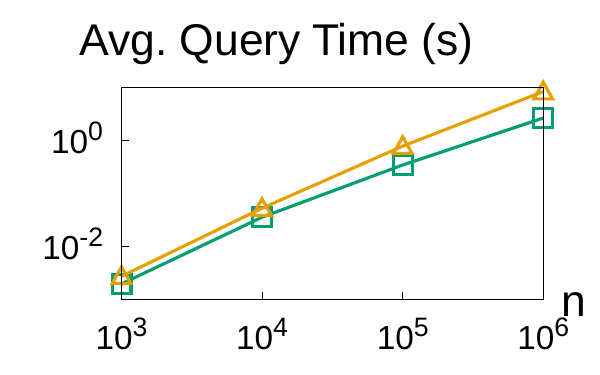}
    }
    \hspace*{-1em}
    \subfigure[\pan\eat{ ($k=64$)}]{
    \label{subfig:6f}
    \includegraphics[width=0.16\linewidth]{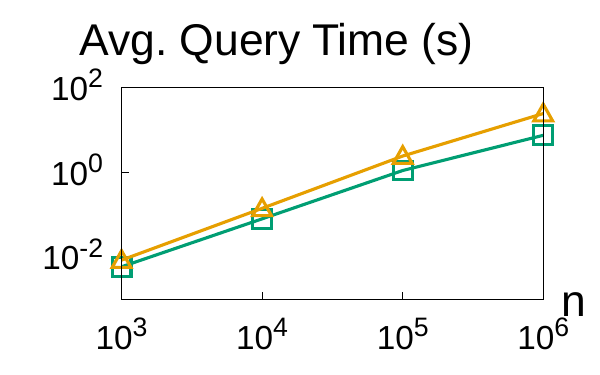}
    }
    \\ \vspace{-1.25em}
    \subfigure[\owt\eat{ ($k=64$)}]{
    \label{subfig:6g}
    \includegraphics[width=0.16\linewidth]{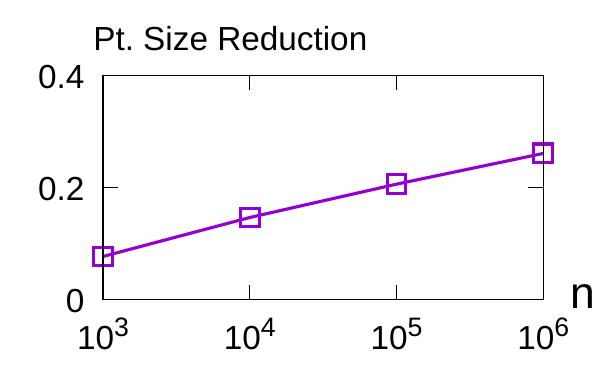}
    }
    \hspace*{-1em}
    \subfigure[\pan\eat{ ($k=64$)}]{
    \label{subfig:6h}
    \includegraphics[width=0.16\linewidth]{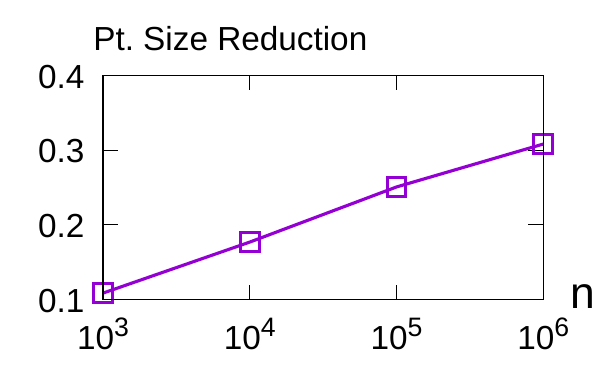}
    }
    \hspace*{-1em}
    \subfigure[\owt\eat{ ($k=64$)}]{
    \label{subfig:6i}
    \includegraphics[width=0.16\linewidth]{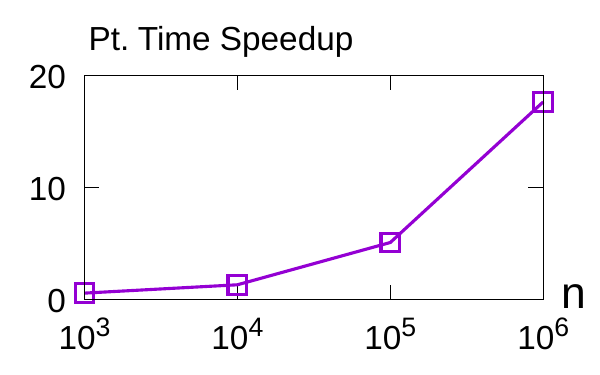}
    }
    \hspace*{-1em}
    \subfigure[\pan\eat{ ($k=64$)}]{
    \label{subfig:6j}
    \includegraphics[width=0.16\linewidth]{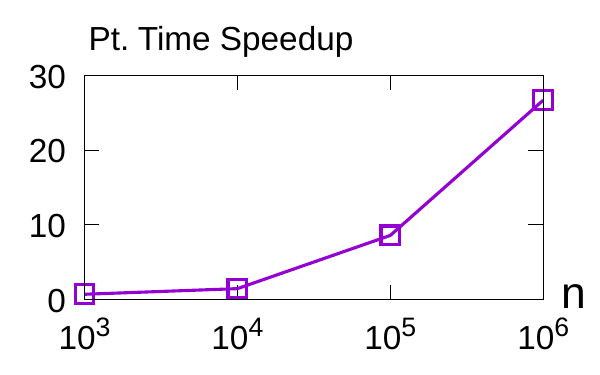}
    }
    \hspace*{-1em}
    \subfigure[\owt\eat{ ($k=64$)}]{
    \label{subfig:6k}
    \includegraphics[width=0.16\linewidth]{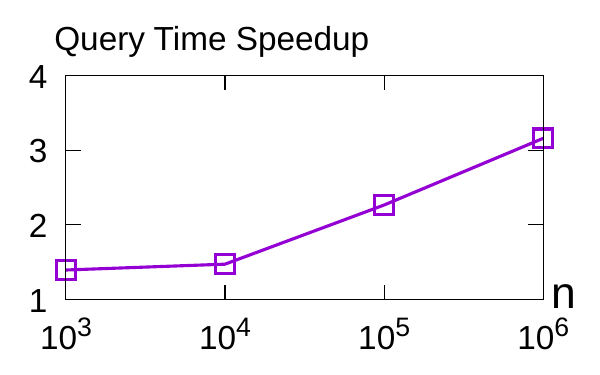}
    }
    \hspace*{-1em}
    \subfigure[\pan\eat{ ($k=64$)}]{
    \label{subfig:6l}
    \includegraphics[width=0.16\linewidth]{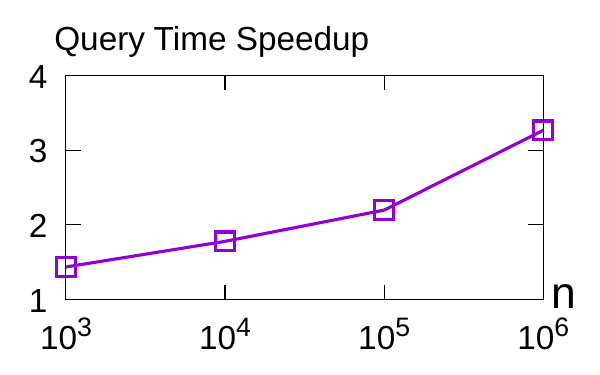}
    }\\
 \vspace{-1.25em}
%    \caption{Comparing with the state-of-the-art for multi-set Jaccard. \miao{k:  sketch size, n:  text length, f:  maximum token frequency.}}
%    \label{exp:allign}
%\end{figure*}
%\begin{figure*}[htbp]
    \subfigure[\owt ($\eat{k=64,}n=5000$)]{
    \label{subfig:7m}
    \includegraphics[width=0.16\linewidth]{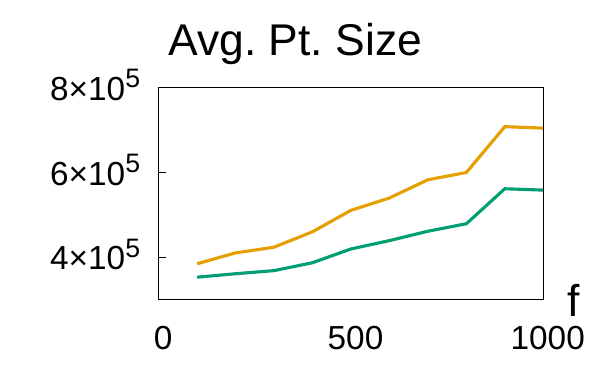}
    }
    \hspace*{-1em}
    \subfigure[\pan ($\eat{k=64,}n=5000$)]{
    \label{subfig:7n}
    \includegraphics[width=0.16\linewidth]{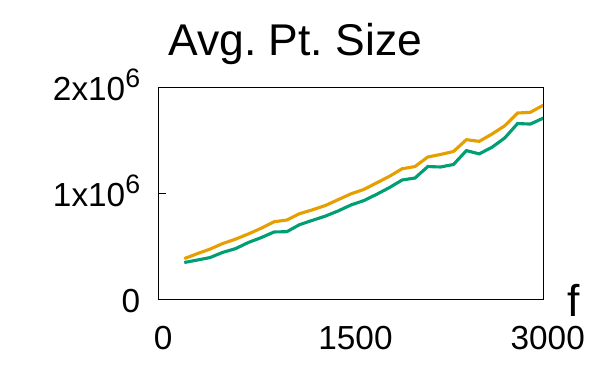}
    }
    \hspace*{-1em}
    \subfigure[\owt ($\eat{k=64,}n=5000$)]{
    \label{subfig:7o}
    \includegraphics[width=0.16\linewidth]{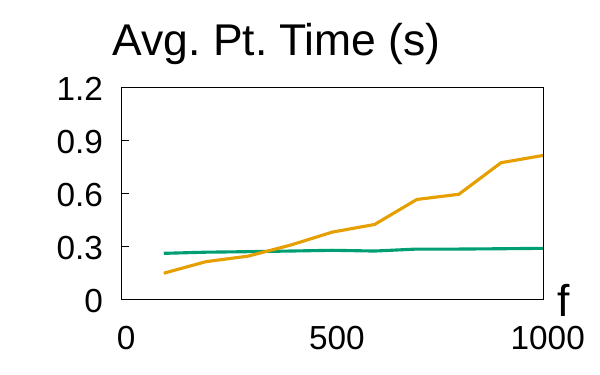}
    }
    \hspace*{-1em}
    \subfigure[\pan ($\eat{k=64,}n=5000$)]{
    \label{subfig:7p}
    \includegraphics[width=0.16\linewidth]{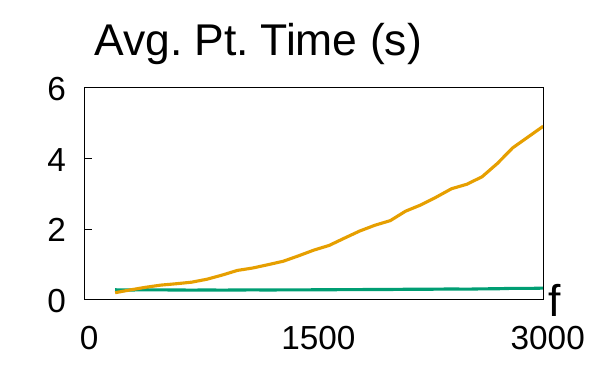}
    }
    \hspace*{-1em}
    \subfigure[\owt ($\eat{k=64,}n=5000$)]{
    \label{subfig:7q}
    \includegraphics[width=0.16\linewidth]{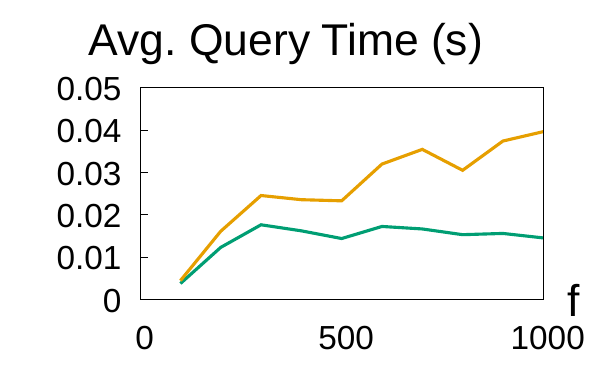}
    }
    \hspace*{-1em}
    \subfigure[\pan ($\eat{k=64,}n=5000$)]{
    \label{subfig:7r}
    \includegraphics[width=0.16\linewidth]{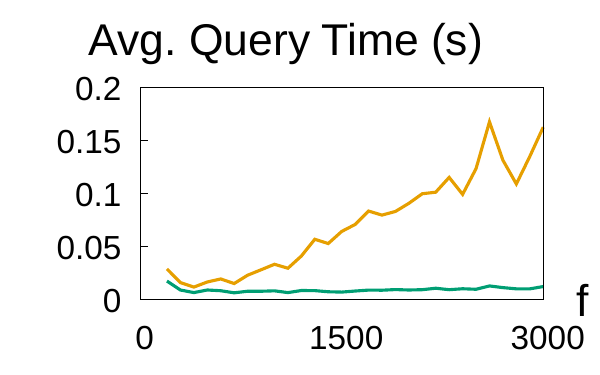}
    }
    \\ \vspace{-1em}
    \subfigure[\owt ($\eat{k=64,}n=5000$)]{
    \label{subfig:7s}
    \includegraphics[width=0.16\linewidth]{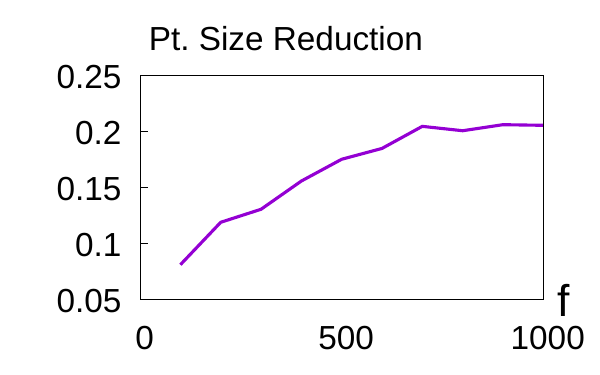}
    }
    \hspace*{-1em}
    \subfigure[\pan ($\eat{k=64,}n=5000$)]{
    \label{subfig:7t}
    \includegraphics[width=0.16\linewidth]{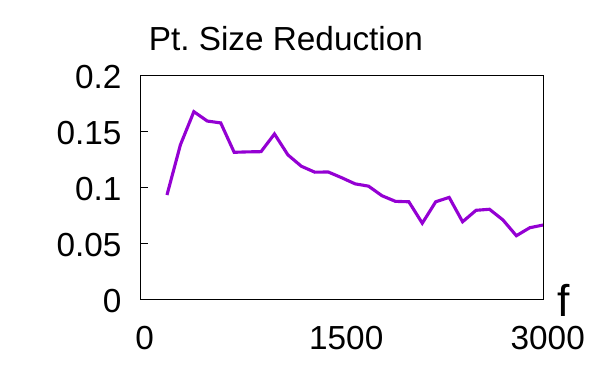}
    }
    \hspace*{-1em}
    \subfigure[\owt ($\eat{k=64,}n=5000$)]{
    \label{subfig:7u}
    \includegraphics[width=0.16\linewidth]{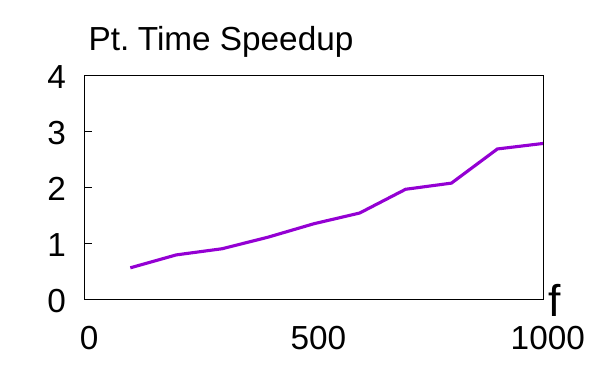}
    }
    \hspace*{-1em}
    \subfigure[\pan ($\eat{k=64,}n=5000$)]{
    \label{subfig:7v}
    \includegraphics[width=0.16\linewidth]{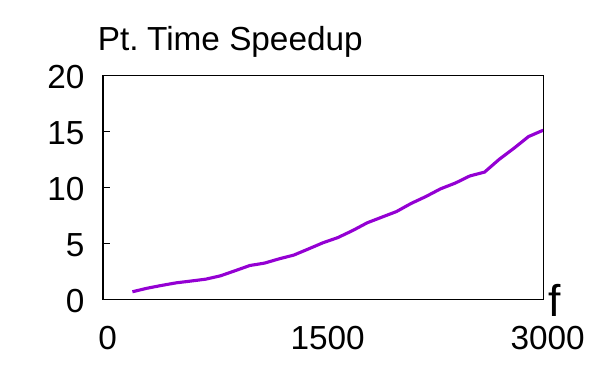}
    }
    \hspace*{-1em}
    \subfigure[\owt ($\eat{k=64,}n=5000$)]{
    \label{subfig:7w}
    \includegraphics[width=0.16\linewidth]{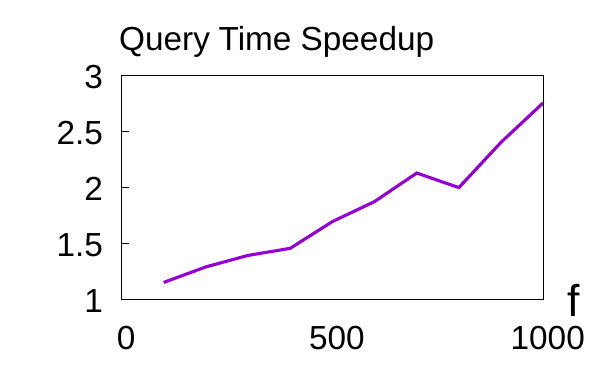}
    }
    \hspace*{-1em}
    \subfigure[\pan ($\eat{k=64,}n=5000$)]{
    \label{subfig:7x}
    \includegraphics[width=0.16\linewidth]{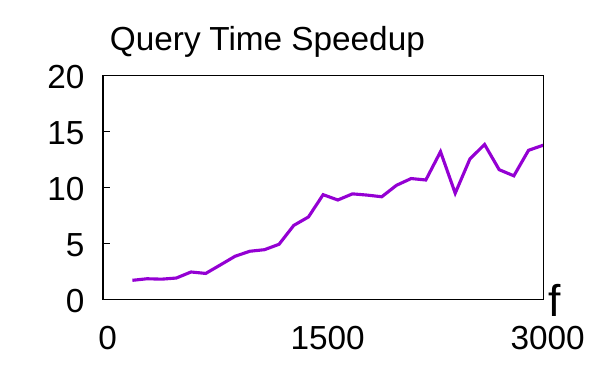}
    }
 \vspace{-1.5em}
    \caption{Comparing with the state-of-the-art for multi-set Jaccard similarity. $k=64$ in all experiments.\eat{ \miao{k:  sketch size, n:  text length, f:  maximum token frequency.}}} \vspace{-1.5em}
    \label{exp:allign}
\end{figure*}

\subsection{Evaluating Weighted Jaccard Similarity}

% To evaluate compact window generation in the weighted setting, we compared the Single-Column and Monotonic algorithms. We used raw counts for TF weight and inverse document frequency for IDF weight, which mitigates the impact of highly frequent words (such as “the”). The results are similar to the unweighted setting: as $k$ increases, both algorithms exhibit linear growth in the number of generated windows and runtime. As $n$ increases, Single-Column grows more rapidly . When $f$ increases, Single-Column decreases slightly while Monotonic increases, lower than Single-Column. Notably, introducing IDF reduces the effective frequency of words and thus decreases the total number of compact windows. For instance, “the” appears in nearly every document, effectively making its maximum frequency the second highest in practice.

This section evaluates near-duplicate text alignment under weighted Jaccard similarity. Our algorithm \mono is the only method that supports this problem. We evaluated four token weight functions $\wt{t}{x}$: (1) binary TF weight $[x\geq 1]$; (2) logarithmic TF weight $\log(x+1)$; (3) raw count TF weight $x$; and (4) squared TF weight $x^2$; all combined with unary IDF weight.

% The tokens are weighted by TF-IDF, where TF is the raw count (\ie token frequency) and IDF is the inverse document frequency, which mitigates the impact of highly frequent tokens (such as ``the''). Note that \ala does not support weighted tokens. Thus we only report the results of \mono and \scp. 
% <<<<<<< HEAD
% <<<<<<< HEAD
Figures~\ref{subfig:9a} and~\ref{subfig:9b} report the partition sizes (\ie, the number of compact windows generated) as the text length $n$ varies from $10^3$ to $10^6$. As observed, the squared weight consistently yields the largest partitions, followed by the raw count weight, whereas the binary weight produces the smallest partitions. For all weight functions, the partition size increases approximately linearly with $n$. For instance, when $n$ increases from 1{,}000 to 1 million on \pan, the partition sizes under the four weight functions increase from 64K, 67K, 71K, and 89K to 64M, 68M, 82M, and 158M, respectively. This is because the partition sizes of the binary, logarithmic, raw count, and squared weights are respectively $O(n)$, $O(n+n\log\log f)$, $O(n+n\log f)$, and $O(n+n\log f)$, according to Lemma~\ref{lemma:weightfunctions}. Although squared and raw count weights have the same complexity, the constant factor of squared weight is larger than that of raw count weight.

Figures~\ref{subfig:9c} and~\ref{subfig:9d} present the average partitioning time (\ie, the compact window generation time) under the four weight functions as the text length $n$ increases from $10^3$ to $10^6$. Among them, the binary weight consistently achieved the fastest partitioning time, followed by the logarithmic weight, while the squared weight resulted in the slowest runtime. All weight functions exhibit quasilinear growth with respect to $n$. For example, when $n$ increased from $10^3$ to $10^6$ on \pan, the partitioning times under the four weights increased from 0.47, 0.50, 0.48, and 0.49 seconds to 60.96, 110.99, 143.43, and 229.02 seconds, respectively. This aligns with our complexity analysis, which shows that the runtime of \mono scales with $|\aset(\textT)| \log n$ for all weight functions.

% In terms of runtime, \mono was always faster than \ala, up to $4.7 \times$. Figure~\ref{exp:subfigure_time_pan_k_3} shows that on \pan, when $n=1000, k=16$, \mono's runtime was 0.039s while \ala took 0.18s. The reasons are 1) \ala generates more compact windows than \mono, 2) \ala is  recursive, whereas \mono is non-recursive, and 3) \ala needs to compute the minimum hash value in the input rectangle, which requires more complex data structures. These  structures require significant preprocessing time, whereas \mono only needs to maintain a self-balanced binary tree. Note that the original paper of \ala uses linear scan to find the minimum hash value, which is even slower. We reimplemented \ala using a persistent segment tree. 

Figures~\ref{subfig:9e} and~\ref{subfig:9f} show the average query latency when varying the text length $n$ from $10^3$ to $10^6$. The queries are sampled from the same dataset of the data texts but are disjoint from them. Each query is a text of length $100$, and the same set of queries is used across all values of $n$. Similarly, the squared weight had the largest query latency, followed by raw count weight, while the binary weight had the smallest latency. This is because query processing cost grows with the number of compact windows sharing the same hash values between the query and the data texts (\aka collided compact windows). Weight functions that generate more compact windows are more likely to incur a larger number of such collisions, thereby increasing query latency.

% We also evaluate the impact of the maximum token frequency $f$. Figures~\ref{subfig:9g}-\ref{subfig:9h} show the average sizes of partitions generated of the four weight functions when varying $f$ and fixing $n=5000$ and $k = 64$.  As the maximum token frequency $f$ increased, the partition size of binary TF unchanged. The partition sizes of logarithmic TF grew modestly with the increase of $f$, while those of raw count TF and squared TF grew sublinearly with $f$. As an example, when increasing $f \in [901, 1000]$ to $f \in [2901, 3000]$, the partition sizes under the four TF weight functions change from 320K, 412K, 607K, 1M to 320K, 625K, 1.5M, 2.8M, respectively. This is because binary TF generates $n$ compact windows ($n \times k = $320K), logarithmic TF generates $O(n+n\log\log f)$ compact windows, raw count TF and squared TF generate $O(n+n\log f)$ compact windows.

Figures~\ref{subfig:9g}–\ref{subfig:9h} show the average partition sizes generated under the four weight functions, with $n = 5{,}000$ and $k = 64$ fixed while varying the maximum token frequency $f$. As $f$ increased, the partition size under the binary weight remained constant. The partition size under logarithmic increased modestly, whereas those under raw count and squared weights grew sublinearly with $f$. For example, when $f$ increased from the range $[901, 1000]$ to $[2901, 3000]$, the partition sizes under the binary, logarithmic, raw count, and squared weights increased from 320K, 412K, 607K, and 1M to 320K, 625K, 1.5M, and 2.8M, respectively. These results are consistent with our theoretical analysis. The numbers of compact windows generated by binary, logarithmic, raw count, and squared weights scale with $n$, $n\log\log f$, $n\log f$, and $n\log f$, respectively.

% binary weight yields $n \times k = 320$K compact windows, logarithmic weight produces $O(n + n \log \log f)$ compact windows, while both raw count and squared TF produce $O(n + n \log f)$ compact windows.

Figures~\ref{subfig:9i}–\ref{subfig:9j} report the average partitioning time of the four weight functions under varying $f$. The partitioning time either remained nearly constant or decreased slightly across all four weight functions. This is because a significant part of \mono's runtime is spent on sorting the active hash values, the number of which is proportional to the number of distinct tokens in the text. As $f$ increases while $n$ remains fixed, the number of distinct tokens decreases, thereby reducing the sorting cost. Finally, Figures~\ref{subfig:9k}–\ref{subfig:9l} show the average query latency of the four weight functions under varying $f$. The two datasets showed different behaviors. This is because the query time depends heavily on the query text (which determines the number of collided compact windows).

\subsection{Comparing with the State-of-the-Art}

This section compares our method \mono with \ala, the state-of-the-art method for multi-set Jaccard similarity. 

Figures~\ref{subfig:6a} and~\ref{subfig:6b} show, when varying the text length $n$ from $10^3$ to $10^6$, the average number of compact windows generated by \mono and \ala (\ie partition sizes, which are proportional to their index sizes) on two datasets \owt and \pan, while Figures~\ref{subfig:6g} and~\ref{subfig:6h} show the reduction ratio (\ie $1-\frac{\mono}{\ala}$). As shown, \mono consistently generated less compact windows than \ala, up to 30.82\%. More importantly, the gap between the number of compact windows for \mono and \ala steadily widens as $n$ increases. For example, when $n$ increased from 1000 to 1 million, the gap grew from 10.86\% to 30.82\% on \pan and from 7.69\% to 26.08\% on \owt. This is because \ala's recursive approach divides the rectangles into multiple rectangular regions, imposing early boundaries that may split what is a single compact window in \mono into multiple. The number of compact windows generated by \mono grew linearly with $n$, which is consistent with our complexity analysis.

% This is because the maximum token frequency $f$ increases with $n$. As the maximum frequency $f$ increased, the performance of  the partition size and partition time of \mono were largely unaffected, whereas those of \ala slightly grew. This is because higher frequencies lead to deeper recursions and more early boundaries, which in turn produced more compact windows and took longer time. %, a trend also visible in the experiment that tracks the number of windows as $f$ grows.

% More importantly, the gap between the number of compact windows for \mono and \ala steadily widens as $n$ increases, as shown in Figures~\ref{subfig:6g} and~\ref{subfig:6h}. 

% \mono is up to 4$\times$ faster than \ala in indexing and the gap widens as the text length increases. Specifically, 

Figures~\ref{subfig:6c} and~\ref{subfig:6d} show, when varying the text length $n$ from $10^3$ to $10^6$, the average compact window generation time (\ie partition time, which is proportional to the index time) of $\ala$ and \mono, respectively, while Figures~\ref{subfig:6i} and~\ref{subfig:6j} show the speedup of \mono over \ala. \mono was significantly faster than \ala. For example, on the \pan dataset, when $n=10^6$, \ala took 7581s on average to generate the compact windows, whereas \mono only spent 284s, achieving a speedup of 26.7$\times$. This is because \ala produces more compact windows than \mono and \ala is recursive whereas \mono is non-recursive. Moreover, the speedup of \mono over \ala increased as $n$ grew. For example, on the \pan dataset, when $n$ increased from $10^4$ to $10^6$, the speedup improved from 1.46 to 26.7. A similar trend is observed on \owt.

% In terms of runtime, \mono was always faster than \ala, up to $4.7 \times$. Figure~\ref{exp:subfigure_time_pan_k_3} shows that on \pan, when $n=1000, k=16$, \mono's runtime was 0.039s while \ala took 0.18s. The reasons are 1) \ala generates more compact windows than \mono, 2) \ala is  recursive, whereas \mono is non-recursive, and 3) \ala needs to compute the minimum hash value in the input rectangle, which requires more complex data structures. These  structures require significant preprocessing time, whereas \mono only needs to maintain a self-balanced binary tree. Note that the original paper of \ala uses linear scan to find the minimum hash value, which is even slower. We reimplemented \ala using a persistent segment tree. 

Figures~\ref{subfig:6e} and~\ref{subfig:6f} show the average query latency of the two methods when varying the text length $n$ from $10^3$ to $10^6$, while Figures~\ref{subfig:6k} and~\ref{subfig:6l} show the speedup of \mono over \ala. The queries are generated in the same way as described earlier. The query latency of \mono was up to 3.15$\times$ faster than that of \ala, though the number of compact windows was only up to $30.82\%$ less. For example, on the \pan dataset, when $n$ increased from $10^3$ to $10^6$, the speedup improved from 1.4 to 3.15. A similar trend is observed on \owt. This is because, although the number of collided compact windows grows linearly with the number of generated compact windows in expectation, the cost of query processing scales quadratically with the number of collided windows in the worst case. Therefore, reducing the number of compact windows can lead to a significant speedup in query time.

We also evaluate the impact of the maximum token frequency $f$. Figures~\ref{subfig:7m}–\ref{subfig:7r} present the number of compact windows generated, the compact window generation time, and the query latency of the two methods, under varying $f$ with fixed text length $n = 5000$. Figures~\ref{subfig:7s}–\ref{subfig:7x} further report the relative improvements. As $f$ increases, the generation time of \mono remains largely stable, whereas that of \ala increases significantly. For example, on the \pan dataset, when $f$ increases from $600$ to $3000$, the speedup improved from 1.41 to 16.26. In addition, the gap in the number of compact windows produced by the two methods first grows and then decreases  with larger $f$. This trend arises because higher token frequencies lead to deeper recursion and more early boundaries in \ala, which fragment compact windows and increase computational cost. However, as $f$ approaches $n$, the text instances approach the worst-case scenario and both methods produce a large number of compact windows. The gap of the query latency increases with the increase of $f$. This is because the collided windows are dominated by tokens with the highest frequencies. In \ala, the high-frequency tokens trigger deeper recursion, resulting in more finely fragmented compact windows. These fragments increase the number of collided compact windows and the number of iterations in the inner loop. As a result, the query latency grows superlinearly with the number of collided compact windows.

% We also evaluate the impact of the maximum token frequency $f$. Figures~\ref{subfig:7m}-\ref{subfig:7x} show, when varying $f$ and fixing $n=5000$, the number of compact window generated, the generation time, and the query latency of the two methods, while Figures~\ref{subfig:7s}-\ref{subfig:7x} show the reduction ration of the number of compact windows, the speedup of compact window generation and query latency of the two methods. As the maximum token frequency $f$ increased, the compact window generation time of \mono remained largely stable, while that of \ala increased significantly. Additionally, the difference in the number of compact windows produced by the two methods grew moderately with larger $f$. This trend is due to higher token frequencies causing deeper recursions and more early boundaries in \ala, resulting in both a larger number of compact windows and longer processing times. However, this gap saturated as $f$ approached $n$, since the text instances approached the worst-case scenario and the performance of both methods degraded accordingly.

\eetitle{Scalability.} As shown in Figures~\ref{subfig:6g}-\ref{subfig:6l}, the performance gain of \mono over \ala increased as the text length $n$ grew in terms of partition size, partition generation time, and query latency. Thus \mono scales better than \ala.  This is attributed to the complexity guarantees of our algorithm.

\section{Related Work}\label{sec:related}

% \stitle{Near-Duplicate Text Detection.}  Both near-duplicate text alignment~\cite{DBLP:conf/usenix/Manber94,DBLP:journals/cn/BroderGMZ97,DBLP:conf/www/HamidBCH09,DBLP:conf/sigmod/BrinDG95,DBLP:conf/sigir/SeoC08,DBLP:conf/www/KimCT09,DBLP:books/aw/Baeza-YatesR99, DBLP:journals/jasis/HoadZ03} and near-duplicate text detection (\aka similarity search)~\cite{DBLP:conf/kdd/KolczCA04,DBLP:journals/pvldb/XiaoWL08,DBLP:conf/www/BayardoMS07,DBLP:conf/www/XiaoWLY08,DBLP:journals/tois/ChowdhuryFGM02,setjoin,DBLP:conf/cikm/ConradGS03,pvldb12-passjoin} have been extensively studied due to their importance in text processing. However, the two problems are rather different. While near-duplicate text detection focuses on if two entire texts are similar, near-duplicate text alignment concerns the subsequences in the texts. There are both exact algorithms~\cite{setjoin,pvldb12-passjoin} and approximate algorithms~\cite{lsh,e2lsh} for near-duplicate text detection. Locality sensitive hashing (LSH)~\cite{lsh,e2lsh,qalsh,c2lsh,lsbtree,lsbtree2,srs,falconn,euclideanlsh,hamminglsh,multiprobelsh} is one of the most well-known techniques for near-duplicate text detection. However, it is non-trivial to apply LSH for near-duplicate text alignment. 

\stitle{Near-Duplicate Text Alignment.} Most of existing methods for near-duplicate text alignment rely on rule-based heuristics and the ``seeding-extension-filtering'' framework~\cite{DBLP:conf/usenix/Manber94,DBLP:journals/cn/BroderGMZ97,DBLP:conf/www/HamidBCH09,DBLP:conf/sigmod/BrinDG95,DBLP:conf/sigir/SeoC08,DBLP:conf/www/KimCT09,DBLP:books/aw/Baeza-YatesR99,DBLP:journals/jasis/HoadZ03}. They first find ``seed matches'' between the data texts and the query. Various kinds of seeds have been proposed, such as fingerprints~\cite{DBLP:conf/clef/PotthastBESR10}, super-shingles~\cite{broder1997syntactic}, 0 mod p~\cite{DBLP:conf/usenix/Manber94}, fixed-length windows~\cite{DBLP:conf/sigmod/WangXQWZI16}, q-grams~\cite{DBLP:conf/sigmod/SchleimerWA03}, and sentences~\cite{DBLP:conf/clef/Sanchez-PerezGS15}. Then they extend the seed matches as far as possible to form candidates. Finally, candidates failing predefined criteria (e.g., length or overlap thresholds) are filtered. However, these heuristics are highly sensitive to the hard-to-tune hyper-parameters~\cite{DBLP:journals/csur/FoltynekMG20} such as the granularity of the seeds and various kinds of thresholds~\cite{DBLP:conf/clef/Sanchez-PerezGS15,DBLP:conf/semeval/AgirreBCDGMRW16}. They also lack accuracy guarantees.

% For example, coarse-grained seeds may miss valid matches, while fine-grained seeds generate an overwhelming number of seed matches, overloading the extension and filtering stages~\cite{DBLP:conf/clef/Sanchez-PerezGS15,DBLP:conf/semeval/AgirreBCDGMRW16}. 

A few recent works propose to use the min-hash techniques~\cite{minhash} for near-duplicate text alignment~\cite{allign,txtalign,llmalign,DBLP:journals/pacmmod/PengZD24}. They introduce the concept of ``compact windows'' to group nearby subsequences sharing the same min-hash. It has been shown that \miao{when duplicate tokens have the same hash value}, the $O(kn^2)$ min-hashes in a text with $n$ tokens can be compressed in compact windows using $O(kn)$ space and $O(kn \log n)$ time, where $k$ is the sketch size~\cite{allign}. \miao{Along this line,} an algorithm is developed to group the $O(n^2)$ bottom-$k$ sketches in a text with $n$ tokens into compact windows using $O(nk^2)$ space and $O(n\log n + nk)$ time~\cite{txtalign}. To further reduce the space cost, another algorithm is designed to group the $O(kn^2)$ one-permutation hashing~\cite{DBLP:conf/nips/0001OZ12} min-hash sketches into compact windows using $O(n+k)$ space and $O(n\log n +k)$ time~\cite{DBLP:journals/pacmmod/PengZD24}. These algorithms are used to evaluate the memorization behaviour in large language models (LLMs), revealing that up to 10\% of texts generated by GPT-2~\cite{radford2019language} had near-duplicates in its training data. It also shows that the min-hash sketches of all subsequences no shorter than $t$ tokens in a text with $n$ tokens can be grouped into $O(\frac{n}{t})$ compact windows on average~\cite{llmalign}. \miao{However, none of these works} can deal with weighted min-hash (\ie consistent weighted sampling~\cite{ioffe2010improved}).

\stitle{Min-Hash Sketch.} Min-hash was originally introduced in statistics for coordinated sampling~\cite{brewer1972selecting} and later adapted for database applications by Flajolet and Martin~\cite{flajolet1985probabilistic}. Broder~\cite{DBLP:journals/cn/BroderGMZ97} employed the min-hash sketch to detect near-duplicate web pages. Several variants of the min-hash sketch have been proposed to improve the sketching time of the classic min-hash sketch, including the bottom-$k$ sketch~\cite{DBLP:conf/stoc/Thorup13}, one-permutation hashing (OPH)~\cite{DBLP:conf/nips/0001OZ12}, and fast similarity sketch~\cite{dahlgaard2017fast}. The number of min-hashes generated by one-permutation hashing is not fixed. A few works propose techniques to address this issue~\cite{shrivastava2014densifying, shrivastava2014improved, shrivastava2017optimal}.

\stitle{Weighted Jaccard Similarity Estimation.} Many techniques have been proposed to estimate the weighted Jaccard similarity~\cite{ertl2018bagminhash,manasse2010consistent,ioffe2010improved,wu2018improved}. Specifically, 
\cite{gollapudi2006exploiting} extends the classic min-hash to estimate the multi-set Jaccard similarity. A method dealing with integer weights is mentioned in~\cite{charikar2002similarity}. It was extended by~\cite{chum2008near} to support a more general weight function. Consistent weighted sampling (CWS) is first proposed in~\cite{manasse2010consistent} to estimate weighted Jaccard similarity. Ioffe proposes improved consistent weighted sampling, which simplified CWS and guarantees worst-case constant time for each non-zero weight~\cite{ioffe2010improved}. Shrivastava proposes to use rejected sampling to estimate the weigthed Jaccard Similarity~\cite{shrivastava2016simple}.

\section{Conclusion}\label{sec:conclude}

In conclusion, this paper extends near-duplicate text alignment to support weighted Jaccard similarity by leveraging consistent weighted sampling. We introduce \mono, an efficient and theoretically optimal algorithm for grouping subsequences based on their consistent weighted samplings. Our analysis establishes tight bounds on the number of groups generated, and our experiments demonstrate substantial improvements over state-of-the-art methods in both index time and index size. These results highlight the practicality and scalability of our approach for real-world text alignment tasks where token weights matter.

\balance

\bibliographystyle{abbrv}

\begin{thebibliography}{10}

\bibitem{DBLP:conf/semeval/AgirreBCDGMRW16}
E.~Agirre, C.~Banea, D.~M. Cer, M.~T. Diab, A.~Gonzalez{-}Agirre, R.~Mihalcea,
  G.~Rigau, and J.~Wiebe.
\newblock Semeval-2016 task 1: Semantic textual similarity, monolingual and
  cross-lingual evaluation.
\newblock In {\em SEMEVAL}, pages 497--511. The Association for Computer
  Linguistics, 2016.

\bibitem{ALTSCHUL1990403}
S.~F. Altschul, W.~Gish, W.~Miller, E.~W. Myers, and D.~J. Lipman.
\newblock Basic local alignment search tool.
\newblock {\em Journal of Molecular Biology}, 215(3):403--410, 1990.

\bibitem{DBLP:books/aw/Baeza-YatesR99}
R.~A. Baeza{-}Yates and B.~A. Ribeiro{-}Neto.
\newblock {\em Modern Information Retrieval}.
\newblock {ACM} Press / Addison-Wesley, 1999.

\bibitem{brewer1972selecting}
K.~R.~W. Brewer, L.~J. Early, and S.~F. Joyce.
\newblock Selecting several samples from a single population.
\newblock {\em Australian Journal of Statistics}, 14(3):231--239, 1972.

\bibitem{DBLP:conf/sigmod/BrinDG95}
S.~Brin, J.~Davis, and H.~Garcia{-}Molina.
\newblock Copy detection mechanisms for digital documents.
\newblock In {\em SIGMOD}, pages 398--409. {ACM} Press, 1995.

\bibitem{minhash}
A.~Z. Broder.
\newblock On the resemblance and containment of documents.
\newblock In {\em SEQUENCES}, pages 21--29. {IEEE}, 1997.

\bibitem{DBLP:journals/cn/BroderGMZ97}
A.~Z. Broder, S.~C. Glassman, M.~S. Manasse, and G.~Zweig.
\newblock Syntactic clustering of the web.
\newblock {\em Comput. Networks}, 29(8-13):1157--1166, 1997.

\bibitem{broder1997syntactic}
A.~Z. Broder, S.~C. Glassman, M.~S. Manasse, and G.~Zweig.
\newblock Syntactic clustering of the web.
\newblock {\em Computer networks and ISDN systems}, 29(8-13):1157--1166, 1997.

\bibitem{quantifymemo}
N.~Carlini, D.~Ippolito, M.~Jagielski, K.~Lee, F.~Tram{\`{e}}r, and C.~Zhang.
\newblock Quantifying memorization across neural language models.
\newblock {\em CoRR}, abs/2202.07646, 2022.

\bibitem{charikar2002similarity}
M.~S. Charikar.
\newblock Similarity estimation techniques from rounding algorithms.
\newblock In {\em Proceedings of the thiry-fourth annual ACM symposium on
  Theory of computing}, pages 380--388, 2002.

\bibitem{chum2008near}
O.~Chum, J.~Philbin, A.~Zisserman, et~al.
\newblock Near duplicate image detection: Min-hash and tf-idf weighting.
\newblock In {\em Bmvc}, volume 810, pages 812--815, 2008.

\bibitem{dahlgaard2017fast}
S.~Dahlgaard, M.~B.~T. Knudsen, and M.~Thorup.
\newblock Fast similarity sketching.
\newblock In {\em 2017 IEEE 58th Annual Symposium on Foundations of Computer
  Science (FOCS)}, pages 663--671. IEEE, 2017.

\bibitem{DBLP:conf/uss/DingZ0M23}
H.~Ding, J.~Zhai, D.~Deng, and S.~Ma.
\newblock The case for learned provenance graph storage systems.
\newblock In J.~A. Calandrino and C.~Troncoso, editors, {\em 32nd {USENIX}
  Security Symposium, {USENIX} Security 2023, Anaheim, CA, USA, August 9-11,
  2023}, pages 3277--3294. {USENIX} Association, 2023.

\bibitem{ertl2018bagminhash}
O.~Ertl.
\newblock Bagminhash-minwise hashing algorithm for weighted sets.
\newblock In {\em Proceedings of the 24th ACM SIGKDD International Conference
  on Knowledge Discovery \& Data Mining}, pages 1368--1377, 2018.

\bibitem{allign}
W.~Feng and D.~Deng.
\newblock Allign: Aligning all-pair near-duplicate passages in long texts.
\newblock In {\em SIGMOD}, pages 541--553. {ACM}, 2021.

\bibitem{flajolet1985probabilistic}
P.~Flajolet and G.~N. Martin.
\newblock Probabilistic counting algorithms for data base applications.
\newblock {\em Journal of computer and system sciences}, 31(2):182--209, 1985.

\bibitem{DBLP:journals/csur/FoltynekMG20}
T.~Folt{\'{y}}nek, N.~Meuschke, and B.~Gipp.
\newblock Academic plagiarism detection: {A} systematic literature review.
\newblock {\em {ACM} Comput. Surv.}, 52(6):112:1--112:42, 2020.

\bibitem{n-gram}
A.~Franz and T.~Brants.
\newblock All our n-gram are belong to you.
\newblock {\em Google Machine Translation Team}, 20, 2006.

\bibitem{bpe}
P.~Gage.
\newblock A new algorithm for data compression.
\newblock {\em C Users J.}, 12(2):23–38, feb 1994.

\bibitem{Gokaslan2019OpenWeb}
A.~Gokaslan and V.~Cohen.
\newblock Openwebtext corpus.

\bibitem{gollapudi2006exploiting}
S.~Gollapudi and R.~Panigrahy.
\newblock Exploiting asymmetry in hierarchical topic extraction.
\newblock In {\em Proceedings of the 15th ACM international conference on
  Information and knowledge management}, pages 475--482, 2006.

\bibitem{DBLP:conf/www/HamidBCH09}
O.~A. Hamid, B.~Behzadi, S.~Christoph, and M.~R. Henzinger.
\newblock Detecting the origin of text segments efficiently.
\newblock In {\em WWW}, pages 61--70. {ACM}, 2009.

\bibitem{DBLP:journals/jasis/HoadZ03}
T.~C. Hoad and J.~Zobel.
\newblock Methods for identifying versioned and plagiarized documents.
\newblock {\em J. Assoc. Inf. Sci. Technol.}, 54(3):203--215, 2003.

\bibitem{DBLP:conf/icdm/Ioffe10}
S.~Ioffe.
\newblock Improved consistent sampling, weighted minhash and {L1} sketching.
\newblock In {\em ICDM}, pages 246--255. {IEEE} Computer Society, 2010.

\bibitem{ioffe2010improved}
S.~Ioffe.
\newblock Improved consistent sampling, weighted minhash and l1 sketching.
\newblock In {\em 2010 IEEE international conference on data mining}, pages
  246--255. IEEE, 2010.

\bibitem{DBLP:conf/www/KimCT09}
J.~W. Kim, K.~S. Candan, and J.~Tatemura.
\newblock Efficient overlap and content reuse detection in blogs and online
  news articles.
\newblock In {\em WWW}, pages 81--90. {ACM}, 2009.

\bibitem{betterlm}
K.~Lee, D.~Ippolito, A.~Nystrom, C.~Zhang, D.~Eck, C.~Callison-Burch, and
  N.~Carlini.
\newblock Deduplicating training data makes language models better.
\newblock In {\em ACL}, pages 8424--8445, 2022.

\bibitem{DBLP:conf/nips/0001OZ12}
P.~Li, A.~B. Owen, and C.~Zhang.
\newblock One permutation hashing.
\newblock In {\em NIPS}, pages 3122--3130, 2012.

\bibitem{DBLP:conf/acl/Magar022}
I.~Magar and R.~Schwartz.
\newblock Data contamination: From memorization to exploitation.
\newblock In {\em Proceedings of the 60th Annual Meeting of the Association for
  Computational Linguistics (Volume 2: Short Papers), {ACL} 2022, Dublin,
  Ireland, May 22-27, 2022}, pages 157--165. Association for Computational
  Linguistics, 2022.

\bibitem{manasse2010consistent}
M.~Manasse, F.~McSherry, and K.~Talwar.
\newblock Consistent weighted sampling.
\newblock {\em Unpublished technical report) http://research. microsoft.
  com/en-us/people/manasse}, 2, 2010.

\bibitem{DBLP:conf/usenix/Manber94}
U.~Manber.
\newblock Finding similar files in a large file system.
\newblock In {\em {USENIX} Winter 1994 Technical Conference, San Francisco,
  California, USA, January 17-21, 1994, Conference Proceedings}, pages 1--10.
  {USENIX} Association, 1994.

\bibitem{DBLP:conf/hicss/PasiB23}
G.~Pasi and G.~Bordogna.
\newblock Introduction to the minitrack on intelligent information access and
  retrieval.
\newblock In T.~X. Bui, editor, {\em 56th Hawaii International Conference on
  System Sciences, {HICSS} 2023, Maui, Hawaii, USA, January 3-6, 2023}, pages
  4169--4170. ScholarSpace, 2023.

\bibitem{llmalign}
Z.~Peng, Z.~Wang, and D.~Deng.
\newblock Near-duplicate sequence search at scale for large language model
  memorization evaluation.
\newblock {\em Proc. {ACM} Manag. Data}, 1(2):179:1--179:18, 2023.

\bibitem{DBLP:journals/pacmmod/PengZD24}
Z.~Peng, Y.~Zhang, and D.~Deng.
\newblock Near-duplicate text alignment with one permutation hashing.
\newblock {\em Proc. {ACM} Manag. Data}, 2(4):200:1--200:26, 2024.

\bibitem{DBLP:conf/clef/PotthastBESR10}
M.~Potthast, A.~Barr{\'{o}}n{-}Cede{\~{n}}o, A.~Eiselt, B.~Stein, and P.~Rosso.
\newblock Overview of the 2nd international competition on plagiarism
  detection.
\newblock In {\em {CLEF} 2010 LABs and Workshops, Notebook Papers}, volume 1176
  of {\em {CEUR} Workshop Proceedings}. CEUR-WS.org, 2010.

\bibitem{DBLP:conf/clef/PotthastEBSR11}
M.~Potthast, A.~Eiselt, A.~Barr{\'{o}}n{-}Cede{\~{n}}o, B.~Stein, and P.~Rosso.
\newblock Overview of the 3rd international competition on plagiarism
  detection.
\newblock In {\em {CLEF} 2011 Labs and Workshop, Notebook Papers}, volume 1177
  of {\em {CEUR} Workshop Proceedings}. CEUR-WS.org, 2011.

\bibitem{DBLP:conf/clef/PotthastGHKMOTBGRS12}
M.~Potthast, T.~Gollub, M.~Hagen, J.~Kiesel, M.~Michel, A.~Oberl{\"{a}}nder,
  M.~Tippmann, A.~Barr{\'{o}}n{-}Cede{\~{n}}o, P.~Gupta, P.~Rosso, and
  B.~Stein.
\newblock Overview of the 4th international competition on plagiarism
  detection.
\newblock In {\em {CLEF} 2012 Evaluation Labs and Workshop}, volume 1178 of
  {\em {CEUR} Workshop Proceedings}. CEUR-WS.org, 2012.

\bibitem{DBLP:conf/clef/PotthastHBBTRS14}
M.~Potthast, M.~Hagen, A.~Beyer, M.~Busse, M.~Tippmann, P.~Rosso, and B.~Stein.
\newblock Overview of the 6th international competition on plagiarism
  detection.
\newblock In {\em Working Notes for {CLEF} 2014 Conference}, volume 1180 of
  {\em {CEUR} Workshop Proceedings}, pages 845--876. CEUR-WS.org, 2014.

\bibitem{DBLP:conf/clef/PotthastHGTKRSS13}
M.~Potthast, M.~Hagen, T.~Gollub, M.~Tippmann, J.~Kiesel, P.~Rosso,
  E.~Stamatatos, and B.~Stein.
\newblock Overview of the 5th international competition on plagiarism
  detection.
\newblock In {\em Working Notes for {CLEF} 2013 Conference}, volume 1179 of
  {\em {CEUR} Workshop Proceedings}. CEUR-WS.org, 2013.

\bibitem{radford2019language}
A.~Radford, J.~Wu, R.~Child, D.~Luan, D.~Amodei, I.~Sutskever, et~al.
\newblock Language models are unsupervised multitask learners.
\newblock {\em OpenAI blog}, 1(8):9, 2019.

\bibitem{DBLP:conf/clef/Sanchez-PerezGS15}
M.~A. Sanchez{-}Perez, A.~F. Gelbukh, and G.~Sidorov.
\newblock Adaptive algorithm for plagiarism detection: The best-performing
  approach at {PAN} 2014 text alignment competition.
\newblock In {\em 6th International Conference of the {CLEF} Association},
  volume 9283 of {\em Lecture Notes in Computer Science}, pages 402--413.
  Springer, 2015.

\bibitem{DBLP:conf/sigmod/SchleimerWA03}
S.~Schleimer, D.~S. Wilkerson, and A.~Aiken.
\newblock Winnowing: Local algorithms for document fingerprinting.
\newblock In {\em SIGMOD}, pages 76--85. {ACM}, 2003.

\bibitem{DBLP:conf/sigir/SeoC08}
J.~Seo and W.~B. Croft.
\newblock Local text reuse detection.
\newblock In {\em SIGIR}, pages 571--578. {ACM}, 2008.

\bibitem{shrivastava2016simple}
A.~Shrivastava.
\newblock Simple and efficient weighted minwise hashing.
\newblock {\em Advances in Neural Information Processing Systems}, 29, 2016.

\bibitem{shrivastava2017optimal}
A.~Shrivastava.
\newblock Optimal densification for fast and accurate minwise hashing.
\newblock In {\em International Conference on Machine Learning}, pages
  3154--3163. PMLR, 2017.

\bibitem{shrivastava2014densifying}
A.~Shrivastava and P.~Li.
\newblock Densifying one permutation hashing via rotation for fast near
  neighbor search.
\newblock In {\em International Conference on Machine Learning}, pages
  557--565. PMLR, 2014.

\bibitem{shrivastava2014improved}
A.~Shrivastava and P.~Li.
\newblock Improved densification of one permutation hashing.
\newblock {\em arXiv preprint arXiv:1406.4784}, 2014.

\bibitem{singhal2001modern}
A.~Singhal et~al.
\newblock Modern information retrieval: A brief overview.
\newblock {\em IEEE Data Eng. Bull.}, 24(4):35--43, 2001.

\bibitem{DBLP:conf/stoc/Thorup13}
M.~Thorup.
\newblock Bottom-k and priority sampling, set similarity and subset sums with
  minimal independence.
\newblock In {\em STOC}, pages 371--380. {ACM}, 2013.

\bibitem{DBLP:conf/emnlp/VuHHS23}
T.~Vu, X.~He, G.~Haffari, and E.~Shareghi.
\newblock Koala: An index for quantifying overlaps with pre-training corpora.
\newblock In Y.~Feng and E.~Lefever, editors, {\em Proceedings of the 2023
  Conference on Empirical Methods in Natural Language Processing, {EMNLP} 2023
  - System Demonstrations, Singapore, December 6-10, 2023}, pages 90--98.
  Association for Computational Linguistics, 2023.

\bibitem{DBLP:conf/sigmod/WangXQWZI16}
P.~Wang, C.~Xiao, J.~Qin, W.~Wang, X.~Zhang, and Y.~Ishikawa.
\newblock Local similarity search for unstructured text.
\newblock In {\em SIGMOD}, pages 1991--2005. {ACM}, 2016.

\bibitem{txtalign}
Z.~Wang, C.~Zuo, and D.~Deng.
\newblock Txtalign: Efficient near-duplicate text alignment search via bottom-k
  sketches for plagiarism detection.
\newblock In {\em SIGMOD}, pages 1146--1159. {ACM}, 2022.

\bibitem{whitespace_splitted_tokens}
J.~J. Webster and C.~Kit.
\newblock Tokenization as the initial phase in nlp.
\newblock In {\em COLING 1992 volume 4: The 14th international conference on
  computational linguistics}, 1992.

\bibitem{wu2018improved}
W.~Wu, B.~Li, L.~Chen, C.~Zhang, and S.~Y. Philip.
\newblock Improved consistent weighted sampling revisited.
\newblock {\em IEEE Transactions on Knowledge and Data Engineering},
  31(12):2332--2345, 2018.

\end{thebibliography}
% The .bbl file is included here instead of using \bibliography command

\newpage
\appendix
\section{Correctness of Monotonic Partitioning}\label{sec:proofs}

This section serves as a proof of Theorem~\ref{theorem:uppernlogf}. Let $n = |\textT|$. 

\stitle{Augmented Skyline.} Since every key $(p, q) \in \key(\textT)$ has $(p, q) \in [1, n] \times [1, n]$, the two guard keys $(0, 0)$ and $(n+1, n+1)$ neither dominate nor are dominated by any key in $\key(\textT)$. For any key set $\keys \subseteq \key(\textT)$, set $\{(0, 0), (n+1, n+1)\} \cup \sky(\keys)$ is a skyline -- no key in the set dominates another. We refer to this set as the augmented skyline of $\keys$. All the previous discussions about a skyline naturally apply to the corresponding augmented skyline. 

By Definition~\ref{def:rec}, $\rec(\{(0, 0), (n+1, n+1)\} \cup \sky(\keys)) = \rec(\sky(\keys)) \cup \rec(0,0) \cup \rec (n+1,n+1) = \rec(\sky(\keys))$. This is because $\rec(0,0)=[1,0]\times[0,n]=\phi$ and $\rec(n+1,n+1)=[1,n+1]\times[n+1,n]=\phi$.

Lemma~\ref{lemma:dom} tests the existence of any key in $\keys$ that dominates $(b,c)$. If not, we need to update the augmented skyline to include $(b,c)$ and, at the same time, generate compact windows according to  Lemma~\ref{lemma:dom2} below. Please see Figure~\ref{fig:combined}(b) as an illustration.

\begin{lemma} \label{lemma:dom2}
Given a key set $\keys\subseteq\key(\textT)$ and a key $(b,c)\not\in\keys$. Let $\sky$ be the augmented skyline of $\keys$ in coordinate order. Let
\begin{itemize}
\item $i$ be the largest index such that $\sky[i].y < c$ and 
\item $j$ be the smallest index such that $\sky[j].x > b$.
\end{itemize}
The two guard keys in $\sky$ ensure that $i$ and $j$ exist. We claim \kw{C1}$\&$\kw{C2}. 

\begin{enumerate}
\item[\kw{C1}] If $i+1 \leq j-1$, then the keys in $\sky$ dominated by $(b,c)$ are exclusively $\sky[i+1], \sky[i+2], \cdots, \sky[j-1]$; otherwise, there is no key in $\sky$ dominated by $(b,c)$.

\item[\kw{C2}] If $(b,c)$ is not dominated by any key in $\sky$, then 
\[
\small r_k \doteq
\begin{cases}
 \left[\sky[k].x+1,\,b] \times [c,\,\sky[k+1].y-1\right], & k = i, \\
\left[\sky[k].x+1,\,b] \times [\sky[k].y,\,\sky[k+1].y-1\right], & k \in [i+1,j-1].
\end{cases}
\]
\noindent are $j-i$ mutually disjoint rectangles that jointly cover the growing part of the skyline when including $(b,c)$, i.e.,  $$r_i\cup r_{i+1}\cup \cdots \cup r_{j-1}=\rec(b,c)\setminus\rec(\sky).$$ 
\end{enumerate}
\end{lemma}
\begin{proof}
%\section{Proof of Theorem~\ref{theorem:domed}} \label{sec:proof_of_theorem_ref_theorem_domed}
% \begin{proof}
We prove \kw{C1} in two complementary cases. 

\textbf{Case $i + 1 > j - 1$.} Assume there is a key $\sky[k]$ that is dominated by $(b, c)$. Then $\sky[k].x \leq b <\sky[j].x$ and $\sky[i].y < c\leq \sky[k].y$. By Lemma~\ref{lemma:orderinskyline2}, $i < k$ and $k < j$. Hence $i \leq k - 1 \leq j - 2$, which contradicts $i+1>j-1$. Thus, no key in $\sky$ is dominated by $(b, c)$. 

\textbf{Case $i+1 \leq j-1$.} For any $k \notin [i+1, j-1]$, either $k \leq i$ or $k \geq j$. By the definition of $i$, $\sky[k].y < c$ for $k \leq i$, and by the definition of $j$, $\sky[k].x > b$ for $k \geq j$. For any $k \leq i$ or $k \geq j$, $(b, c)$ does not dominate $\sky[k]$. For any $k\in[i+1,j-1]$, by the definitions of $i$ and $j$, we have $\sky[k].x \leq b$ and $\sky[k].y \geq c$. Since $(b,c)\notin\keys$, we must have $(b,c)\neq \sky[k]$. Thus $[b,c]\subset [\sky[k].x,\sky[k].y]$. Thus, $(b,c)$ dominates $\sky[k]$ for $\forall k\in[i+1,j-1]$. This proves \kw{C1}. %Thus, the keys in $\sky$ dominated by $(b, c)$ are exactly $\sky[i+1], \dots, \sky[j-1]$.
%Therefore, $\sky[i+1], \sky[i+2], \cdots, \sky[j-1]$ are all dominated by $(b,c)$. 
%The proof completes.
%Thus, the keys in $\sky$ dominated by $(b, c)$ are exactly $\sky[i+1], \dots, \sky[j-1]$.

%Given text $\textT$ with and hash function $\hf$. Consider a key $(b,c)\in\keys(\textT)$ with hash value $v$. Let $\keys$ be the subset of keys in $\key(\textT)$ with hash values smaller than $v$ (\ie visited before $(b,c)$ in our algorithm). Clearly, $\{(0,0),(n+1,n+1)\}\cup\sky(\keys)$ is a skyline. This is because for any

% for any $(p,q)\in\keys(\textT)$, $[p,q]\subseteq[1,n]$, while $[1,n]$ is disjoint with $[0,0]$ and $[n+1,n+1]$. Therefore the two guard keys are not dominate by, or dominates, any key in $\keys(\textT)$, which implies $\{(0,0),(n+1,n+1)\}\cup\sky(\keys)$ is a skyline. Let  $\sky = \{\sky[1], \cdots, \sky[l], \sky[l+1]\}$ in coordinate order. 

%Given a key set $\keys\subseteq\key(\textT)$ and a key $(b,c)\not\in\keys$. Let $\sky$ be the augmented skyline of $\keys$ in coordinate order. Let
%\begin{itemize}
%\item $i$ be the largest index such that $\sky[i].y < c$ and 
%\item $j$ be the smallest index such that $\sky[j].x > b$.
%\end{itemize}
%The two guard keys in $\sky$ guarantee that $i$ and $j$ exist. 

We prove \kw{C2} below by considering when $(b,c)$ is not dominated by any key in $\sky$. 
We first show $\sky[i+1]$ exists in the augmented skyline by showing $i<j$. If otherwise, \ie, $i\geq j$, by Lemma~\ref{lemma:orderinskyline2},  $\sky[i].x \geq \sky[j].x >b$, in addition, $\sky[i].y < c$ by definition, thus $\sky[i]$ dominates $(b,c)$, contradiction. $\sky[j]$ exists and thus $\sky[i+1]$ exists.

We then show $r_i, r_{i+1}, \cdots, r_{j-1}$ are mutually disjoint: the $y$-ranges of these rectangles, \ie $[c,\sky[i+1].y-1], [\sky[i+1].y, \sky[i+2].y],\cdots,$ $ [\sky[j-2].y, \sky[j-1].y], [\sky[j-1].y, \sky[j].y]$ are mutually disjoint. 

Next, consider any $(p,q)\in\rec(b,c)\setminus\rec(\sky)$, we show there exists $k\in [i,j-1]$ such that $(p,q)\in r_k$. Since $(p,q)\in\rec(b,c)$ and thus $q \in [c,n]$, $q$ will fall in exactly one of the 3 cases below.
%Furthermore, we define $\rec(0,0)=\rec(n+1,n+1)=\phi$. Then 
\begin{itemize}
    \item[Case 1]: $q\in[c,\sky[i+1].y-1]$,
    \item[Case 2]: $q\in[\sky[k].y,\sky[k+1].y-1]$ for some $k\in[i+1,j-1]$,
    \item[Case 3]: $q\in[\sky[j].y, n]$.
\end{itemize}
\noindent In Case 1, the definition of $i$ ensures that $\sky[i].y<c$. Besides, for any $\sky[u]$ in the augmented skyline, $\sky[u].y-1\leq n$, and thus  $[c,\sky[i+1].y-1] \subseteq [\sky[i].y, n]$. $(p,q)$ is not dominated by any key in $\sky$ including $\sky[i]$ and $q\in[c,\sky[i+1].y-1]$, $p\not\in[1,\sky[i].x]$. As $(p,q) \in \rec(b,c)$, $p\in[1,b]$, and thus $p\in[\sky[i].x+1, b]$. Thus, $(p,q)\in r_i$.

In Case 2, consider $k\in [i+1,j-1]$ such that $q\in[\sky[k].y,\sky[k+1].y-1]$. As $(p,q)\not\in\rec(\sky[k])=[1,\sky[k].x]\times[\sky[k].y,n]$ and $p\in [\sky[k].y,\sky[k+1].y-1] \subseteq [\sky[k].y, n]$, $p\not\in[1,\sky[k].x]$. Together with $p\in[1,b]$, we have $p\in[\sky[k].x+1, b]$. Thus $(p,q)\in r_k$.

Lastly, we show Case 3 is impossible. If $\sky[j]=(n+1,n+1)$, then $[\sky[j].y,n]=\phi$; otherwise, since $(p,q)$ is not dominated by $\sky[j]$ and $q\in [\sky[j].y,n]$, 
%\not\in\rec(\sky[j])=[1,\sky[j].x]\times[\sky[j].y,n]$ and $q\in [\sky[j].y,n]$, 
we have $p\not\in [1,\sky[j].x]$. As $p\in [1,b]$,  $\sky[j].x>b\geq p > \sky[j].x$, contradiction. Thus, Case 3 is impossible.

In conclusion, $\sky[i+1]$ exists and thus the ranges of the three cases are well defined. Case 3 is impossible, thus $q$ can only fall in Case 1 or 2, and in either case, there is $k \in [i,j-1]$ s.t. $(p,q) \in r_k$. 
\end{proof}

Now we prove Theorem~\ref{theorem:uppernlogf}. Firstly, it is trivial to verify \textsc{GenerateKeys} produces $\key(\textT)$. Algorithm~\ref{algo:skyline} examines all keys in $\key(\textT)$ in the increasing order of their hash values. Let $(b,c)$ with hash value $v$ be the current visiting key and $\keys$ be the set of visited keys. We prove by induction that at the end of each for-loop of Algorithm~\ref{algo:skyline}, $\sky$ is the augmented skyline of $\keys\cup\{(b,c)\}$ and the set of compact windows $\pt(b,c)$ generated during the loop constitutes a partition of $\rec(b,c) \setminus \rec(\sky)$. Thus the union of all generated compact windows forms a partition of $\rec(\sky)$. Note that a guard key $(p,q)$ has $\rec(p,q) = \phi$, thus $\rec(\sky) = \rec(\sky(\keys))$.

We assume $\sky$ is the augmented skyline of $\keys$ at the beginning of each for loop. This is true at the beginning of the first loop, when $\keys=\phi$ and $\sky=\{(0,0),(n+1,n+1)\}$. Clearly, $\sky$ is the augmented skyline of $\keys$. Next, we consider two cases for each loop.

In the case there exists a key in $\sky$ dominates $(b,c)$, by Lemma~\ref{lemma:dom}, Line~\ref{algo:skyline:5} evaluates to true and the loop ends\footnote{When using Lemma~\ref{lemma:dom} on $\sky(\keys)$, $j$ does not exist only if Line~\ref{algo:skyline:5} on augmented skyline $\sky$ has $\sky[j']$ taking the guard key $(0,0)$ -- Line~\ref{algo:skyline:5} also tests false; otherwise, the same key is returned, \ie $\sky(\keys)[j]=\sky[j']$. Thus, Line~\ref{algo:skyline:5} correctly facilitates Lemma~\ref{lemma:dom}.}. $\sky$ remains unchanged. The augmented skyline of $\keys\cup\{(b,c)\}$ is $\sky$ as $(b,c)$ is dominated.  $\pt(b,c)=\phi$ is a partition of $\rec(b,c)\setminus\rec(\sky) =  \phi$. %Since $\sky$ is the augmented skyline of $\keys$, it must also be the augmented skyline of $\keys\cup\{(b,c)\}$.

In the case $(b,c)$ is not dominated by any key in $\sky(\keys)$, by Lemma~\ref{lemma:dom}, Line~\ref{algo:skyline:5} evaluates to false. Furthermore, by Lemma~\ref{lemma:dom2}(\kw{C1}), all keys dominated by $(b,c)$ in $\sky$ (if there is any) are precisely $\sky[i+1],\cdots,\sky[j-1]$. By the definition of skyline, $\sky\cup\{(b,c)\}\setminus\{\sky[i+1],\cdots,\sky[j-1]\}$ must be the augmented skyline of $\keys\cup\{(b,c)\}$. Next, we show that $\pt(b,c)$ generated during the loop is a partition of $\rec(b,c)\setminus\rec(\keys)$. 

Consider Lemma~\ref{lemma:dom2} (\kw{C2}) and determine $i$,$j$ and $r_k$ (denoted as $ [a_k,b_k]\times[c_k,d_k]$), $k \in [i,j-1]$ accordingly. Note that $r_k \neq \phi$ only if $a_k \leq b_k$ and $c_k\leq d_k$. The tuples added to $\pt(\textT,\hf)$ in Lines~\ref{algo:skyline:9}-\ref{algo:skyline:13} are \[ \pt(b,c)=\{W_k\doteq\langle T,h,v,a_k,b_k,c_k,d_k\rangle| k \in [i,j-1], r_k\neq \emptyset\}.\]

First, any $W_k\in \pt(b,c)$ is a compact window because by Lemma~\ref{lemma:iffcontain} and Lemma~\ref{lemma:exclude}, all subsequences in $\rec(b,c)\setminus\rec(\sky(\keys))$ have min-hash $v$. Secondly, all the compact windows $W_k$ in $\pt(b,c)$ constitute a partition of $\rec(b,c)\setminus\rec(\sky)$ because by Lemma~\ref{lemma:dom2} (\kw{C2}), the rectangles $\{r_k|k \in [i,j-1]\}$, form a partition of $\rec(b,c)\setminus\rec(\sky)$.

Since along the loop, the increments $\rec(b,c)\setminus\rec(\sky)$ are disjoint, the union of all generated $\pt(b,c)$ forms a partition of $\rec(\sky(\keys))$.

%%%%
The resulting $\pt(\textT,\hf)$ satisfies the coverage condition in Definition~\ref{def:partition}: for every subsequence $\textT[x,y]$ of $\textT$, $(x,y)$ is in $\rec(\sky(\keys))$ and is thus covered by exactly one compact window. Denote by $v$ the hash value of subsequence $\textT[x,y]$. By Lemma~\ref{lemma:iffcontain}, $\textT[x,y]$ must contain a key, say $(b,c)$, with hash value $v$. Thus, $(x,y) \in \rec(b,c) \subseteq \rec(\sky(\keys))$. This completes the proof of Theorem~\ref{theorem:uppernlogf}.
\begin{lemma} \label{lemma:disjointcomplete}
Algorithm~\ref{algo:skyline} with \textnormal{\textsc{GenerateActiveKeys}} produces a partition of $\textT$. 
\end{lemma}
\begin{proof}
It is easy to verify \textsc{GenerateActiveKeys} produces all the active keys $\aset(\textT)$. Based on Lemma~\ref{lemma:activehashingnumber1}, for any non-active key $(b,c)$, it is skipped in our algorithm, which means $\pt(b,c)=\phi$. Based on the proof of Theorem~\ref{theorem:uppernlogf}, Algorithm~\ref{algo:skyline} with \textnormal{\textsc{GenerateActiveKeys}} produces a partitioning of $\textT$. 
\end{proof}

% Based on Theorem~\ref{theorem:domed}(1)-(3), Algorithm~\ref{algo:skyline} generates, for each key $(b,c)$ of value $v$ that has not been dominated by any key in $L$ (the set of keys with hash values smaller than $v$), the compact windows that disjointly covers all the subsequences in $[1,b]\times[c,n]$ with hash value $v$, i.e., the subsequences in $[1,b]\times[c,n] \setminus \bigcup_{(p,q) \in L } ([1,p] \times [q,n])$. Since each sequence's hash value must equal to the hash value of a key in $K$ (generated in Line~1, Algorithm~\ref{algo:skyline}), it must have been covered by a compacted generated. Therefore, the compact wondows are a disjoint partition of all the subsequences. This proves Theorem~\ref{theorem:uppernlogf}.

\section{Lower Bound Analysis} \label{sec:lowerbound}

This section analyzes the lower bound of the partition generation problem. A hard case is when all the tokens are the same, \ie there is only one token $t$ in the text $\textT$. Next we show every  partition $\pt(\textT,\hf)$ has $\Omega(n\log n)$ compact windows in expectation in the worst case. Formally, we define two subsequences are \textit{mergeable} as below.

\begin{definition}[Mergeable]\label{def:merge}
Given text $\textT$ and hash function $\hf$, two subsequences $\textT[i,j]$ and $\textT[p,q]$ are mergeable if they can be represented by the same compact window $\cw$, \ie $i,p\in [a,b]$, $b \leq c$,  $j,q \in [c,d]$, and $\hf(\textT[x,y]) = v, \forall x\in [a,b], y \in [c,d]$.
\end{definition}

Clearly, for a set of subsequences from the same text that are mutually not mergeable, it needs at least one compact window in the partition for each subsequence in the set. Thus we analyze the lower bound of the partition generation problem by constructing a set of subsequences that are mutually not mergeable.

% \begin{example}
% \todo{do we need an example here?} \todo{DO we need a figure here?}
% \end{example}

%Next we show that the hard case where all tokens are the same in a text contains $O(n\log\fmax)$ subsequences that are mutually not mergeable in expectation. Specifically, 

\begin{lemma}\label{lemma:nlogn}
For a text $\textT$ of length $n$ of one token $t$, \ie all tokens in $\textT$ are duplicate, any partition $\pt(\textT,\hf)$ of $\textT$ contains \miao{$\Omega(n + n\log n)$} compact windows in expectation. The expectation is introduced by the randomness of the hash function as opposed to any other assumption.
\end{lemma}
\begin{proof}
 Let $\hf(t,x_1), \hf(t,x_2), \cdots, \hf(t,x_m)$ be all the active hash of the token $t$ where $1\leq x_1 < x_2 < \cdots < x_m \leq n=|\textT|$. That means, $\hf(t,x_j)>\hf(t,x_i)$ for every $1\leq j < i\leq m$. Construct a set of subsequences $S=S_1\cup S_2\cup \cdots \cup S_m$ where $S_i=\{\textT[p,q]\mid q-p+1=x_{i}\}$. That is to say, $S_i$ is the set of all subsequences in $\textT$ containing exactly $x_i$ tokens. Consider two sets $S_i$ and $S_j$ where $1\leq i\neq j\leq m$. First, they must be disjoint as the lengths of their subsequences are different. Furthermore, based on the definition of active hash, the min-hash of subsequences in $S_i$ and $S_j$ are respectively $\hf(t,x_i)$ and $\hf(t,x_j)$. Since $x_i\neq x_j$, the subsequences from $S_i$ and $S_j$ are not mergeable based on Definition~\ref{def:merge}. 

Assume that there are two distinct subsequences $\textT[p_1,q_1]$ and $\textT[p_2,q_2]$ from the same set $S_i$ that are mergeable. As they all of the same length and are distinct, either $p_1<p_2$ or $p_1>p_2$. Without loss of generality, assume $p_1<p_2$. As the two subsequences are mergeable, there exists a compact window \cw representing both subsequences. Thus $p_1<p_2\leq b\leq c\leq q_1 $. The subsequence $\textT[b,c]$ then contains less than $x_i$ tokens because $c-b+1\leq q_1-p_2+1 = (q_1-p_1+1)+p_1-p_2=x_i + p_1-p_2 < x_i$. Thus $\hf(\textT[b,c])>x_i$ since $\hf(t,x_i)$ is active hash, contradicting the definition of a compact window, i.e.,  $\hf(\textT[b,c])=\hf(\textT[p_1,q_1])=\hf(t,x_i)$.

Based on the discussion above, the subsequences in $S$ are mutually not mergeable. Next, we calculate the expected size of $S$. There are $n-x_i+1$ subsequences in $S_i$. Thus $|S|=\sum_{i=1}^m (n-x_{i}+1)$. Let us consider an arbitrary index $j\in[1,n]$. We observe that $j=x_i$ for some $i\in[1,m]$ if and only if $\hf(t,j)$ is the smallest among $\hf(t,1), \hf(t,2), \cdots, \hf(t,j)$. The probability is $\frac{1}{j}$. Thus we have 
\begin{equation}
\begin{split}
    \mathbb{E}[|S|]&=\mathbb{E}[\sum_{i=1}^m (n-x_{i}+1)] = \sum_{i=1}^{n}(n-i+1)\frac{1}{i} \\
& = (n+1)\sum_{i=1}^n \frac{1}{i} -n \geq (n+1)\ln n - n = \Omega(n\log n).
\end{split}
\end{equation}
\miao{Note that when $ n = 1$, $\mathbb{E}[|S|] = 1$, besides, $n \leq n \log n$ when $n \geq 2$. So $\mathbb{E}[|S|] = \Omega(n\log n + n)$. }
% Since the hash values produced by the universal hash function \hf are random and unique, we have $\mathbb{E}[m]=\sum_{k=1}^n 1/k = H_n = O(\log n)$ in expectation. 
% For any $S_i$, there are $O(n)$ subsequences in it, one subsequence $\textT[p,q]$ for each of $1\leq p\leq n-x_{i+1}+1$. In totoal there are $O(n\log n)$ subsequences in $S$. 
The subsequences in $S$ are mutually not mergeable. One compact window is needed to represent each of the subsequence in $S$. Thus any partition $\pt(\textT,\hf)$ contains $\Omega(n\log n + n)$ compact windows in expectation.
\end{proof}
%Please find the proof in Appendix~\ref{sub:proof_of_lemma_nlogn}.

Lemma~\ref{lemma:nlogn} shows a hard case  in terms of partition size. \miao{The reason we keep factor $n$ in the lower bound is for the special case when $n = 1$. This will be useful when we conduct frequency-aware analysis below.} 

%Note that the expectation is introduced by the randomness of the hash function while 

%Based on Lemma~\ref{lemma:uppernlogn}, it is also a worst case.

%Consider the set of active keys induced by the set of active hash values. That is $\{(p,q)\mid q-p+1=x_i, 1\leq i\leq m\}$. Clearly there are $O(n\log n)$ subsequences in the set. Next we show the subsequences in the set are mutually not mergeable. 

% the chances that $\hf(t,i)>\hf(t,j)$ and $\hf(t,i)<\hf(t,j)$ are the same for any $i\neq j$. Thus the number of active hash is $O(n\log n)$ in expectation.

% \stitle{When $\fmax\ll n$.} 

%When the hash value is not random, the worst case happens when the hash values of the token $t$ increase monotonically with the frequency of $t$. In this case, the size of the partition is $\Omega(n^2)$. This is because for any two subsequences $\textT[i,j]$ and $\textT[i,j+1]$ cannot be ``merged'' as $\hf(\textT[i,j])\neq \hf(\textT[i,j+1])$.

% \begin{theorem}
% The Single-Column Algorithm is worst-case optimal in terms of partition size.
% \end{theorem}

\stitle{When $\fmax\ll n$.} Next we consider the cases when we know the length and the maximum token frequency $\fmax$ of a text.

\begin{lemma} \label{lemma:nlogf}

Consider integer $n>0$ and $\fmax \in [1,n]$, there is a text $\textT$ of length $n$ and maximum frquency $\fmax$ such that any partition $\pt(\textT,\hf)$ of $\textT$ contains $\Omega(n + n\log \fmax)$ compact windows in expectation. 
\end{lemma}
\begin{proof}

%\section{Proof of Lemma~\ref{lemma:nlogf}} \label{sub:proof_of_lemma_ref_lemma_nlogf}

Without loss of generality, assume that $n$ is a multiple of $\fmax$. Let $m$ be $\frac{n}{\fmax}$. Let $\textT=t_1 t_1\cdots t_1 t_2 t_2\cdots t_2\cdots t_2 t_m\cdots t_m$ where each $t_i$ (where $1\leq i \leq m$) appears $\fmax$ times next to each other in $\textT$. 

We  construct a set of subsequences $S=S^1\cup S^2\cup \cdots \cup S^m$ where $S^i$ is constructed in the same way as in the proof in Lemma~\ref{lemma:nlogn} for $t_i$. The subsequences within each $S^i$ are mutually not mergeable as shown earlier. Next we show any two subsequences $\textT[p_1,q_1]$ and $\textT[p_2,q_2]$ from $S^i$ and $S^j$ (where $i\neq j)$ are not mergeable. This is obvious as $\hf(\textT[p_1,q_1])=\hf(t_i,x)$ for some $x\in[1,f_i]$ and $\hf(\textT[p_2,q_2])=\hf(t_j,y)$ for some $y\in[1,f_y]$. However, $\hf(t_i,x)\neq\hf(t_j,y)$  assuming the universal hash function $\hf$ has no collision. %Thus $\hf(\textT[p_1,q_1])\neq\hf(\textT[p_2,q_2])$ and the two subsequence are not mergeable.

% =S_1^i\cup S_2^i\cup\cdots S_j^i$ and $S_j^i$
Based on the discussion above, it needs $\Omega(\fmax(1+\log \fmax))$ compact windows to represent all the subsequences $\textT[p,q]$ where $\textT[p] =\text[q] = t_i$ for each $1\leq i\leq m$. Thus in total, it needs $\Omega(\sum_{i\in[m]} \fmax(1+\log \fmax)) = \Omega(n + n\log \fmax)$ compact windows to represent all the subsequences in $\textT$. Thus every partition $\pt(\textT,\hf)$ has $\Theta(n + n\log\fmax)$ compact windows.   
\end{proof}
%Please find the proof in Appendix~\ref{sub:proof_of_lemma_ref_lemma_nlogf}.

%Again, the expectation is introduced by the hash function, not $\textT$.

% \vspace{-1em}
Now we prove Theorem~\ref{theorem:timecomplexity}. Theorem~\ref{theorem:uppernlogf} shows that the expected number of compact windows of Algorithm~\ref{algo:skyline} is $O(n + n\log \fmax)$ while Lemma~\ref{lemma:nlogf} shows a case of $\textT$, given $n$ and $\fmax$, such that no partition can have the number of compact windows smaller than $\Omega(n + n\log \fmax)$ in expectation. Thus, given $n$ and $\fmax$, Algorithm~\ref{algo:skyline} is optimal in the worst case.

\end{document}